\documentclass[submission,copyright,creativecommons]{eptcs}
\providecommand{\event}{MSFP 2014}
\usepackage{breakurl}
\makeatletter
\@ifundefined{lhs2tex.lhs2tex.sty.read}%
  {\@namedef{lhs2tex.lhs2tex.sty.read}{}%
   \newcommand\SkipToFmtEnd{}%
   \newcommand\EndFmtInput{}%
   \long\def\SkipToFmtEnd#1\EndFmtInput{}%
  }\SkipToFmtEnd

\newcommand\ReadOnlyOnce[1]{\@ifundefined{#1}{\@namedef{#1}{}}\SkipToFmtEnd}
\usepackage{amstext}
\usepackage{amssymb}
\usepackage{stmaryrd}
\DeclareFontFamily{OT1}{cmtex}{}
\DeclareFontShape{OT1}{cmtex}{m}{n}
  {<5><6><7><8>cmtex8
   <9>cmtex9
   <10><10.95><12><14.4><17.28><20.74><24.88>cmtex10}{}
\DeclareFontShape{OT1}{cmtex}{m}{it}
  {<-> ssub * cmtt/m/it}{}

\DeclareFontShape{OT1}{cmtt}{bx}{n}
  {<5><6><7><8>cmtt8
   <9>cmbtt9
   <10><10.95><12><14.4><17.28><20.74><24.88>cmbtt10}{}
\DeclareFontShape{OT1}{cmtex}{bx}{n}
  {<-> ssub * cmtt/bx/n}{}

\newcommand{\Conid}[1]{\mathit{#1}}
\newcommand{\Varid}[1]{\mathit{#1}}
\newcommand{\anonymous}{\kern0.06em \vbox{\hrule\@width.5em}}
\newcommand{\plus}{\mathbin{+\!\!\!+}}
\newcommand{\bind}{\mathbin{>\!\!\!>\mkern-6.7mu=}}

\renewcommand{\geq}{\geqslant}
\usepackage{polytable}

\@ifundefined{mathindent}%
  {\newdimen\mathindent\mathindent\leftmargini}%
  {}%

\def\resethooks{%
  \global\let\SaveRestoreHook\empty
  \global\let\ColumnHook\empty}
\newcommand*{\savecolumns}[1][default]%
  {\g@addto@macro\SaveRestoreHook{\savecolumns[#1]}}
\newcommand*{\restorecolumns}[1][default]%
  {\g@addto@macro\SaveRestoreHook{\restorecolumns[#1]}}
\newcommand*{\aligncolumn}[2]%
  {\g@addto@macro\ColumnHook{\column{#1}{#2}}}

\resethooks

\newcommand{\onelinecommentchars}{\quad-{}- }
\newcommand{\commentbeginchars}{\enskip\{-}
\newcommand{\commentendchars}{-\}\enskip}

\newcommand{\visiblecomments}{%
  \let\onelinecomment=\onelinecommentchars
  \let\commentbegin=\commentbeginchars
  \let\commentend=\commentendchars}

\newcommand{\invisiblecomments}{%
  \let\onelinecomment=\empty
  \let\commentbegin=\empty
  \let\commentend=\empty}

\visiblecomments

\newlength{\blanklineskip}
\newcommand{\hsindent}[1]{\quad}
\let\hspre\empty
\let\hspost\empty

\EndFmtInput
\makeatother
\ReadOnlyOnce{polycode.fmt}%
\makeatletter

\newcommand{\hsnewpar}[1]%
  {{\parskip=0pt\parindent=0pt\par\vskip #1\noindent}}

\newcommand{\hscodestyle}{}

\newcommand{\sethscode}[1]%
  {\expandafter\let\expandafter\hscode\csname #1\endcsname
   \expandafter\let\expandafter\endhscode\csname end#1\endcsname}

  {\par\noindent
   \advance\leftskip\mathindent
   \hscodestyle
   \let\\=\@normalcr
   \let\hspre\(\let\hspost\)%
   \pboxed}%
  {\endpboxed\)%
   \par\noindent
   \ignorespacesafterend}

  {\hsnewpar\abovedisplayskip
   \advance\leftskip\mathindent
   \hscodestyle
   \let\hspre\(\let\hspost\)%
   \pboxed}%
  {\endpboxed%
   \hsnewpar\belowdisplayskip
   \ignorespacesafterend}

  {\hsnewpar\abovedisplayskip
   \advance\leftskip\mathindent
   \hscodestyle
   \let\\=\@normalcr
   \(\pboxed}%
  {\endpboxed\)%
   \hsnewpar\belowdisplayskip
   \ignorespacesafterend}

\newcommand{\plainhs}{\sethscode{plainhscode}}

\plainhs

  {\hsnewpar\abovedisplayskip
   \advance\leftskip\mathindent
   \hscodestyle
   \let\\=\@normalcr
   \(\parray}%
  {\endparray\)%
   \hsnewpar\belowdisplayskip
   \ignorespacesafterend}

  {\parray}{\endparray}

  {\(\parray}{\endparray\)}

\def\codeframewidth{\arrayrulewidth}
\RequirePackage{calc}

  {\parskip=\abovedisplayskip\par\noindent
   \hscodestyle
   \arrayrulewidth=\codeframewidth
   \tabular{@{}|p{\linewidth-2\arraycolsep-2\arrayrulewidth-2pt}|@{}}%
   \hline\framedhslinecorrect\\{-1.5ex}%
   \let\endoflinesave=\\
   \let\\=\@normalcr
   \(\pboxed}%
  {\endpboxed\)%
   \framedhslinecorrect\endoflinesave{.5ex}\hline
   \endtabular
   \parskip=\belowdisplayskip\par\noindent
   \ignorespacesafterend}

\newcommand{\framedhslinecorrect}[2]%
  {#1[#2]}

  {\(\def\column##1##2{}%
   \let\>\undefined\let\<\undefined\let\\\undefined
   \newcommand\>[1][]{}\newcommand\<[1][]{}\newcommand\\[1][]{}%
   \def\fromto##1##2##3{##3}%
   }{\) }%

  {\let\orighscode=\hscode
   \let\origendhscode=\endhscode
   \def\endhscode{\def\hscode{\endgroup\def\@currenvir{hscode}\\}\begingroup}
   \orighscode\def\hscode{\endgroup\def\@currenvir{hscode}}}%
  {\origendhscode
   \global\let\hscode=\orighscode
   \global\let\endhscode=\origendhscode}%

\makeatother
\EndFmtInput
\usepackage{amsthm}
\usepackage{amsmath}
\usepackage{amssymb}
\usepackage{hyperref}
\usepackage{cclicenses}
\usepackage{xypic}
\ReadOnlyOnce{forall.fmt}%
\makeatletter

\let\HaskellResetHook\empty
\newcommand*{\AtHaskellReset}[1]{%
  \g@addto@macro\HaskellResetHook{#1}}
\newcommand*{\HaskellReset}{\HaskellResetHook}

\newcommand\hsforall{\global\let\hsdot=\hsperiodonce}
\newcommand*\hsperiodonce[2]{#2\global\let\hsdot=\hscompose}
\newcommand*\hscompose[2]{#1}

\AtHaskellReset{\global\let\hsdot=\hscompose}

\HaskellReset

\makeatother
\EndFmtInput
\ReadOnlyOnce{colorcode.fmt}%

\RequirePackage{colortbl}
\RequirePackage{calc}

\makeatletter
  {\hsnewpar\abovedisplayskip
   \hscodestyle
   \tabular{@{}>{\columncolor{codecolor}}p{\linewidth}@{}}%
   \let\\=\@normalcr
   \(\pboxed}%
  {\endpboxed\)%
   \endtabular
   \hsnewpar\belowdisplayskip
   \ignorespacesafterend}

  {\hsnewpar\abovedisplayskip
   \hscodestyle
   \tabular{@{}>{\columncolor{codecolor}\(}l<{\)}@{}}%
   \pmboxed}%
  {\endpmboxed%
   \endtabular
   \hsnewpar\belowdisplayskip
   \ignorespacesafterend}

  {\hsnewpar\abovedisplayskip
   \hscodestyle
   \arrayrulecolor{codecolor}%
   \arrayrulewidth=\coderulewidth
   \tabular{|p{\linewidth-\arrayrulewidth-\tabcolsep}@{}}%
   \let\\=\@normalcr
   \(\pboxed}%
  {\endpboxed\)%
   \endtabular
   \hsnewpar\belowdisplayskip
   \ignorespacesafterend}
\makeatother

\definecolor{codecolor}{rgb}{1,1,.667}
\newlength{\coderulewidth}
\EndFmtInput

\definecolor{red}{RGB}{128,0,0}
\definecolor{blue}{RGB}{0,0,20}
\definecolor{Gray}{RGB}{80,80,80}
\def\commentbegin{\color{Gray}$\langle$\ }
\def\commentend{$\rangle$}
\renewcommand{\Conid}[1]{{\color{red}\mathtt{#1}}}
\renewcommand{\Varid}[1]{\mathtt{#1}}

\title{Free Applicative Functors}
\author{Paolo Capriotti
\institute{University of Nottingham, United Kingdom}\\
\email{pvc@cs.nott.ac.uk}
\and
Ambrus Kaposi
\institute{University of Nottingham, United Kingdom}\\
\email{auk@cs.nott.ac.uk}}
\hypersetup{
  pdfkeywords={},
  pdfsubject={},
  pdfcreator={Emacs Org-mode version 7.8.11}}

\def\titlerunning{Free Applicative Functors}
\def\authorrunning{P. Capriotti \& A. Kaposi}

\newtheorem{prop}{Proposition}
\newtheorem{lem}{Lemma}
\newtheorem{defn}{Definition}

\newcommand{\A}{\mathcal{A}}
\newcommand{\B}{\mathcal{B}}
\newcommand{\C}{\mathcal{C}}
\newcommand{\D}{\mathcal{D}}
\newcommand{\colim}{\operatornamewithlimits{\mathsf{colim}}}

\begin{document}

\maketitle

\begin{abstract}
Applicative functors \cite{applicative} are a generalisation of monads. Both
allow the expression of effectful computations into an otherwise pure
language, like Haskell \cite{haskell2010}.
Applicative functors are to be preferred to monads when the structure of a
computation is fixed \emph{a priori}.  That makes it possible to perform certain
kinds of static analysis on applicative values.
We define a notion of \emph{free applicative functor}, prove that it satisfies
the appropriate laws, and that the construction is left adjoint to a suitable
forgetful functor.
We show how free applicative functors can be used to implement embedded DSLs
which can be statically analysed.
\end{abstract}

%
%
%
%

\section{Introduction}
\label{sec_intro}

\emph{Free} monads in Haskell are a very well-known and practically used
construction.  Given any
endofunctor \ensuremath{\Varid{f}}, the free monad on \ensuremath{\Varid{f}} is given by a simple inductive
definition:

\begin{hscode}\SaveRestoreHook
\column{B}{@{}>{\hspre}l<{\hspost}@{}}%
\column{3}{@{}>{\hspre}c<{\hspost}@{}}%
\column{3E}{@{}l@{}}%
\column{6}{@{}>{\hspre}l<{\hspost}@{}}%
\column{E}{@{}>{\hspre}l<{\hspost}@{}}%
\>[B]{}\texttt{\color{blue}\textbf{data}}\;\Conid{Free}\;\Varid{f}\;\Varid{a}{}\<[E]%
\\
\>[B]{}\hsindent{3}{}\<[3]%
\>[3]{}\mathrel{=}{}\<[3E]%
\>[6]{}\Conid{Return}\;\Varid{a}{}\<[E]%
\\
\>[B]{}\hsindent{3}{}\<[3]%
\>[3]{}\mid {}\<[3E]%
\>[6]{}\Conid{Free}\;\texttt{(}\Varid{f}\;\texttt{(}\Conid{Free}\;\Varid{f}\;\Varid{a}\texttt{)}\texttt{)}{}\<[E]%
\ColumnHook
\end{hscode}\resethooks

The typical use case for this construction is creating embedded DSLs (see for
example \cite{data_types_a_la_carte}, where \ensuremath{\Conid{Free}} is called \ensuremath{\Conid{Term}}).  In this
context, the functor \ensuremath{\Varid{f}} is usually obtained as the coproduct of a number of
functors representing ``basic operations'', and the resulting DSL is the minimal
embedded language including those operations.

One problem of the free monad approach is that programs written in a monadic DSL
are not amenable to static analysis.  It is impossible to examine the structure
of a monadic computation without executing it.

In this paper, we show how a similar ``free construction'' can be realised in
the context of applicative functors.  In particular, we make the following
contributions:

\begin{itemize}
\item We give two definitions of \emph{free applicative functor} in Haskell
(section \ref{sec:def}), and show that they are equivalent (section \ref{sec:isomorphism}).
\item We prove that our definition is correct, in the sense that it really is an
applicative functor (section \ref{sec:applicative}), and that it is ``free'' in a precise sense
(section \ref{sec:adjoint}).
\item We present a number of examples where the use of free applicative functors
helps make the code more elegant, removes duplication or enables certain kinds
of optimizations which are not possible when using free monads. We describe the
differences between expressivity of DSLs using free applicatives and free monads
(section \ref{sec:applications}).
\item We compare our definition to other existing implementations of the same
idea (section \ref{sec:related}).
\end{itemize}

Applicative functors can be regarded as monoids in the category of endofunctors
with Day convolution (see for instance \cite{day_thesis}, example 3.2.2).  There
exists a general theory for constructing free monoids in monoidal categories
\cite{kelly_transfinite}, but in this paper we aim to describe the special
case of applicative functors using a formalism that is accessible to an audience
of Haskell programmers.

Familiarity with applicative functors is not required, although it is helpful to
understand the motivation behind this work.  We make use of category theoretical
concepts to justify our definition, but the Haskell code we present can also
stand on its own.

The proofs in this paper are carried out using equational reasoning in an
informally defined total subset of Haskell.  In sections \ref{sec:totality} and
\ref{sec:semantics} we will show how to interpret all our definitions and proofs
in a general (locally presentable) cartesian closed category, such as the
category of sets.

\subsection{Applicative functors}

\emph{Applicative functors} (also called \emph{idioms}) were first introduced in
\cite{applicative} as a generalisation of monads that provides a lighter
notation for expressing monadic computations in an applicative style.

They have since been used in a variety of different applications, including
efficient parsing (see section \ref{example:applicative_parsers}), regular
expressions and bidirectional routing.

Applicative functors are defined by the following type class:

\begin{hscode}\SaveRestoreHook
\column{B}{@{}>{\hspre}l<{\hspost}@{}}%
\column{3}{@{}>{\hspre}l<{\hspost}@{}}%
\column{E}{@{}>{\hspre}l<{\hspost}@{}}%
\>[B]{}\mathbf{class}\;\Conid{Functor}\;\Varid{f}\Rightarrow \Conid{Applicative}\;\Varid{f}\;\texttt{\color{blue}\textbf{where}}{}\<[E]%
\\
\>[B]{}\hsindent{3}{}\<[3]%
\>[3]{}\Varid{pure}\mathbin{::}\Varid{a}\to \Varid{f}\;\Varid{a}{}\<[E]%
\\
\>[B]{}\hsindent{3}{}\<[3]%
\>[3]{}\texttt{(}\mathbin{\texttt{<*>}}\texttt{)}\mathbin{::}\Varid{f}\;\texttt{(}\Varid{a}\to \Varid{b}\texttt{)}\to \Varid{f}\;\Varid{a}\to \Varid{f}\;\Varid{b}{}\<[E]%
\ColumnHook
\end{hscode}\resethooks

The idea is that a value of type \ensuremath{\Varid{f}\;\Varid{a}} represents an ``effectful'' computation
returning a result of type \ensuremath{\Varid{a}}.  The \ensuremath{\Varid{pure}} method creates a trivial computation
without any effect, and \ensuremath{\texttt{(}\mathbin{\texttt{<*>}}\texttt{)}} allows two computations to be sequenced, by
applying a function returned by the first, to the value returned by the second.

Since every monad can be made into an applicative functor in a
canonical way,\footnote{To be precise, there are two canonical ways to
turn a monad into an applicative functor, with opposite orderings of
effects.}  the abundance of monads in the practice of Haskell
programming naturally results in a significant number of practically
useful applicative functors.

Applicatives not arising from monads, however, are not as widespread, probably
because, although it is relatively easy to combine existing applicatives (see
for example \cite{constructing_applicatives}), techniques to construct new ones
have not been thoroughly explored so far.

In this paper we are going to define an applicative functor \ensuremath{\Conid{FreeA}\;\Varid{f}} for any
Haskell functor \ensuremath{\Varid{f}}, thus providing a systematic way to create new applicatives,
which can be used for a variety of applications.

The meaning of \ensuremath{\Conid{FreeA}\;\Varid{f}} will be clarified in section \ref{sec:adjoint}, but for
the sake of the following examples, \ensuremath{\Conid{FreeA}\;\Varid{f}} can be thought of as the
``simplest'' applicative functor which can be built using \ensuremath{\Varid{f}}.

\subsection{Example: option parsers}\label{example:option_intro}

To illustrate how the free applicative construction can be used in practice, we
take as a running example a parser for options of a command-line tool.

For simplicity, we will limit ourselves to an interface which can only accept
options that take a single argument.  We will use a double dash as a prefix for
the option name.

For example, a tool to create a new user in a Unix system could be used as
follows:

\begin{tabbing}\tt
~create\char95{}user~\char45{}\char45{}username~john~\char92{}\\
\tt ~~~~~~~~~~~~~\char45{}\char45{}fullname~\char34{}John~Doe\char34{}~\char92{}\\
\tt ~~~~~~~~~~~~~\char45{}\char45{}id~1002
\end{tabbing}

Our parser could be run over the argument list and it would return a record of
the following type:

\begin{hscode}\SaveRestoreHook
\column{B}{@{}>{\hspre}l<{\hspost}@{}}%
\column{3}{@{}>{\hspre}l<{\hspost}@{}}%
\column{6}{@{}>{\hspre}l<{\hspost}@{}}%
\column{E}{@{}>{\hspre}l<{\hspost}@{}}%
\>[B]{}\texttt{\color{blue}\textbf{data}}\;\Conid{User}\mathrel{=}\Conid{User}{}\<[E]%
\\
\>[B]{}\hsindent{3}{}\<[3]%
\>[3]{}\{\mskip1.5mu {}\<[6]%
\>[6]{}\Varid{username}\mathbin{::}\Conid{String}{}\<[E]%
\\
\>[B]{}\hsindent{3}{}\<[3]%
\>[3]{}\mathbin{\texttt{,}}{}\<[6]%
\>[6]{}\Varid{fullname}\mathbin{::}\Conid{String}{}\<[E]%
\\
\>[B]{}\hsindent{3}{}\<[3]%
\>[3]{}\mathbin{\texttt{,}}{}\<[6]%
\>[6]{}\Varid{id}\mathbin{::}\Conid{Int}\mskip1.5mu\}{}\<[E]%
\\
\>[B]{}\hsindent{3}{}\<[3]%
\>[3]{}\texttt{\color{blue}\textbf{deriving}}\;\Conid{Show}{}\<[E]%
\ColumnHook
\end{hscode}\resethooks

Furthermore, given a parser, it should be possible to automatically produce a
summary of all the options that it supports, to be presented to the user of the
tool as documentation.

We can define a data structure representing a parser for an individual option,
with a specified type, as a functor:

\begin{hscode}\SaveRestoreHook
\column{B}{@{}>{\hspre}l<{\hspost}@{}}%
\column{3}{@{}>{\hspre}l<{\hspost}@{}}%
\column{6}{@{}>{\hspre}l<{\hspost}@{}}%
\column{E}{@{}>{\hspre}l<{\hspost}@{}}%
\>[B]{}\texttt{\color{blue}\textbf{data}}\;\Conid{Option}\;\Varid{a}\mathrel{=}\Conid{Option}{}\<[E]%
\\
\>[B]{}\hsindent{3}{}\<[3]%
\>[3]{}\{\mskip1.5mu {}\<[6]%
\>[6]{}\Varid{optName}\mathbin{::}\Conid{String}{}\<[E]%
\\
\>[B]{}\hsindent{3}{}\<[3]%
\>[3]{}\mathbin{\texttt{,}}{}\<[6]%
\>[6]{}\Varid{optDefault}\mathbin{::}\Conid{Maybe}\;\Varid{a}{}\<[E]%
\\
\>[B]{}\hsindent{3}{}\<[3]%
\>[3]{}\mathbin{\texttt{,}}{}\<[6]%
\>[6]{}\Varid{optReader}\mathbin{::}\Conid{String}\to \Conid{Maybe}\;\Varid{a}\mskip1.5mu\}{}\<[E]%
\\
\>[B]{}\hsindent{3}{}\<[3]%
\>[3]{}\texttt{\color{blue}\textbf{deriving}}\;\Conid{Functor}{}\<[E]%
\ColumnHook
\end{hscode}\resethooks

We now want to create a DSL based on the \ensuremath{\Conid{Option}} functor, which would allow us
to combine options for different types into a single value representing the full
parser.  As stated in the introduction, a common way to create a DSL from a
functor is to use free monads.

However, taking the free monad over the \ensuremath{\Conid{Option}} functor would not be very
useful here.  First of all, sequencing of options should be \emph{independent}:
later options should not depend on the value parsed by previous ones.  Secondly,
monads cannot be inspected without running them, so there is no way to obtain a
summary of all options of a parser automatically.

What we really need is a way to construct a parser DSL in such a way that the
values returned by the individual options can be combined using an \ensuremath{\Conid{Applicative}}
interface.  And that is exactly what \ensuremath{\Conid{FreeA}} will provide.

Thus, if we use \ensuremath{\Conid{FreeA}\;\Conid{Option}\;\Varid{a}} as our embedded DSL, we can interpret it as the
type of a parser with an unspecified number of options, of possibly different
types.  When run, those options would be matched against the input command line,
in an arbitrary order, and the resulting values will be eventually combined to
obtain a final result of type \ensuremath{\Varid{a}}.

In our specific example, an expression to specify the command line option parser
for \text{\tt create\char95{}user} would look like this:

\begin{hscode}\SaveRestoreHook
\column{B}{@{}>{\hspre}l<{\hspost}@{}}%
\column{3}{@{}>{\hspre}l<{\hspost}@{}}%
\column{E}{@{}>{\hspre}l<{\hspost}@{}}%
\>[B]{}\Varid{userP}\mathbin{::}\Conid{FreeA}\;\Conid{Option}\;\Conid{User}{}\<[E]%
\\
\>[B]{}\Varid{userP}\mathrel{=}\Conid{User}{}\<[E]%
\\
\>[B]{}\hsindent{3}{}\<[3]%
\>[3]{}\mathbin{\texttt{<\$>}}\Varid{one}\;\texttt{(}\Conid{Option}\;\text{\tt \char34 username\char34}\;\Conid{Nothing}\;\Conid{Just}\texttt{)}{}\<[E]%
\\
\>[B]{}\hsindent{3}{}\<[3]%
\>[3]{}\mathbin{\texttt{<*>}}\Varid{one}\;\texttt{(}\Conid{Option}\;\text{\tt \char34 fullname\char34}\;\texttt{(}\Conid{Just}\;\text{\tt \char34 \char34}\texttt{)}\;\Conid{Just}\texttt{)}{}\<[E]%
\\
\>[B]{}\hsindent{3}{}\<[3]%
\>[3]{}\mathbin{\texttt{<*>}}\Varid{one}\;\texttt{(}\Conid{Option}\;\text{\tt \char34 id\char34}\;\Conid{Nothing}\;\Varid{readInt}\texttt{)}{}\<[E]%
\\[\blanklineskip]%
\>[B]{}\Varid{readInt}\mathbin{::}\Conid{String}\to \Conid{Maybe}\;\Conid{Int}{}\<[E]%
\ColumnHook
\end{hscode}\resethooks

where we need a ``generic smart constructor'':

\begin{hscode}\SaveRestoreHook
\column{B}{@{}>{\hspre}l<{\hspost}@{}}%
\column{E}{@{}>{\hspre}l<{\hspost}@{}}%
\>[B]{}\Varid{one}\mathbin{::}\Conid{Option}\;\Varid{a}\to \Conid{FreeA}\;\Conid{Option}\;\Varid{a}{}\<[E]%
\ColumnHook
\end{hscode}\resethooks

which lifts an option to a parser.

\subsection{Example: web service client}\label{example:web_service_intro}

One of the applications of free monads, exemplified in
\cite{data_types_a_la_carte}, is the definition of special-purpose monads,
allowing to express computations which make use of a limited and well-defined
subset of IO operations.

Given the following functor:

\begin{hscode}\SaveRestoreHook
\column{B}{@{}>{\hspre}l<{\hspost}@{}}%
\column{3}{@{}>{\hspre}l<{\hspost}@{}}%
\column{6}{@{}>{\hspre}l<{\hspost}@{}}%
\column{12}{@{}>{\hspre}l<{\hspost}@{}}%
\column{E}{@{}>{\hspre}l<{\hspost}@{}}%
\>[B]{}\texttt{\color{blue}\textbf{data}}\;\Conid{WebService}\;\Varid{a}\mathrel{=}{}\<[E]%
\\
\>[B]{}\hsindent{6}{}\<[6]%
\>[6]{}\Conid{GET}\;{}\<[12]%
\>[12]{}\{\mskip1.5mu \Varid{url}\mathbin{::}\Conid{URL}\mathbin{\texttt{,}}\Varid{params}\mathbin{::}[\mskip1.5mu \Conid{String}\mskip1.5mu]\mathbin{\texttt{,}}\Varid{result}\mathbin{::}\texttt{(}\Conid{String}\to \Varid{a}\texttt{)}\mskip1.5mu\}{}\<[E]%
\\
\>[B]{}\hsindent{3}{}\<[3]%
\>[3]{}\mid {}\<[6]%
\>[6]{}\Conid{POST}\;{}\<[12]%
\>[12]{}\{\mskip1.5mu \Varid{url}\mathbin{::}\Conid{URL}\mathbin{\texttt{,}}\Varid{params}\mathbin{::}[\mskip1.5mu \Conid{String}\mskip1.5mu]\mathbin{\texttt{,}}\Varid{body}\mathbin{::}\Conid{String}\mathbin{\texttt{,}}\Varid{cont}\mathbin{::}\Varid{a}\mskip1.5mu\}{}\<[E]%
\\
\>[B]{}\hsindent{3}{}\<[3]%
\>[3]{}\texttt{\color{blue}\textbf{deriving}}\;\Conid{Functor}{}\<[E]%
\ColumnHook
\end{hscode}\resethooks

the free monad on \ensuremath{\Conid{WebService}} allows the definition of an
application interacting with a web service with the same convenience as the \text{\tt IO}
monad, once ``smart constructors'' are defined for the two basic operations
of getting and posting:

\begin{hscode}\SaveRestoreHook
\column{B}{@{}>{\hspre}l<{\hspost}@{}}%
\column{E}{@{}>{\hspre}l<{\hspost}@{}}%
\>[B]{}\Varid{get}\mathbin{::}\Conid{URL}\to [\mskip1.5mu \Conid{String}\mskip1.5mu]\to \Conid{Free}\;\Conid{WebService}\;\Conid{String}{}\<[E]%
\\
\>[B]{}\Varid{get}\;\Varid{url}\;\Varid{params}\mathrel{=}\Conid{Free}\;\texttt{(}\Conid{GET}\;\Varid{url}\;\Varid{params}\;\Conid{Return}\texttt{)}{}\<[E]%
\\[\blanklineskip]%
\>[B]{}\Varid{post}\mathbin{::}\Conid{URL}\to [\mskip1.5mu \Conid{String}\mskip1.5mu]\to \Conid{String}\to \Conid{Free}\;\Conid{WebService}\;\texttt{(}\texttt{)}{}\<[E]%
\\
\>[B]{}\Varid{post}\;\Varid{url}\;\Varid{params}\;\Varid{body}\mathrel{=}\Conid{Free}\;\texttt{(}\Conid{POST}\;\Varid{url}\;\Varid{params}\;\Varid{body}\;\texttt{(}\Conid{Return}\;\texttt{(}\texttt{)}\texttt{)}\texttt{)}{}\<[E]%
\ColumnHook
\end{hscode}\resethooks

For example, one can implement an operation which copies data from one
server to another as follows:

\begin{hscode}\SaveRestoreHook
\column{B}{@{}>{\hspre}l<{\hspost}@{}}%
\column{E}{@{}>{\hspre}l<{\hspost}@{}}%
\>[B]{}\Varid{copy}\mathbin{::}\Conid{URL}\to [\mskip1.5mu \Conid{String}\mskip1.5mu]\to \Conid{URL}\to [\mskip1.5mu \Conid{String}\mskip1.5mu]\to \Conid{Free}\;\Conid{WebService}\;\texttt{(}\texttt{)}{}\<[E]%
\\
\>[B]{}\Varid{copy}\;\Varid{srcURL}\;\Varid{srcPars}\;\Varid{dstURL}\;\Varid{dstPars}\mathrel{=}\Varid{get}\;\Varid{srcURL}\;\Varid{srcPars}\bind \Varid{post}\;\Varid{dstURL}\;\Varid{dstPars}{}\<[E]%
\ColumnHook
\end{hscode}\resethooks

For some applications, we might need to have more control over the operations
that are going to be executed when we eventually run the embedded program
contained in a value of type \ensuremath{\Conid{Free}\;\Conid{WebService}\;\Varid{a}}.

For example, a web service client application executing a large number of \text{\tt GET}
and \text{\tt POST} operations might want to rate limit the number of requests to a
particular server by putting delays between them, and, on the other hand,
parallelise requests to different servers. Another useful feature would be to
estimate the time it would take to execute an embedded Web Service application.

However, there is no way to achieve that using the free monad approach.  In
fact, it is not even possible to define a function like:

\begin{hscode}\SaveRestoreHook
\column{B}{@{}>{\hspre}l<{\hspost}@{}}%
\column{E}{@{}>{\hspre}l<{\hspost}@{}}%
\>[B]{}\Varid{count}\mathbin{::}\Conid{Free}\;\Conid{WebService}\;\Varid{a}\to \Conid{Int}{}\<[E]%
\ColumnHook
\end{hscode}\resethooks

which returns the total number of \text{\tt GET}/\text{\tt POST} operations performed by a value
of type \ensuremath{\Conid{Free}\;\Conid{WebService}\;\Varid{a}}.

To see why, consider the following example, which updates the email field in all
the blog posts on a particular website:

\begin{hscode}\SaveRestoreHook
\column{B}{@{}>{\hspre}l<{\hspost}@{}}%
\column{3}{@{}>{\hspre}l<{\hspost}@{}}%
\column{5}{@{}>{\hspre}l<{\hspost}@{}}%
\column{E}{@{}>{\hspre}l<{\hspost}@{}}%
\>[B]{}\Varid{updateEmails}\mathbin{::}\Conid{String}\to \Conid{Free}\;\Conid{WebService}\;\texttt{(}\texttt{)}{}\<[E]%
\\
\>[B]{}\Varid{updateEmails}\;\Varid{newEmail}\mathrel{=}\texttt{\color{blue}\textbf{do}}{}\<[E]%
\\
\>[B]{}\hsindent{3}{}\<[3]%
\>[3]{}\Varid{entryURLs}\leftarrow \Varid{get}\;\text{\tt \char34 myblog.com\char34}\;[\mskip1.5mu \text{\tt \char34 list\char95 entries\char34}\mskip1.5mu]{}\<[E]%
\\
\>[B]{}\hsindent{3}{}\<[3]%
\>[3]{}\Varid{forM\char95 }\;\texttt{(}\Varid{words}\;\Varid{entryURLs}\texttt{)}\mathbin{\$}\lambda \Varid{entryURL}\to {}\<[E]%
\\
\>[3]{}\hsindent{2}{}\<[5]%
\>[5]{}\Varid{post}\;\Varid{entryURL}\;[\mskip1.5mu \text{\tt \char34 updateEmail\char34}\mskip1.5mu]\;\Varid{newEmail}{}\<[E]%
\ColumnHook
\end{hscode}\resethooks

Now, the number of \text{\tt POST} operations performed by \text{\tt updateEmails} is
the same as the number of blog posts on \text{\tt myblog\char46{}com} which cannot be
determined by a pure function like \text{\tt count}.

The \ensuremath{\Conid{FreeA}} construction, presented in this paper, represents a general solution
for the problem of constructing embedded languages that allow the definition of
functions performing static analysis on embedded programs, of which
\ensuremath{\Varid{count}\mathbin{::}\Conid{FreeA}\;\Conid{WebService}\;\Varid{a}\to \Conid{Int}} is a very simple example.

\subsection{Example: applicative parsers}\label{example:applicative_parsers}

The idea that monads are ``too flexible'' has also been explored, again in the
context of parsing, by Swierstra and Duponcheel \cite{arrow_parsers}, who
showed how to improve both performance and error-reporting capabilities of an
embedded language for grammars by giving up some of the expressivity of monads.

The basic principle is that, by weakening the monadic interface to that of an
applicative functor (or, more precisely, an \emph{alternative} functor), it
becomes possible to perform enough static analysis to compute first sets for
productions.

The approach followed in \cite{arrow_parsers} is ad-hoc: an applicative functor
is defined, which keeps track of first sets, and whether a parser accepts the
empty string.  This is combined with a traditional monadic parser, regarded as
an applicative functor, using a generalised semi-direct product, as described in
\cite{constructing_applicatives}.

The question, then, is whether it is possible to express this construction in a
general form, in such a way that, given a functor representing a notion of
``parser'' for an individual symbol in the input stream, applying the
construction one would automatically get an Applicative functor, allowing such
elementary parsers to be sequenced.

Free applicative functors can be used to that end. We start with a functor \ensuremath{\Varid{f}},
such that \ensuremath{\Varid{f}\;\Varid{a}} describes an elementary parser for individual elements of the
input, returning values of type \ensuremath{\Varid{a}}.  \ensuremath{\Conid{FreeA}\;\Varid{f}\;\Varid{a}} is then a parser which can be
used on the full input, and combines all the outputs of the individual parsers
out of which it is built, yielding a result of type \ensuremath{\Varid{a}}.

Unfortunately, applying this technique directly results in a strictly less
expressive solution.  In fact, since \ensuremath{\Conid{FreeA}\;\Varid{f}} is the simplest applicative over
\ensuremath{\Varid{f}}, it is necessarily \emph{just} an applicative, i.e. it cannot also have an
\ensuremath{\Conid{Alternative}} instance, which in this case is essential.

The \ensuremath{\Conid{Alternative}} type class is defined as follows:

\begin{hscode}\SaveRestoreHook
\column{B}{@{}>{\hspre}l<{\hspost}@{}}%
\column{3}{@{}>{\hspre}l<{\hspost}@{}}%
\column{E}{@{}>{\hspre}l<{\hspost}@{}}%
\>[B]{}\mathbf{class}\;\Conid{Applicative}\;\Varid{f}\Rightarrow \Conid{Alternative}\;\Varid{f}\;\texttt{\color{blue}\textbf{where}}{}\<[E]%
\\
\>[B]{}\hsindent{3}{}\<[3]%
\>[3]{}\Varid{empty}\mathbin{::}\Varid{f}\;\Varid{a}{}\<[E]%
\\
\>[B]{}\hsindent{3}{}\<[3]%
\>[3]{}\texttt{(}\mathbin{\texttt{<|>}}\texttt{)}\mathbin{::}\Varid{f}\;\Varid{a}\to \Varid{f}\;\Varid{a}\to \Varid{f}\;\Varid{a}{}\<[E]%
\ColumnHook
\end{hscode}\resethooks

An Alternative instance gives an applicative functor the structure of a monoid,
with \ensuremath{\Varid{empty}} as the unit element, and \ensuremath{\mathbin{\texttt{<|>}}} as the binary operation. In the
case of parsers, \ensuremath{\Varid{empty}} matches no input string, while \ensuremath{\mathbin{\texttt{<|>}}} is a choice
operator between two parsers.

We discuss the issue of \ensuremath{\Conid{Alternative}} in more detail in section
\ref{sec:discussion}.

\section{Definition of free applicative functors}\label{sec:def}

To obtain a suitable definition for the free applicative functor generated by a
functor \ensuremath{\Varid{f}}, we first pause to reflect on how one could naturally arrive at the
definition of the \ensuremath{\Conid{Applicative}} class via an obvious generalisation of the
notion of functor.

Given a functor \ensuremath{\Varid{f}}, the \ensuremath{\Varid{fmap}} method gives us a way to lift \emph{unary} pure
functions \ensuremath{\Varid{a}\to \Varid{b}} to effectful functions \ensuremath{\Varid{f}\;\Varid{a}\to \Varid{f}\;\Varid{b}}, but what about
functions of arbitrary arity?

For example, given a value of type \ensuremath{\Varid{a}}, we can regard it as a nullary pure
function, which we might want to lift to a value of type \ensuremath{\Varid{f}\;\Varid{a}}.

Similarly, given a binary function \ensuremath{\Varid{h}\mathbin{::}\Varid{a}\to \Varid{b}\to \Varid{c}}, it is quite reasonable to
ask for a lifting of \ensuremath{\Varid{h}} to something of type \ensuremath{\Varid{f}\;\Varid{a}\to \Varid{f}\;\Varid{b}\to \Varid{f}\;\Varid{c}}.

The \ensuremath{\Conid{Functor}} instance alone cannot provide either of such liftings, nor any of
the higher-arity liftings which we could define.

It is therefore natural to define a type class for generalised functors, able to
lift functions of arbitrary arity:

\begin{hscode}\SaveRestoreHook
\column{B}{@{}>{\hspre}l<{\hspost}@{}}%
\column{3}{@{}>{\hspre}l<{\hspost}@{}}%
\column{E}{@{}>{\hspre}l<{\hspost}@{}}%
\>[B]{}\mathbf{class}\;\Conid{Functor}\;\Varid{f}\Rightarrow \Conid{MultiFunctor}\;\Varid{f}\;\texttt{\color{blue}\textbf{where}}{}\<[E]%
\\
\>[B]{}\hsindent{3}{}\<[3]%
\>[3]{}\mathtt{fmap}_\mathtt{0}\mathbin{::}\Varid{a}\to \Varid{f}\;\Varid{a}{}\<[E]%
\\[\blanklineskip]%
\>[B]{}\hsindent{3}{}\<[3]%
\>[3]{}\mathtt{fmap}_\mathtt{1}\mathbin{::}\texttt{(}\Varid{a}\to \Varid{b}\texttt{)}\to \Varid{f}\;\Varid{a}\to \Varid{f}\;\Varid{b}{}\<[E]%
\\
\>[B]{}\hsindent{3}{}\<[3]%
\>[3]{}\mathtt{fmap}_\mathtt{1}\mathrel{=}\Varid{fmap}{}\<[E]%
\\[\blanklineskip]%
\>[B]{}\hsindent{3}{}\<[3]%
\>[3]{}\mathtt{fmap}_\mathtt{2}\mathbin{::}\texttt{(}\Varid{a}\to \Varid{b}\to \Varid{c}\texttt{)}\to \Varid{f}\;\Varid{a}\to \Varid{f}\;\Varid{b}\to \Varid{f}\;\Varid{c}{}\<[E]%
\ColumnHook
\end{hscode}\resethooks

It is easy to see that a higher-arity \ensuremath{\mathtt{fmap}_\mathtt{n}} can now be defined in terms of
\ensuremath{\mathtt{fmap}_\mathtt{2}}. For example, for $\ensuremath{\Varid{n}} = 3$:

\begin{hscode}\SaveRestoreHook
\column{B}{@{}>{\hspre}l<{\hspost}@{}}%
\column{8}{@{}>{\hspre}l<{\hspost}@{}}%
\column{E}{@{}>{\hspre}l<{\hspost}@{}}%
\>[B]{}\mathtt{fmap}_\mathtt{3}{}\<[8]%
\>[8]{}\mathbin{::}\Conid{MultiFunctor}\;\Varid{f}{}\<[E]%
\\
\>[8]{}\Rightarrow \texttt{(}\Varid{a}\to \Varid{b}\to \Varid{c}\to \Varid{d}\texttt{)}{}\<[E]%
\\
\>[8]{}\to \Varid{f}\;\Varid{a}\to \Varid{f}\;\Varid{b}\to \Varid{f}\;\Varid{c}\to \Varid{f}\;\Varid{d}{}\<[E]%
\\
\>[B]{}\mathtt{fmap}_\mathtt{3}\;\Varid{h}\;\Varid{x}\;\Varid{y}\;\Varid{z}\mathrel{=}\mathtt{fmap}_\mathtt{2}\;\texttt{(}\mathbin{\texttt{\$}}\texttt{)}\;\texttt{(}\mathtt{fmap}_\mathtt{2}\;\Varid{h}\;\Varid{x}\;\Varid{y}\texttt{)}\;\Varid{z}{}\<[E]%
\ColumnHook
\end{hscode}\resethooks

However, before trying to think of what the laws for such a type class ought to
be, we can observe that \ensuremath{\Conid{MultiFunctor}} is actually none other than \ensuremath{\Conid{Applicative}}
in disguise.

In fact, \ensuremath{\mathtt{fmap}_\mathtt{0}} has exactly the same type as \ensuremath{\Varid{pure}}, and we can easily convert
\ensuremath{\mathtt{fmap}_\mathtt{2}} to \ensuremath{\texttt{(}\mathbin{\texttt{<*>}}\texttt{)}} and vice versa:

\begin{hscode}\SaveRestoreHook
\column{B}{@{}>{\hspre}l<{\hspost}@{}}%
\column{E}{@{}>{\hspre}l<{\hspost}@{}}%
\>[B]{}\Varid{g}\mathbin{\texttt{<*>}}\Varid{x}\mathrel{=}\mathtt{fmap}_\mathtt{2}\;\texttt{(}\mathbin{\texttt{\$}}\texttt{)}\;\Varid{g}\;\Varid{x}{}\<[E]%
\\
\>[B]{}\mathtt{fmap}_\mathtt{2}\;\Varid{h}\;\Varid{x}\;\Varid{y}\mathrel{=}\Varid{fmap}\;\Varid{h}\;\Varid{x}\mathbin{\texttt{<*>}}\Varid{y}{}\<[E]%
\ColumnHook
\end{hscode}\resethooks

The difference between \ensuremath{\texttt{(}\mathbin{\texttt{<*>}}\texttt{)}} and \ensuremath{\mathtt{fmap}_\mathtt{2}} is that \ensuremath{\texttt{(}\mathbin{\texttt{<*>}}\texttt{)}} expects the first two
arguments of \ensuremath{\mathtt{fmap}_\mathtt{2}}, of types \ensuremath{\Varid{a}\to \Varid{b}\to \Varid{c}} and \ensuremath{\Varid{f}\;\Varid{a}} respectively, to be
combined in a single argument of type \ensuremath{\Varid{f}\;\texttt{(}\Varid{b}\to \Varid{c}\texttt{)}}.

This can always be done with a single use of \ensuremath{\Varid{fmap}}, so, if we assume that \ensuremath{\Varid{f}}
is a functor, \ensuremath{\texttt{(}\mathbin{\texttt{<*>}}\texttt{)}} and \ensuremath{\mathtt{fmap}_\mathtt{2}} are effectively equivalent.

Nevertheless, this roundabout way of arriving to the definition of \ensuremath{\Conid{Applicative}}
shows that an applicative functor is just a functor that knows how to lift
functions of arbitrary arities. An overloaded notation to express the
application of $\ensuremath{\Varid{fmap}}_i$ for all $i$ is defined in \cite{applicative}, where it
is referred to as \emph{idiom brackets}.

Given a pure function of arbitrary arity and
effectful arguments:
\begin{align*}
& \ensuremath{\Varid{h}}   \ensuremath{\mathbin{:}\Varid{b}}_1 \rightarrow \ensuremath{\Varid{b}}_2 \rightarrow \cdots \rightarrow \ensuremath{\Varid{b}}_n \rightarrow \ensuremath{\Varid{a}} \\
& \ensuremath{\Varid{x}}_1 \ensuremath{\mathbin{:}\Varid{f}\;\Varid{b}}_1 \\
& \ensuremath{\Varid{x}}_2 \ensuremath{\mathbin{:}\Varid{f}\;\Varid{b}}_2 \\
& \cdots           \\
& \ensuremath{\Varid{x}}_n \ensuremath{\mathbin{:}\Varid{f}\;\Varid{b}}_n
\end{align*}
the idiom bracket notation is defined as:
\begin{equation*}
\llbracket \: \ensuremath{\Varid{h}} \: \ensuremath{\Varid{x}}_1 \: \ensuremath{\Varid{x}}_2 \cdots \ensuremath{\Varid{x}}_n \: \rrbracket =
  \ensuremath{\Varid{pure}\;\Varid{h}} \ensuremath{\mathbin{\texttt{<*>}}} \ensuremath{\Varid{x}}_1 \ensuremath{\mathbin{\texttt{<*>}}} \ensuremath{\Varid{x}}_2 \ensuremath{\mathbin{\texttt{<*>}}} \cdots \ensuremath{\mathbin{\texttt{<*>}}} \ensuremath{\Varid{x}}_n
\end{equation*}
We can build such an expression formally by using a \ensuremath{\Conid{PureL}} constructor corresponding to
\ensuremath{\Varid{pure}} and a left-associative infix \ensuremath{\texttt{(}\mathbin{\texttt{\color{red}{:*:}}}\texttt{)}} constructor corresponding to \ensuremath{\texttt{(}\mathbin{\texttt{<*>}}\texttt{)}}:
\begin{equation*}
\ensuremath{\Conid{PureL}\;\Varid{h}} \ensuremath{\mathbin{\texttt{\color{red}{:*:}}}} \ensuremath{\Varid{x}}_1 \ensuremath{\mathbin{\texttt{\color{red}{:*:}}}} \ensuremath{\Varid{x}}_2 \ensuremath{\mathbin{\texttt{\color{red}{:*:}}}} \cdots \ensuremath{\mathbin{\texttt{\color{red}{:*:}}}} \ensuremath{\Varid{x}}_n
\end{equation*}
The corresponding inductive definition is:

\begin{hscode}\SaveRestoreHook
\column{B}{@{}>{\hspre}l<{\hspost}@{}}%
\column{3}{@{}>{\hspre}l<{\hspost}@{}}%
\column{35}{@{}>{\hspre}c<{\hspost}@{}}%
\column{35E}{@{}l@{}}%
\column{40}{@{}>{\hspre}l<{\hspost}@{}}%
\column{E}{@{}>{\hspre}l<{\hspost}@{}}%
\>[B]{}\texttt{\color{blue}\textbf{data}}\;\Conid{FreeAL}\;\Varid{f}\;\Varid{a}{}\<[E]%
\\
\>[B]{}\hsindent{3}{}\<[3]%
\>[3]{}\mathrel{=}\Conid{PureL}\;\Varid{a}{}\<[E]%
\\
\>[B]{}\hsindent{3}{}\<[3]%
\>[3]{}\mid \forall \Varid{b}\hsforall \hsdot{\circ }{\texttt{.}}\Conid{FreeAL}\;\Varid{f}\;\texttt{(}\Varid{b}\to \Varid{a}\texttt{)}{}\<[35]%
\>[35]{}\mathbin{\texttt{\color{red}{:*:}}}{}\<[35E]%
\>[40]{}\Varid{f}\;\Varid{b}{}\<[E]%
\\
\>[B]{}\texttt{\color{blue}\textbf{infixl}}\;\mathrm{4}\mathbin{\texttt{\color{red}{:*:}}}{}\<[E]%
\ColumnHook
\end{hscode}\resethooks

The \ensuremath{\Conid{MultiFunctor}} typeclass, the idiom brackets and the \ensuremath{\Conid{FreeAL}}
definition correspond to the left parenthesised
canonical form\footnote{Sometimes called simplified form because
it is not necessarily unique.} of expressions built with \ensuremath{\Varid{pure}} and
\ensuremath{\texttt{(}\mathbin{\texttt{<*>}}\texttt{)}}. Just as lists built with concatenation have two canonical forms
(cons-list and snoc-list) we can also define a right-parenthesised
canonical form for applicative functors --- a pure value over
which a sequence of effectful functions are applied:
\begin{align*}
& \ensuremath{\Varid{x}}   \ensuremath{\mathbin{:}\Varid{b}}_1                          \\
& \ensuremath{\Varid{h}}_1 \ensuremath{\mathbin{:}\Varid{f}\;\texttt{(}\Varid{b}}_1 \rightarrow \ensuremath{\Varid{b}}_2 \ensuremath{\texttt{)}} \\
& \ensuremath{\Varid{h}}_2 \ensuremath{\mathbin{:}\Varid{f}\;\texttt{(}\Varid{b}}_2 \rightarrow \ensuremath{\Varid{b}}_3 \ensuremath{\texttt{)}} \\
& \cdots                                  \\
& \ensuremath{\Varid{h}}_n \ensuremath{\mathbin{:}\Varid{f}\;\texttt{(}\Varid{b}}_n \rightarrow \ensuremath{\Varid{a}\texttt{)}}      \\
& \ensuremath{\Varid{h}}_n \ensuremath{\mathbin{\texttt{<*>}}} (\cdots \ensuremath{\mathbin{\texttt{<*>}}} (\ensuremath{\Varid{h}}_2 \ensuremath{\mathbin{\texttt{<*>}}} (\ensuremath{\Varid{h}}_1 \ensuremath{\mathbin{\texttt{<*>}}} \ensuremath{\Varid{pure}\;\Varid{x}})) \cdots )
\end{align*}
Replacing \ensuremath{\Varid{pure}} with a constructor \ensuremath{\Conid{Pure}} and \ensuremath{\texttt{(}\mathbin{\texttt{<*>}}\texttt{)}} by a right-associative infix \ensuremath{\texttt{(}\mathbin{\texttt{\color{red}{:\$:}}}\texttt{)}}
constructor gives the following expression:
\begin{equation*}
\ensuremath{\Varid{h}}_n \ensuremath{\mathbin{\texttt{\color{red}{:\$:}}}} \cdots \ensuremath{\mathbin{\texttt{\color{red}{:\$:}}}} \ensuremath{\Varid{h}}_2 \ensuremath{\mathbin{\texttt{\color{red}{:\$:}}}} \ensuremath{\Varid{h}}_1 \ensuremath{\mathbin{\texttt{\color{red}{:\$:}}}} \ensuremath{\Conid{Pure}\;\Varid{x}} \\
\end{equation*}
The corresponding inductive type:

\begin{hscode}\SaveRestoreHook
\column{B}{@{}>{\hspre}l<{\hspost}@{}}%
\column{3}{@{}>{\hspre}l<{\hspost}@{}}%
\column{E}{@{}>{\hspre}l<{\hspost}@{}}%
\>[B]{}\texttt{\color{blue}\textbf{data}}\;\Conid{FreeA}\;\Varid{f}\;\Varid{a}{}\<[E]%
\\
\>[B]{}\hsindent{3}{}\<[3]%
\>[3]{}\mathrel{=}\Conid{Pure}\;\Varid{a}{}\<[E]%
\\
\>[B]{}\hsindent{3}{}\<[3]%
\>[3]{}\mid \forall \Varid{b}\hsforall \hsdot{\circ }{\texttt{.}}\Varid{f}\;\texttt{(}\Varid{b}\to \Varid{a}\texttt{)}\mathbin{\texttt{\color{red}{:\$:}}}\Conid{FreeA}\;\Varid{f}\;\Varid{b}{}\<[E]%
\\
\>[B]{}\texttt{\color{blue}\textbf{infixr}}\;\mathrm{4}\mathbin{\texttt{\color{red}{:\$:}}}{}\<[E]%
\ColumnHook
\end{hscode}\resethooks

\ensuremath{\Conid{FreeAL}} and \ensuremath{\Conid{FreeA}} are isomorphic (see section \ref{sec:isomorphism}); we pick
the right-parenthesised version as our official definition since it is simpler
to define the \ensuremath{\Conid{Functor}} and \ensuremath{\Conid{Applicative}} instances:

\begin{hscode}\SaveRestoreHook
\column{B}{@{}>{\hspre}l<{\hspost}@{}}%
\column{3}{@{}>{\hspre}l<{\hspost}@{}}%
\column{21}{@{}>{\hspre}l<{\hspost}@{}}%
\column{E}{@{}>{\hspre}l<{\hspost}@{}}%
\>[B]{}\texttt{\color{blue}\textbf{instance}}\;\Conid{Functor}\;\Varid{f}\Rightarrow \Conid{Functor}\;\texttt{(}\Conid{FreeA}\;\Varid{f}\texttt{)}\;\texttt{\color{blue}\textbf{where}}{}\<[E]%
\\
\>[B]{}\hsindent{3}{}\<[3]%
\>[3]{}\Varid{fmap}\;\Varid{g}\;\texttt{(}\Conid{Pure}\;\Varid{x}\texttt{)}{}\<[21]%
\>[21]{}\mathrel{=}\Conid{Pure}\;\texttt{(}\Varid{g}\;\Varid{x}\texttt{)}{}\<[E]%
\\
\>[B]{}\hsindent{3}{}\<[3]%
\>[3]{}\Varid{fmap}\;\Varid{g}\;\texttt{(}\Varid{h}\mathbin{\texttt{\color{red}{:\$:}}}\Varid{x}\texttt{)}{}\<[21]%
\>[21]{}\mathrel{=}\Varid{fmap}\;\texttt{(}\Varid{g}\hsdot{\circ }{\texttt{.}}\texttt{)}\;\Varid{h}\mathbin{\texttt{\color{red}{:\$:}}}\Varid{x}{}\<[E]%
\ColumnHook
\end{hscode}\resethooks

The functor laws can be verified by structural induction, simply applying the
definitions and using the functor laws for \ensuremath{\Varid{f}}.

\begin{hscode}\SaveRestoreHook
\column{B}{@{}>{\hspre}l<{\hspost}@{}}%
\column{3}{@{}>{\hspre}l<{\hspost}@{}}%
\column{20}{@{}>{\hspre}l<{\hspost}@{}}%
\column{23}{@{}>{\hspre}l<{\hspost}@{}}%
\column{E}{@{}>{\hspre}l<{\hspost}@{}}%
\>[B]{}\texttt{\color{blue}\textbf{instance}}\;\Conid{Functor}\;\Varid{f}\Rightarrow \Conid{Applicative}\;\texttt{(}\Conid{FreeA}\;\Varid{f}\texttt{)}\;\texttt{\color{blue}\textbf{where}}{}\<[E]%
\\
\>[B]{}\hsindent{3}{}\<[3]%
\>[3]{}\Varid{pure}\mathrel{=}\Conid{Pure}{}\<[E]%
\\
\>[B]{}\hsindent{3}{}\<[3]%
\>[3]{}\Conid{Pure}\;\Varid{g}\mathbin{\texttt{<*>}}\Varid{y}{}\<[20]%
\>[20]{}\mathrel{=}\Varid{fmap}\;\Varid{g}\;\Varid{y}{}\<[E]%
\\
\>[B]{}\hsindent{3}{}\<[3]%
\>[3]{}\texttt{(}\Varid{h}\mathbin{\texttt{\color{red}{:\$:}}}\Varid{x}\texttt{)}\mathbin{\texttt{<*>}}\Varid{y}{}\<[20]%
\>[20]{}\mathrel{=}{}\<[23]%
\>[23]{}\Varid{fmap}\;\Varid{uncurry}\;\Varid{h}\mathbin{\texttt{\color{red}{:\$:}}}{}\<[E]%
\\
\>[23]{}\texttt{(}\texttt{(}\mathbin{\texttt{,}}\texttt{)}\mathbin{\texttt{<\$>}}\Varid{x}\mathbin{\texttt{<*>}}\Varid{y}\texttt{)}{}\<[E]%
\ColumnHook
\end{hscode}\resethooks

In the last clause of the \ensuremath{\Conid{Applicative}} instance, \ensuremath{\Varid{h}} has type \ensuremath{\Varid{f}\;\texttt{(}\Varid{x}\to \Varid{y}\to \Varid{z}\texttt{)}}, and we need to return a
value of type \ensuremath{\Conid{FreeA}\;\Varid{f}\;\Varid{z}}. Since \ensuremath{\texttt{(}\mathbin{\texttt{\color{red}{:\$:}}}\texttt{)}} only allows us to express applications of
1-argument ``functions'', we uncurry \ensuremath{\Varid{h}} to get a value of type \ensuremath{\Varid{f}\;\texttt{(}\texttt{(}\Varid{x}\mathbin{\texttt{,}}\Varid{y}\texttt{)}\to \Varid{z}\texttt{)}}, then we use \ensuremath{\texttt{(}\mathbin{\texttt{<*>}}\texttt{)}} recursively (see section \ref{sec:totality} for a
justification of this recursive call) to pair \ensuremath{\Varid{x}} and \ensuremath{\Varid{y}} into a value of type
\ensuremath{\Conid{FreeA}\;\Varid{f}\;\texttt{(}\Varid{x}\mathbin{\texttt{,}}\Varid{y}\texttt{)}}, and finally use the \ensuremath{\texttt{(}\mathbin{\texttt{\color{red}{:\$:}}}\texttt{)}} constructor to build the result.
Note the analogy between the definition of \ensuremath{\texttt{(}\mathbin{\texttt{<*>}}\texttt{)}} and \ensuremath{\texttt{(}\plus \texttt{)}} for lists.

\section{Applications}\label{sec:applications}

\subsection{Example: option parsers (continued)}

By using our definition of free applicative, we can compose the command line option parser
exactly as shown in section \ref{example:option_intro} in the definition of \ensuremath{\Varid{userP}}. 
The smart constructor \ensuremath{\Varid{one}} which lifts an 
option (a functor representing a basic operation of our embedded language) to a term in our
language can now be implemented as follows:

\begin{hscode}\SaveRestoreHook
\column{B}{@{}>{\hspre}l<{\hspost}@{}}%
\column{E}{@{}>{\hspre}l<{\hspost}@{}}%
\>[B]{}\Varid{one}\mathbin{::}\Conid{Option}\;\Varid{a}\to \Conid{FreeA}\;\Conid{Option}\;\Varid{a}{}\<[E]%
\\
\>[B]{}\Varid{one}\;\Varid{opt}\mathrel{=}\Varid{fmap}\;\Varid{const}\;\Varid{opt}\mathbin{\texttt{\color{red}{:\$:}}}\Conid{Pure}\;\texttt{(}\texttt{)}{}\<[E]%
\ColumnHook
\end{hscode}\resethooks

A function which computes the global default value of a parser can also be defined:

\begin{hscode}\SaveRestoreHook
\column{B}{@{}>{\hspre}l<{\hspost}@{}}%
\column{3}{@{}>{\hspre}l<{\hspost}@{}}%
\column{E}{@{}>{\hspre}l<{\hspost}@{}}%
\>[B]{}\Varid{parserDefault}\mathbin{::}\Conid{FreeA}\;\Conid{Option}\;\Varid{a}\to \Conid{Maybe}\;\Varid{a}{}\<[E]%
\\
\>[B]{}\Varid{parserDefault}\;\texttt{(}\Conid{Pure}\;\Varid{x}\texttt{)}\mathrel{=}\Conid{Just}\;\Varid{x}{}\<[E]%
\\
\>[B]{}\Varid{parserDefault}\;\texttt{(}\Varid{g}\mathbin{\texttt{\color{red}{:\$:}}}\Varid{x}\texttt{)}\mathrel{=}{}\<[E]%
\\
\>[B]{}\hsindent{3}{}\<[3]%
\>[3]{}\Varid{optDefault}\;\Varid{g}\mathbin{\texttt{<*>}}\Varid{parserDefault}\;\Varid{x}{}\<[E]%
\ColumnHook
\end{hscode}\resethooks

In section \ref{sec:adjoint} we show that our definition is a free
construction which gives us general ways to structure
programs. Specifically, we are able to define a generic version of
\ensuremath{\Varid{one}} which works for any functor. By exploiting the adjunction
describing the free construction we are able to shorten the definition
of \ensuremath{\Varid{parserDefault}}, define a function listing all possible options and
a function parsing a list of command line arguments given in arbitrary
order (section \ref{example:option_raise}).

\subsection{Example: web service client (continued)}

In section \ref{example:web_service_intro} we showed an embedded DSL for web
service clients based on free monads does not support certain kinds of static
analysis.

However, we can now remedy this by using a free applicative, over the same
functor \text{\tt WebService}.  In fact, the \ensuremath{\Varid{count}} function is now definable for \ensuremath{\Conid{FreeA}\;\Conid{WebService}\;\Varid{a}}.  Moreover, this is not limited to this particular example: it is
possible to define \ensuremath{\Varid{count}} for the free applicative over \emph{any} functor.

\begin{hscode}\SaveRestoreHook
\column{B}{@{}>{\hspre}l<{\hspost}@{}}%
\column{E}{@{}>{\hspre}l<{\hspost}@{}}%
\>[B]{}\Varid{count}\mathbin{::}\Conid{FreeA}\;\Varid{f}\;\Varid{a}\to \Conid{Int}{}\<[E]%
\\
\>[B]{}\Varid{count}\;\texttt{(}\Conid{Pure}\;\anonymous \texttt{)}\mathrel{=}\mathrm{0}{}\<[E]%
\\
\>[B]{}\Varid{count}\;\texttt{(}\anonymous \mathbin{\texttt{\color{red}{:\$:}}}\Varid{u}\texttt{)}\mathrel{=}\mathrm{1}\mathbin{+}\Varid{count}\;\Varid{u}{}\<[E]%
\ColumnHook
\end{hscode}\resethooks

Static analysis of the embedded code now also allows decorating
requests with parallelization instructions statically as well as
rearranging requests to the same server.

Of course, the extra power comes at a cost.  Namely, the expressivity
of the corresponding embedded language is severely reduced.

Using \text{\tt FreeA~WebService}, all the URLs of the servers to which
requests are sent must be known in advance, as well as the parameters
and content of every request.

In particular, what one posts to a server cannot depend on what has
been previously read from another server, so operations like \text{\tt copy}
cannot be implemented.

\subsection{Summary of examples}

Applicative functors are useful for describing certain kinds of effectful computations. The
free applicative construct over a given functor specifying the ``basic operations'' of an
embedded language gives rise to terms of the embedded DSL built by applicative operators. 
These terms are only capable of representing a certain kind of effectful computation which
can be described best with the help of the left-parenthesised canonical form: a pure
function applied to effectful arguments. The calculation of the arguments may involve effects
but in the end the arguments are composed by a pure function, which means that the effects
performed are fixed when specifying the applicative expression.

In the case of the option parser example \ensuremath{\Varid{userP}}, the pure function is given by the \ensuremath{\Conid{User}}
constructor and the ``basic operation'' \ensuremath{\Conid{Option}} is defining an option. The effects performed
depend on how an evaluator is defined over an expression of type \ensuremath{\Conid{FreeA}\;\Conid{Option}\;\Varid{a}} and the
order of effects can depend on the implementation of the evaluator.

For example, if one defines an embedded language for querying a database, and constructs
applicative expressions using \ensuremath{\Conid{FreeA}}, one might analyze the applicative expression and
collect information on the individual database queries by defining functions similar to
the \ensuremath{\Varid{count}} function in the web service example. Then, different, possibly expensive duplicate
queries can be merged and performed at once instead of executing the effectful computations
one by one. By restricting the expressivity of our language we gain freedom in defining how
the evaluator works.

One might define parts of an expression in an embedded DSL using the usual free monad
construction, other parts using \ensuremath{\Conid{FreeA}} and compose them by lifting the free applicative
expression to the free monad using the following function:

\begin{hscode}\SaveRestoreHook
\column{B}{@{}>{\hspre}l<{\hspost}@{}}%
\column{3}{@{}>{\hspre}l<{\hspost}@{}}%
\column{19}{@{}>{\hspre}l<{\hspost}@{}}%
\column{E}{@{}>{\hspre}l<{\hspost}@{}}%
\>[B]{}\Varid{liftA2M}\mathbin{::}\Conid{Functor}\;\Varid{f}\Rightarrow \Conid{FreeA}\;\Varid{f}\;\Varid{a}\to \Conid{Free}\;\Varid{f}\;\Varid{a}{}\<[E]%
\\
\>[B]{}\Varid{liftA2M}\;\texttt{(}\Conid{Pure}\;\Varid{x}\texttt{)}{}\<[19]%
\>[19]{}\mathrel{=}\Conid{Return}\;\Varid{x}{}\<[E]%
\\
\>[B]{}\Varid{liftA2M}\;\texttt{(}\Varid{h}\mathbin{\texttt{\color{red}{:\$:}}}\Varid{x}\texttt{)}\mathrel{=}\Conid{Free}{}\<[E]%
\\
\>[B]{}\hsindent{3}{}\<[3]%
\>[3]{}\texttt{(}\Varid{fmap}\;\texttt{(}\lambda \Varid{f}\to \Varid{fmap}\;\Varid{f}\;\texttt{(}\Varid{liftA2M}\;\Varid{x}\texttt{)}\texttt{)}\;\Varid{h}\texttt{)}{}\<[E]%
\ColumnHook
\end{hscode}\resethooks

In the parts of the expression defined using the free monad construction, the order of effects
is fixed and the effects performed can depend on the result of previous effectful computations,
while the free applicative parts have a fixed structure with effects not depending on each other.
The monadic parts of the computation can depend on the result of static analysis carried out
over the applicative part:

\begin{hscode}\SaveRestoreHook
\column{B}{@{}>{\hspre}l<{\hspost}@{}}%
\column{3}{@{}>{\hspre}l<{\hspost}@{}}%
\column{21}{@{}>{\hspre}l<{\hspost}@{}}%
\column{E}{@{}>{\hspre}l<{\hspost}@{}}%
\>[B]{}\Varid{test}\mathbin{::}\Conid{FreeA}\;\Conid{FileSystem}\;\Conid{Int}\to \Conid{Free}\;\Conid{FileSystem}\;\texttt{(}\texttt{)}{}\<[E]%
\\
\>[B]{}\Varid{test}\;\Varid{op}\mathrel{=}\texttt{\color{blue}\textbf{do}}{}\<[E]%
\\
\>[B]{}\hsindent{3}{}\<[3]%
\>[3]{}\mathbin{...}{}\<[E]%
\\
\>[B]{}\hsindent{3}{}\<[3]%
\>[3]{}\texttt{\color{blue}\textbf{let}}\;\Varid{n}\mathrel{=}\Varid{count}\;\Varid{op}{}\<[21]%
\>[21]{}\mbox{\onelinecomment  result of static analysis}{}\<[E]%
\\
\>[B]{}\hsindent{3}{}\<[3]%
\>[3]{}\Varid{n'}\leftarrow \Varid{liftA2M}\;\Varid{op}{}\<[21]%
\>[21]{}\mbox{\onelinecomment  result of applicative computation}{}\<[E]%
\\
\>[B]{}\hsindent{3}{}\<[3]%
\>[3]{}\Varid{max}\leftarrow \Varid{read}\;\text{\tt \char34 max\char34}{}\<[E]%
\\
\>[B]{}\hsindent{3}{}\<[3]%
\>[3]{}\Varid{when}\;\texttt{(}\Varid{max}\geq \Varid{n}\mathbin{+}\Varid{n'}\texttt{)}\mathbin{\texttt{\$}}\Varid{write}\;\text{\tt \char34 /tmp/test\char34}\;\text{\tt \char34 blah\char34}{}\<[E]%
\\
\>[B]{}\hsindent{3}{}\<[3]%
\>[3]{}\mathbin{...}{}\<[E]%
\ColumnHook
\end{hscode}\resethooks

The possibility of using the results of static analysis instead of the need of specifying them
by hand (in our example, this would account to counting certain function calls in an expression
by looking at the code) can make the program less redundant.

\section{Parametricity}
\label{sec-parametricity}

In order to prove anything about our free applicative construction, we need to
make an important observation about its definition.

The \ensuremath{\texttt{(}\mathbin{\texttt{\color{red}{:\$:}}}\texttt{)}} constructor is defined using an existential type \ensuremath{\Varid{b}}, and it is clear
intuitively that there is no way, given a value of the form \ensuremath{\Varid{g}\mathbin{\texttt{\color{red}{:\$:}}}\Varid{x}}, to make
use of the type \ensuremath{\Varid{b}} hidden in it.

More specifically, any function on \ensuremath{\Conid{FreeA}\;\Varid{f}\;\Varid{a}} must be defined
\emph{polymorphically} over all possible types \ensuremath{\Varid{b}} which could be used for the
existentially quantified variable in the definition of \ensuremath{\texttt{(}\mathbin{\texttt{\color{red}{:\$:}}}\texttt{)}}.

To make this intuition precise, we assume that some form of
\emph{relational parametricity} \cite{abstraction}
\cite{free_theorems} holds in our total subset of Haskell. In particular, in the case
of the \ensuremath{\texttt{(}\mathbin{\texttt{\color{red}{:\$:}}}\texttt{)}} constructor, we require that:

\begin{hscode}\SaveRestoreHook
\column{B}{@{}>{\hspre}l<{\hspost}@{}}%
\column{E}{@{}>{\hspre}l<{\hspost}@{}}%
\>[B]{}\texttt{(}\mathbin{\texttt{\color{red}{:\$:}}}\texttt{)}\mathbin{::}\forall \Varid{b}\hsforall \hsdot{\circ }{\texttt{.}}\Varid{f}\;\texttt{(}\Varid{b}\to \Varid{a}\texttt{)}\to \texttt{(}\Conid{FreeA}\;\Varid{f}\;\Varid{b}\to \Conid{FreeA}\;\Varid{f}\;\Varid{a}\texttt{)}{}\<[E]%
\ColumnHook
\end{hscode}\resethooks

is a natural transformation of contravariant functors. The two contravariant
functors here could be defined, in Haskell, using a \ensuremath{\texttt{\color{blue}\textbf{newtype}}}:

\begin{hscode}\SaveRestoreHook
\column{B}{@{}>{\hspre}l<{\hspost}@{}}%
\column{3}{@{}>{\hspre}l<{\hspost}@{}}%
\column{E}{@{}>{\hspre}l<{\hspost}@{}}%
\>[B]{}\texttt{\color{blue}\textbf{newtype}}\;\Conid{F1}\;\Varid{f}\;\Varid{a}\;\Varid{x}\mathrel{=}\Conid{F1}\;\texttt{(}\Varid{f}\;\texttt{(}\Varid{x}\to \Varid{a}\texttt{)}\texttt{)}{}\<[E]%
\\
\>[B]{}\texttt{\color{blue}\textbf{newtype}}\;\Conid{F2}\;\Varid{f}\;\Varid{a}\;\Varid{x}\mathrel{=}\Conid{F2}\;\texttt{(}\Conid{FreeA}\;\Varid{f}\;\Varid{x}\to \Conid{FreeA}\;\Varid{f}\;\Varid{a}\texttt{)}{}\<[E]%
\\[\blanklineskip]%
\>[B]{}\texttt{\color{blue}\textbf{instance}}\;\Conid{Functor}\;\Varid{f}\Rightarrow \Conid{Contravariant}\;\texttt{(}\Conid{F1}\;\Varid{f}\;\Varid{a}\texttt{)}\;\texttt{\color{blue}\textbf{where}}{}\<[E]%
\\
\>[B]{}\hsindent{3}{}\<[3]%
\>[3]{}\Varid{contramap}\;\Varid{h}\;\texttt{(}\Conid{F1}\;\Varid{g}\texttt{)}\mathrel{=}\Conid{F1}\mathbin{\texttt{\$}}\Varid{fmap}\;\texttt{(}\hsdot{\circ }{\texttt{.}}\Varid{h}\texttt{)}\;\Varid{g}{}\<[E]%
\\[\blanklineskip]%
\>[B]{}\texttt{\color{blue}\textbf{instance}}\;\Conid{Functor}\;\Varid{f}\Rightarrow \Conid{Contravariant}\;\texttt{(}\Conid{F2}\;\Varid{f}\;\Varid{a}\texttt{)}\;\texttt{\color{blue}\textbf{where}}{}\<[E]%
\\
\>[B]{}\hsindent{3}{}\<[3]%
\>[3]{}\Varid{contramap}\;\Varid{h}\;\texttt{(}\Conid{F2}\;\Varid{g}\texttt{)}\mathrel{=}\Conid{F2}\mathbin{\texttt{\$}}\Varid{g}\hsdot{\circ }{\texttt{.}}\Varid{fmap}\;\Varid{h}{}\<[E]%
\ColumnHook
\end{hscode}\resethooks

The action of \ensuremath{\Conid{F1}} and \ensuremath{\Conid{F2}} on morphisms is defined in the obvious way. Note
that here we make use of the fact that \ensuremath{\Conid{FreeA}\;\Varid{f}} is a functor.

Naturality of \ensuremath{\texttt{(}\mathbin{\texttt{\color{red}{:\$:}}}\texttt{)}} means that, given types \ensuremath{\Varid{x}} and \ensuremath{\Varid{y}}, and a function \ensuremath{\Varid{h}\mathbin{:}\Varid{x}\to \Varid{y}}, the following holds:
\begin{align}\label{naturality}
& \ensuremath{\forall \Varid{g}\hsforall \mathbin{::}\Varid{f}\;\texttt{(}\Varid{y}\to \Varid{a}\texttt{)}}, \ensuremath{\Varid{u}\mathbin{::}\Conid{FreeA}\;\Varid{f}\;\Varid{x}\hsdot{\circ }{\texttt{.}}} \nonumber\\
& \ensuremath{\Varid{fmap}\;\texttt{(}\hsdot{\circ }{\texttt{.}}\Varid{h}\texttt{)}\;\Varid{g}\mathbin{\texttt{\color{red}{:\$:}}}\Varid{u}} \equiv \ensuremath{\Varid{g}\mathbin{\texttt{\color{red}{:\$:}}}\Varid{fmap}\;\Varid{h}\;\Varid{u}}
\end{align}
where we have unfolded the definitions of \ensuremath{\Varid{contramap}} for \ensuremath{\Conid{F1}} and \ensuremath{\Conid{F2}}, and
removed the newtypes.

Note that the results in \cite{abstraction} do not actually imply
naturality of \ensuremath{\texttt{(}\mathbin{\texttt{\color{red}{:\$:}}}\texttt{)}} at this generality, since \ensuremath{\Varid{f}} is a type variable
with an arbitrary \ensuremath{\Conid{Functor}} instance, not a concrete positive type
expression together with its canonical instance. However, in the
interpretation given in sections \ref{sec:totality} and
\ref{sec:semantics}, \ensuremath{\Conid{FreeA}} will be defined in such a way that
equation \ref{naturality} holds automatically.

\section{Isomorphism of the two definitions} \label{sec:isomorphism}

In this section we show that the two definitions of free applicatives given in
section \ref{sec:def} are isomorphic.

First of all, if \ensuremath{\Varid{f}} is a functor, \ensuremath{\Conid{FreeAL}\;\Varid{f}} is also a functor:

\begin{hscode}\SaveRestoreHook
\column{B}{@{}>{\hspre}l<{\hspost}@{}}%
\column{3}{@{}>{\hspre}l<{\hspost}@{}}%
\column{21}{@{}>{\hspre}l<{\hspost}@{}}%
\column{E}{@{}>{\hspre}l<{\hspost}@{}}%
\>[B]{}\texttt{\color{blue}\textbf{instance}}\;\Conid{Functor}\;\Varid{f}\Rightarrow \Conid{Functor}\;\texttt{(}\Conid{FreeAL}\;\Varid{f}\texttt{)}\;\texttt{\color{blue}\textbf{where}}{}\<[E]%
\\
\>[B]{}\hsindent{3}{}\<[3]%
\>[3]{}\Varid{fmap}\;\Varid{g}\;\texttt{(}\Conid{PureL}\;\Varid{x}\texttt{)}{}\<[21]%
\>[21]{}\mathrel{=}\Conid{PureL}\;\texttt{(}\Varid{g}\;\Varid{x}\texttt{)}{}\<[E]%
\\
\>[B]{}\hsindent{3}{}\<[3]%
\>[3]{}\Varid{fmap}\;\Varid{g}\;\texttt{(}\Varid{h}\mathbin{\texttt{\color{red}{:*:}}}\Varid{x}\texttt{)}{}\<[21]%
\>[21]{}\mathrel{=}\texttt{(}\Varid{fmap}\;\texttt{(}\Varid{g}\hsdot{\circ }{\texttt{.}}\texttt{)}\;\Varid{h}\texttt{)}\mathbin{\texttt{\color{red}{:*:}}}\Varid{x}{}\<[E]%
\ColumnHook
\end{hscode}\resethooks

Again, the functor laws can be verified by a simple structural induction.

For the \ensuremath{\texttt{(}\mathbin{\texttt{\color{red}{:*:}}}\texttt{)}} constructor, a free theorem can be derived in a completely analogous way
to deriving equation \ref{naturality}. This equation states that \ensuremath{\texttt{(}\mathbin{\texttt{\color{red}{:*:}}}\texttt{)}} is a natural
transformation:
\begin{align}\label{naturalityL}
& \ensuremath{\forall \Varid{h}\hsforall \mathbin{::}\Varid{x}\to \Varid{y}}, \ensuremath{\Varid{g}\mathbin{::}\Conid{FreeAL}\;\Varid{f}\;\texttt{(}\Varid{y}\to \Varid{a}\texttt{)}}, \ensuremath{\Varid{u}\mathbin{::}\Varid{f}\;\Varid{x}\hsdot{\circ }{\texttt{.}}} \nonumber\\
& \ensuremath{\Varid{fmap}\;\texttt{(}\hsdot{\circ }{\texttt{.}}\Varid{h}\texttt{)}\;\Varid{g}\mathbin{\texttt{\color{red}{:*:}}}\Varid{u}} \equiv \ensuremath{\Varid{g}\mathbin{\texttt{\color{red}{:*:}}}\Varid{fmap}\;\Varid{h}\;\Varid{u}}
\end{align}

We define functions to convert between the two definitions:

\begin{hscode}\SaveRestoreHook
\column{B}{@{}>{\hspre}l<{\hspost}@{}}%
\column{15}{@{}>{\hspre}l<{\hspost}@{}}%
\column{39}{@{}>{\hspre}l<{\hspost}@{}}%
\column{E}{@{}>{\hspre}l<{\hspost}@{}}%
\>[B]{}\Varid{r2l}\mathbin{::}\Conid{Functor}\;\Varid{f}\Rightarrow \Conid{FreeA}\;\Varid{f}\;\Varid{a}\to \Conid{FreeAL}\;\Varid{f}\;\Varid{a}{}\<[E]%
\\
\>[B]{}\Varid{r2l}\;\texttt{(}\Conid{Pure}\;\Varid{x}\texttt{)}{}\<[15]%
\>[15]{}\mathrel{=}\Conid{PureL}\;\Varid{x}{}\<[E]%
\\
\>[B]{}\Varid{r2l}\;\texttt{(}\Varid{h}\mathbin{\texttt{\color{red}{:\$:}}}\Varid{x}\texttt{)}\mathrel{=}\Varid{fmap}\;\texttt{(}\Varid{flip}\;\texttt{(}\mathbin{\texttt{\$}}\texttt{)}\texttt{)}\;\texttt{(}\Varid{r2l}\;{}\<[39]%
\>[39]{}\Varid{x}\texttt{)}\mathbin{\texttt{\color{red}{:*:}}}\Varid{h}{}\<[E]%
\\[\blanklineskip]%
\>[B]{}\Varid{l2r}\mathbin{::}\Conid{Functor}\;\Varid{f}\Rightarrow \Conid{FreeAL}\;\Varid{f}\;\Varid{a}\to \Conid{FreeA}\;\Varid{f}\;\Varid{a}{}\<[E]%
\\
\>[B]{}\Varid{l2r}\;\texttt{(}\Conid{PureL}\;\Varid{x}\texttt{)}\mathrel{=}\Conid{Pure}\;\Varid{x}{}\<[E]%
\\
\>[B]{}\Varid{l2r}\;\texttt{(}\Varid{h}\mathbin{\texttt{\color{red}{:*:}}}\Varid{x}\texttt{)}\mathrel{=}\Varid{fmap}\;\texttt{(}\Varid{flip}\;\texttt{(}\mathbin{\texttt{\$}}\texttt{)}\texttt{)}\;\Varid{x}\mathbin{\texttt{\color{red}{:\$:}}}\Varid{l2r}\;\Varid{h}{}\<[E]%
\ColumnHook
\end{hscode}\resethooks

We will also need the fact that \ensuremath{\Varid{l2r}} is a natural transformation:
\begin{align}\label{naturalityl2r}
& \ensuremath{\forall \Varid{h}\hsforall \mathbin{::}\Varid{x}\to \Varid{y}}, \ensuremath{\Varid{u}\mathbin{::}\Conid{FreeAL}\;\Varid{f}\;\Varid{x}\hsdot{\circ }{\texttt{.}}} \nonumber\\
& \ensuremath{\Varid{l2r}\;\texttt{(}\Varid{fmap}\;\Varid{h}\;\Varid{u}\texttt{)}} \equiv \ensuremath{\Varid{fmap}\;\Varid{h}\;\texttt{(}\Varid{l2r}\;\Varid{u}\texttt{)}}
\end{align}

\begin{prop}
\ensuremath{\Varid{r2l}} is an isomorphism, the inverse of which is \ensuremath{\Varid{l2r}}.
\end{prop}
\begin{proof}
  First we prove that \ensuremath{\forall \Varid{u}\hsforall \mathbin{::}\Conid{FreeA}\;\Varid{f}\;\Varid{a}\hsdot{\circ }{\texttt{.}}\Varid{l2r}\;\texttt{(}\Varid{r2l}\;\Varid{u}\texttt{)}} $\equiv$ \ensuremath{\Varid{u}}.
  We compute using equational reasoning with induction on \ensuremath{\Varid{u}}:
\begin{hscode}\SaveRestoreHook
\column{B}{@{}>{\hspre}l<{\hspost}@{}}%
\column{7}{@{}>{\hspre}l<{\hspost}@{}}%
\column{E}{@{}>{\hspre}l<{\hspost}@{}}%
\>[7]{}\Varid{l2r}\;\texttt{(}\Varid{r2l}\;\texttt{(}\Conid{Pure}\;\Varid{x}\texttt{)}\texttt{)}{}\<[E]%
\\
\>[B]{}\equiv \mbox{\commentbegin  definition of \ensuremath{\Varid{r2l}}  \commentend}{}\<[E]%
\\
\>[B]{}\hsindent{7}{}\<[7]%
\>[7]{}\Varid{l2r}\;\texttt{(}\Conid{PureL}\;\Varid{x}\texttt{)}{}\<[E]%
\\
\>[B]{}\equiv \mbox{\commentbegin  definition of \ensuremath{\Varid{l2r}}  \commentend}{}\<[E]%
\\
\>[B]{}\hsindent{7}{}\<[7]%
\>[7]{}\Conid{Pure}\;\Varid{x}{}\<[E]%
\ColumnHook
\end{hscode}\resethooks

\begin{hscode}\SaveRestoreHook
\column{B}{@{}>{\hspre}l<{\hspost}@{}}%
\column{7}{@{}>{\hspre}l<{\hspost}@{}}%
\column{34}{@{}>{\hspre}l<{\hspost}@{}}%
\column{E}{@{}>{\hspre}l<{\hspost}@{}}%
\>[7]{}\Varid{l2r}\;\texttt{(}\Varid{r2l}\;\texttt{(}\Varid{h}\mathbin{\texttt{\color{red}{:\$:}}}\Varid{x}\texttt{)}\texttt{)}{}\<[E]%
\\
\>[B]{}\equiv \mbox{\commentbegin  definition of \ensuremath{\Varid{r2l}}  \commentend}{}\<[E]%
\\
\>[B]{}\hsindent{7}{}\<[7]%
\>[7]{}\Varid{l2r}\;\texttt{(}\Varid{fmap}\;\texttt{(}\Varid{flip}\;\texttt{(}\mathbin{\texttt{\$}}\texttt{)}\texttt{)}\;\texttt{(}\Varid{r2l}\;{}\<[34]%
\>[34]{}\Varid{x}\texttt{)}\mathbin{\texttt{\color{red}{:*:}}}\Varid{h}\texttt{)}{}\<[E]%
\\
\>[B]{}\equiv \mbox{\commentbegin  definition of \ensuremath{\Varid{l2r}}  \commentend}{}\<[E]%
\\
\>[B]{}\hsindent{7}{}\<[7]%
\>[7]{}\Varid{fmap}\;\texttt{(}\Varid{flip}\;\texttt{(}\mathbin{\texttt{\$}}\texttt{)}\texttt{)}\;\Varid{h}\mathbin{\texttt{\color{red}{:\$:}}}{}\<[E]%
\\
\>[B]{}\hsindent{7}{}\<[7]%
\>[7]{}\Varid{l2r}\;\texttt{(}\Varid{fmap}\;\texttt{(}\Varid{flip}\;\texttt{(}\mathbin{\texttt{\$}}\texttt{)}\texttt{)}\;\texttt{(}\Varid{r2l}\;\Varid{x}\texttt{)}\texttt{)}{}\<[E]%
\\
\>[B]{}\equiv \mbox{\commentbegin  equation \ref{naturalityl2r}  \commentend}{}\<[E]%
\\
\>[B]{}\hsindent{7}{}\<[7]%
\>[7]{}\Varid{fmap}\;\texttt{(}\Varid{flip}\;\texttt{(}\mathbin{\texttt{\$}}\texttt{)}\texttt{)}\;\Varid{h}\mathbin{\texttt{\color{red}{:\$:}}}{}\<[E]%
\\
\>[B]{}\hsindent{7}{}\<[7]%
\>[7]{}\Varid{fmap}\;\texttt{(}\Varid{flip}\;\texttt{(}\mathbin{\texttt{\$}}\texttt{)}\texttt{)}\;\texttt{(}\Varid{l2r}\;\texttt{(}\Varid{r2l}\;\Varid{x}\texttt{)}\texttt{)}{}\<[E]%
\\
\>[B]{}\equiv \mbox{\commentbegin  inductive hypothesis  \commentend}{}\<[E]%
\\
\>[B]{}\hsindent{7}{}\<[7]%
\>[7]{}\Varid{fmap}\;\texttt{(}\Varid{flip}\;\texttt{(}\mathbin{\texttt{\$}}\texttt{)}\texttt{)}\;\Varid{h}\mathbin{\texttt{\color{red}{:\$:}}}\Varid{fmap}\;\texttt{(}\Varid{flip}\;\texttt{(}\mathbin{\texttt{\$}}\texttt{)}\texttt{)}\;\Varid{x}{}\<[E]%
\\
\>[B]{}\equiv \mbox{\commentbegin  equation \ref{naturality}  \commentend}{}\<[E]%
\\
\>[B]{}\hsindent{7}{}\<[7]%
\>[7]{}\Varid{fmap}\;\texttt{(}\hsdot{\circ }{\texttt{.}}\texttt{(}\Varid{flip}\;\texttt{(}\mathbin{\texttt{\$}}\texttt{)}\texttt{)}\texttt{)}\;\texttt{(}\Varid{fmap}\;\texttt{(}\Varid{flip}\;\texttt{(}\mathbin{\texttt{\$}}\texttt{)}\texttt{)}\;\Varid{h}\texttt{)}\mathbin{\texttt{\color{red}{:\$:}}}\Varid{x}{}\<[E]%
\\
\>[B]{}\equiv \mbox{\commentbegin  \ensuremath{\Varid{f}} is a functor  \commentend}{}\<[E]%
\\
\>[B]{}\hsindent{7}{}\<[7]%
\>[7]{}\Varid{fmap}\;\texttt{(}\texttt{(}\hsdot{\circ }{\texttt{.}}\texttt{(}\Varid{flip}\;\texttt{(}\mathbin{\texttt{\$}}\texttt{)}\texttt{)}\texttt{)}\hsdot{\circ }{\texttt{.}}\Varid{flip}\;\texttt{(}\mathbin{\texttt{\$}}\texttt{)}\texttt{)}\;\Varid{h}\mathbin{\texttt{\color{red}{:\$:}}}\Varid{x}{}\<[E]%
\\
\>[B]{}\equiv \mbox{\commentbegin  definition of \ensuremath{\Varid{flip}} and \ensuremath{\texttt{(}\mathbin{\texttt{\$}}\texttt{)}}  \commentend}{}\<[E]%
\\
\>[B]{}\hsindent{7}{}\<[7]%
\>[7]{}\Varid{fmap}\;\Varid{id}\;\Varid{h}\mathbin{\texttt{\color{red}{:\$:}}}\Varid{x}{}\<[E]%
\\
\>[B]{}\equiv \mbox{\commentbegin  \ensuremath{\Varid{f}} is a functor  \commentend}{}\<[E]%
\\
\>[B]{}\hsindent{7}{}\<[7]%
\>[7]{}\Varid{h}\mathbin{\texttt{\color{red}{:\$:}}}\Varid{x}{}\<[E]%
\ColumnHook
\end{hscode}\resethooks

Next, we prove that \ensuremath{\forall \Varid{u}\hsforall \mathbin{::}\Conid{FreeAL}\;\Varid{f}\;\Varid{a}\hsdot{\circ }{\texttt{.}}\Varid{r2l}\;\texttt{(}\Varid{l2r}\;\Varid{u}\texttt{)}} $\equiv$ \ensuremath{\Varid{u}}. Again,
we compute using equational reasoning with induction on \ensuremath{\Varid{u}}:

\begin{hscode}\SaveRestoreHook
\column{B}{@{}>{\hspre}l<{\hspost}@{}}%
\column{7}{@{}>{\hspre}l<{\hspost}@{}}%
\column{E}{@{}>{\hspre}l<{\hspost}@{}}%
\>[7]{}\Varid{r2l}\;\texttt{(}\Varid{l2r}\;\texttt{(}\Conid{PureL}\;\Varid{x}\texttt{)}\texttt{)}{}\<[E]%
\\
\>[B]{}\equiv \mbox{\commentbegin  definition of \ensuremath{\Varid{l2r}}  \commentend}{}\<[E]%
\\
\>[B]{}\hsindent{7}{}\<[7]%
\>[7]{}\Varid{r2l}\;\texttt{(}\Conid{Pure}\;\Varid{x}\texttt{)}{}\<[E]%
\\
\>[B]{}\equiv \mbox{\commentbegin  definition of \ensuremath{\Varid{r2l}}  \commentend}{}\<[E]%
\\
\>[B]{}\hsindent{7}{}\<[7]%
\>[7]{}\Conid{PureL}\;\Varid{x}{}\<[E]%
\ColumnHook
\end{hscode}\resethooks

\begin{hscode}\SaveRestoreHook
\column{B}{@{}>{\hspre}l<{\hspost}@{}}%
\column{7}{@{}>{\hspre}l<{\hspost}@{}}%
\column{E}{@{}>{\hspre}l<{\hspost}@{}}%
\>[7]{}\Varid{r2l}\;\texttt{(}\Varid{l2r}\;\texttt{(}\Varid{h}\mathbin{\texttt{\color{red}{:*:}}}\Varid{x}\texttt{)}\texttt{)}{}\<[E]%
\\
\>[B]{}\equiv \mbox{\commentbegin  definition of \ensuremath{\Varid{l2r}}  \commentend}{}\<[E]%
\\
\>[B]{}\hsindent{7}{}\<[7]%
\>[7]{}\Varid{r2l}\;\texttt{(}\Varid{fmap}\;\texttt{(}\Varid{flip}\;\texttt{(}\mathbin{\texttt{\$}}\texttt{)}\texttt{)}\;\Varid{x}\mathbin{\texttt{\color{red}{:\$:}}}\Varid{l2r}\;\Varid{h}\texttt{)}{}\<[E]%
\\
\>[B]{}\equiv \mbox{\commentbegin  definition of \ensuremath{\Varid{r2l}}  \commentend}{}\<[E]%
\\
\>[B]{}\hsindent{7}{}\<[7]%
\>[7]{}\Varid{fmap}\;\texttt{(}\Varid{flip}\;\texttt{(}\mathbin{\texttt{\$}}\texttt{)}\texttt{)}\;\texttt{(}\Varid{r2l}\;\texttt{(}\Varid{l2r}\;\Varid{h}\texttt{)}\texttt{)}\mathbin{\texttt{\color{red}{:*:}}}\Varid{fmap}\;\texttt{(}\Varid{flip}\;\texttt{(}\mathbin{\texttt{\$}}\texttt{)}\texttt{)}\;\Varid{x}{}\<[E]%
\\
\>[B]{}\equiv \mbox{\commentbegin  inductive hypothesis  \commentend}{}\<[E]%
\\
\>[B]{}\hsindent{7}{}\<[7]%
\>[7]{}\Varid{fmap}\;\texttt{(}\Varid{flip}\;\texttt{(}\mathbin{\texttt{\$}}\texttt{)}\texttt{)}\;\Varid{h}\mathbin{\texttt{\color{red}{:*:}}}\Varid{fmap}\;\texttt{(}\Varid{flip}\;\texttt{(}\mathbin{\texttt{\$}}\texttt{)}\texttt{)}\;\Varid{x}{}\<[E]%
\\
\>[B]{}\equiv \mbox{\commentbegin  equation \ref{naturalityL}  \commentend}{}\<[E]%
\\
\>[B]{}\hsindent{7}{}\<[7]%
\>[7]{}\Varid{fmap}\;\texttt{(}\hsdot{\circ }{\texttt{.}}\texttt{(}\Varid{flip}\;\texttt{(}\mathbin{\texttt{\$}}\texttt{)}\texttt{)}\texttt{)}\;\texttt{(}\Varid{fmap}\;\texttt{(}\Varid{flip}\;\texttt{(}\mathbin{\texttt{\$}}\texttt{)}\texttt{)}\;\Varid{h}\texttt{)}\mathbin{\texttt{\color{red}{:*:}}}\Varid{x}{}\<[E]%
\\
\>[B]{}\equiv \mbox{\commentbegin  \ensuremath{\Conid{FreeAL}\;\Varid{f}} is a functor  \commentend}{}\<[E]%
\\
\>[B]{}\hsindent{7}{}\<[7]%
\>[7]{}\Varid{fmap}\;\texttt{(}\texttt{(}\hsdot{\circ }{\texttt{.}}\texttt{(}\Varid{flip}\;\texttt{(}\mathbin{\texttt{\$}}\texttt{)}\texttt{)}\texttt{)}\hsdot{\circ }{\texttt{.}}\Varid{flip}\;\texttt{(}\mathbin{\texttt{\$}}\texttt{)}\texttt{)}\;\Varid{h}\mathbin{\texttt{\color{red}{:*:}}}\Varid{x}{}\<[E]%
\\
\>[B]{}\equiv \mbox{\commentbegin  definition of \ensuremath{\Varid{flip}} and \ensuremath{\texttt{(}\mathbin{\texttt{\$}}\texttt{)}}  \commentend}{}\<[E]%
\\
\>[B]{}\hsindent{7}{}\<[7]%
\>[7]{}\Varid{fmap}\;\Varid{id}\;\Varid{h}\mathbin{\texttt{\color{red}{:*:}}}\Varid{x}{}\<[E]%
\\
\>[B]{}\equiv \mbox{\commentbegin  \ensuremath{\Conid{FreeAL}\;\Varid{f}} is a functor  \commentend}{}\<[E]%
\\
\>[B]{}\hsindent{7}{}\<[7]%
\>[7]{}\Varid{h}\mathbin{\texttt{\color{red}{:*:}}}\Varid{x}{}\<[E]%
\ColumnHook
\end{hscode}\resethooks
\end{proof}

In the next sections, we will prove that \ensuremath{\Conid{FreeA}} is a free applicative functor.
Because of the isomorphism of the two definitions, these results will carry over
to \ensuremath{\Conid{FreeAL}}.

\section{Applicative laws}
\label{sec:applicative}

Following \cite{applicative}, the laws for an \ensuremath{\Conid{Applicative}} instance are:
\begin{align}
  \ensuremath{\Varid{pure}\;\Varid{id}\mathbin{\texttt{<*>}}\Varid{u}} & \equiv \ensuremath{\Varid{u}} \label{app:id} \\
  \ensuremath{\Varid{pure}\;\texttt{(}\hsdot{\circ }{\texttt{.}}\texttt{)}\mathbin{\texttt{<*>}}\Varid{u}\mathbin{\texttt{<*>}}\Varid{v}\mathbin{\texttt{<*>}}\Varid{x}} & \equiv \ensuremath{\Varid{u}\mathbin{\texttt{<*>}}\texttt{(}\Varid{v}\mathbin{\texttt{<*>}}\Varid{x}\texttt{)}} \label{app:comp} \\
  \ensuremath{\Varid{pure}\;\Varid{f}\mathbin{\texttt{<*>}}\Varid{pure}\;\Varid{x}} & \equiv \ensuremath{\Varid{pure}\;\texttt{(}\Varid{f}\;\Varid{x}\texttt{)}} \label{app:hom} \\
  \ensuremath{\Varid{u}\mathbin{\texttt{<*>}}\Varid{pure}\;\Varid{x}} & \equiv \ensuremath{\Varid{pure}\;\texttt{(}\mathbin{\texttt{\$}}\Varid{x}\texttt{)}\mathbin{\texttt{<*>}}\Varid{u}} \label{app:interchange}
\end{align}

We introduce a few abbreviations to help make the notation lighter:

\begin{hscode}\SaveRestoreHook
\column{B}{@{}>{\hspre}l<{\hspost}@{}}%
\column{E}{@{}>{\hspre}l<{\hspost}@{}}%
\>[B]{}\Varid{uc}\mathrel{=}\Varid{uncurry}{}\<[E]%
\\
\>[B]{}\Varid{pair}\;\Varid{x}\;\Varid{y}\mathrel{=}\texttt{(}\mathbin{\texttt{,}}\texttt{)}\mathbin{\texttt{<\$>}}\Varid{x}\mathbin{\texttt{<*>}}\Varid{y}{}\<[E]%
\ColumnHook
\end{hscode}\resethooks

\begin{lem} \label{lem:pure}
For all
\begin{align*}
  \ensuremath{\Varid{u}} & \ensuremath{\mathbin{::}\Varid{y}\to \Varid{z}} \\
  \ensuremath{\Varid{v}} & \ensuremath{\mathbin{::}\Conid{FreeA}\;\Varid{f}\;\texttt{(}\Varid{x}\to \Varid{y}\texttt{)}} \\
  \ensuremath{\Varid{x}} & \ensuremath{\mathbin{::}\Conid{FreeA}\;\Varid{f}\;\Varid{x}}
\end{align*}
the following equation holds:
$$
  \ensuremath{\Varid{fmap}\;\Varid{u}\;\texttt{(}\Varid{v}\mathbin{\texttt{<*>}}\Varid{x}\texttt{)}} \equiv \ensuremath{\Varid{fmap}\;\texttt{(}\Varid{u}\hsdot{\circ }{\texttt{.}}\texttt{)}\;\Varid{v}\mathbin{\texttt{<*>}}\Varid{x}}
$$
\end{lem}
\begin{proof}
We compute:
\begin{hscode}\SaveRestoreHook
\column{B}{@{}>{\hspre}l<{\hspost}@{}}%
\column{5}{@{}>{\hspre}l<{\hspost}@{}}%
\column{E}{@{}>{\hspre}l<{\hspost}@{}}%
\>[5]{}\Varid{fmap}\;\Varid{u}\;\texttt{(}\Conid{Pure}\;\Varid{v}\mathbin{\texttt{<*>}}\Varid{x}\texttt{)}{}\<[E]%
\\
\>[B]{}\equiv \mbox{\commentbegin  definition of \ensuremath{\texttt{(}\mathbin{\texttt{<*>}}\texttt{)}}  \commentend}{}\<[E]%
\\
\>[B]{}\hsindent{5}{}\<[5]%
\>[5]{}\Varid{fmap}\;\Varid{u}\;\texttt{(}\Varid{fmap}\;\Varid{v}\;\Varid{x}\texttt{)}{}\<[E]%
\\
\>[B]{}\equiv \mbox{\commentbegin  \ensuremath{\Conid{FreeA}\;\Varid{f}} is a functor  \commentend}{}\<[E]%
\\
\>[B]{}\hsindent{5}{}\<[5]%
\>[5]{}\Varid{fmap}\;\texttt{(}\Varid{u}\hsdot{\circ }{\texttt{.}}\Varid{v}\texttt{)}\;\Varid{x}{}\<[E]%
\\
\>[B]{}\equiv \mbox{\commentbegin  definition of \ensuremath{\texttt{(}\mathbin{\texttt{<*>}}\texttt{)}}  \commentend}{}\<[E]%
\\
\>[B]{}\hsindent{5}{}\<[5]%
\>[5]{}\Conid{Pure}\;\texttt{(}\Varid{u}\hsdot{\circ }{\texttt{.}}\Varid{v}\texttt{)}\mathbin{\texttt{<*>}}\Varid{x}{}\<[E]%
\\
\>[B]{}\equiv \mbox{\commentbegin  definition of \ensuremath{\Varid{fmap}}  \commentend}{}\<[E]%
\\
\>[B]{}\hsindent{5}{}\<[5]%
\>[5]{}\Varid{fmap}\;\texttt{(}\Varid{u}\hsdot{\circ }{\texttt{.}}\texttt{)}\;\texttt{(}\Conid{Pure}\;\Varid{v}\texttt{)}\mathbin{\texttt{<*>}}\Varid{x}{}\<[E]%
\ColumnHook
\end{hscode}\resethooks

\begin{hscode}\SaveRestoreHook
\column{B}{@{}>{\hspre}l<{\hspost}@{}}%
\column{5}{@{}>{\hspre}l<{\hspost}@{}}%
\column{E}{@{}>{\hspre}l<{\hspost}@{}}%
\>[5]{}\Varid{fmap}\;\Varid{u}\;\texttt{(}\texttt{(}\Varid{g}\mathbin{\texttt{\color{red}{:\$:}}}\Varid{y}\texttt{)}\mathbin{\texttt{<*>}}\Varid{x}\texttt{)}{}\<[E]%
\\
\>[B]{}\equiv \mbox{\commentbegin  definition of \ensuremath{\texttt{(}\mathbin{\texttt{<*>}}\texttt{)}}  \commentend}{}\<[E]%
\\
\>[B]{}\hsindent{5}{}\<[5]%
\>[5]{}\Varid{fmap}\;\Varid{u}\;\texttt{(}\Varid{fmap}\;\Varid{uc}\;\Varid{g}\mathbin{\texttt{\color{red}{:\$:}}}\Varid{pair}\;\Varid{y}\;\Varid{x}\texttt{)}{}\<[E]%
\\
\>[B]{}\equiv \mbox{\commentbegin  definition of \ensuremath{\Varid{fmap}}  \commentend}{}\<[E]%
\\
\>[B]{}\hsindent{5}{}\<[5]%
\>[5]{}\Varid{fmap}\;\texttt{(}\Varid{u}\hsdot{\circ }{\texttt{.}}\texttt{)}\;\texttt{(}\Varid{fmap}\;\Varid{uc}\;\Varid{g}\texttt{)}\mathbin{\texttt{\color{red}{:\$:}}}\Varid{pair}\;\Varid{y}\;\Varid{x}{}\<[E]%
\\
\>[B]{}\equiv \mbox{\commentbegin  \ensuremath{\Varid{f}} is a functor  \commentend}{}\<[E]%
\\
\>[B]{}\hsindent{5}{}\<[5]%
\>[5]{}\Varid{fmap}\;\texttt{(}\lambda \Varid{g}\to \Varid{u}\hsdot{\circ }{\texttt{.}}\Varid{uc}\;\Varid{g}\texttt{)}\;\Varid{g}\mathbin{\texttt{\color{red}{:\$:}}}\Varid{pair}\;\Varid{y}\;\Varid{x}{}\<[E]%
\\
\>[B]{}\equiv \mbox{\commentbegin  \ensuremath{\Varid{f}} is a functor  \commentend}{}\<[E]%
\\
\>[B]{}\hsindent{5}{}\<[5]%
\>[5]{}\Varid{fmap}\;\Varid{uc}\;\texttt{(}\Varid{fmap}\;\texttt{(}\texttt{(}\Varid{u}\hsdot{\circ }{\texttt{.}}\texttt{)}\hsdot{\circ }{\texttt{.}}\texttt{)}\;\Varid{g}\texttt{)}\mathbin{\texttt{\color{red}{:\$:}}}\Varid{pair}\;\Varid{y}\;\Varid{x}{}\<[E]%
\\
\>[B]{}\equiv \mbox{\commentbegin  definition of \ensuremath{\texttt{(}\mathbin{\texttt{<*>}}\texttt{)}}  \commentend}{}\<[E]%
\\
\>[B]{}\hsindent{5}{}\<[5]%
\>[5]{}\texttt{(}\Varid{fmap}\;\texttt{(}\texttt{(}\Varid{u}\hsdot{\circ }{\texttt{.}}\texttt{)}\hsdot{\circ }{\texttt{.}}\texttt{)}\;\Varid{g}\mathbin{\texttt{\color{red}{:\$:}}}\Varid{y}\texttt{)}\mathbin{\texttt{<*>}}\Varid{x}{}\<[E]%
\\
\>[B]{}\equiv \mbox{\commentbegin  definition of \ensuremath{\Varid{fmap}}  \commentend}{}\<[E]%
\\
\>[B]{}\hsindent{5}{}\<[5]%
\>[5]{}\Varid{fmap}\;\texttt{(}\Varid{u}\hsdot{\circ }{\texttt{.}}\texttt{)}\;\texttt{(}\Varid{g}\mathbin{\texttt{\color{red}{:\$:}}}\Varid{y}\texttt{)}\mathbin{\texttt{<*>}}\Varid{x}{}\<[E]%
\ColumnHook
\end{hscode}\resethooks
\end{proof}

\begin{lem}\label{lem:app:comp}
Property \ref{app:comp} holds for \ensuremath{\Conid{FreeA}\;\Varid{f}}, i.e. for all
\begin{align*}
  \ensuremath{\Varid{u}} & \ensuremath{\mathbin{::}\Conid{FreeA}\;\Varid{f}\;\texttt{(}\Varid{y}\to \Varid{z}\texttt{)}} \\
  \ensuremath{\Varid{v}} & \ensuremath{\mathbin{::}\Conid{FreeA}\;\Varid{f}\;\texttt{(}\Varid{x}\to \Varid{y}\texttt{)}} \\
  \ensuremath{\Varid{x}} & \ensuremath{\mathbin{::}\Conid{FreeA}\;\Varid{f}\;\Varid{x}},
\end{align*}
\begin{equation*}
\ensuremath{\Varid{pure}\;\texttt{(}\hsdot{\circ }{\texttt{.}}\texttt{)}\mathbin{\texttt{<*>}}\Varid{u}\mathbin{\texttt{<*>}}\Varid{v}\mathbin{\texttt{<*>}}\Varid{x}} \equiv \ensuremath{\Varid{u}\mathbin{\texttt{<*>}}\texttt{(}\Varid{v}\mathbin{\texttt{<*>}}\Varid{x}\texttt{)}}
\end{equation*}
\end{lem}
\begin{proof}
Suppose first that \ensuremath{\Varid{u}\mathrel{=}\Conid{Pure}\;\mathtt{u}_\mathtt{0}} for some \ensuremath{\mathtt{u}_\mathtt{0}\mathbin{::}\Varid{y}\to \Varid{z}}:

\begin{hscode}\SaveRestoreHook
\column{B}{@{}>{\hspre}l<{\hspost}@{}}%
\column{5}{@{}>{\hspre}l<{\hspost}@{}}%
\column{E}{@{}>{\hspre}l<{\hspost}@{}}%
\>[5]{}\Conid{Pure}\;\texttt{(}\hsdot{\circ }{\texttt{.}}\texttt{)}\mathbin{\texttt{<*>}}\Conid{Pure}\;\mathtt{u}_\mathtt{0}\mathbin{\texttt{<*>}}\Varid{v}\mathbin{\texttt{<*>}}\Varid{x}{}\<[E]%
\\
\>[B]{}\equiv \mbox{\commentbegin  definition of \ensuremath{\texttt{(}\mathbin{\texttt{<*>}}\texttt{)}}  \commentend}{}\<[E]%
\\
\>[B]{}\hsindent{5}{}\<[5]%
\>[5]{}\Conid{Pure}\;\texttt{(}\mathtt{u}_\mathtt{0}\hsdot{\circ }{\texttt{.}}\texttt{)}\mathbin{\texttt{<*>}}\Varid{v}\mathbin{\texttt{<*>}}\Varid{x}{}\<[E]%
\\
\>[B]{}\equiv \mbox{\commentbegin  definition of \ensuremath{\texttt{(}\mathbin{\texttt{<*>}}\texttt{)}}  \commentend}{}\<[E]%
\\
\>[B]{}\hsindent{5}{}\<[5]%
\>[5]{}\Varid{fmap}\;\texttt{(}\mathtt{u}_\mathtt{0}\hsdot{\circ }{\texttt{.}}\texttt{)}\;\Varid{v}\mathbin{\texttt{<*>}}\Varid{x}{}\<[E]%
\\
\>[B]{}\equiv \mbox{\commentbegin  lemma \ref{lem:pure}  \commentend}{}\<[E]%
\\
\>[B]{}\hsindent{5}{}\<[5]%
\>[5]{}\Varid{fmap}\;\mathtt{u}_\mathtt{0}\;\texttt{(}\Varid{v}\mathbin{\texttt{<*>}}\Varid{x}\texttt{)}{}\<[E]%
\\
\>[B]{}\equiv \mbox{\commentbegin  definition of \ensuremath{\texttt{(}\mathbin{\texttt{<*>}}\texttt{)}}  \commentend}{}\<[E]%
\\
\>[B]{}\hsindent{5}{}\<[5]%
\>[5]{}\Conid{Pure}\;\mathtt{u}_\mathtt{0}\mathbin{\texttt{<*>}}\texttt{(}\Varid{v}\mathbin{\texttt{<*>}}\Varid{x}\texttt{)}{}\<[E]%
\ColumnHook
\end{hscode}\resethooks

To tackle the case where \ensuremath{\Varid{u}\mathrel{=}\Varid{g}\mathbin{\texttt{\color{red}{:\$:}}}\Varid{w}}, for
\begin{align*}
  \ensuremath{\Varid{g}} & \ensuremath{\mathbin{::}\Varid{f}\;\texttt{(}\Varid{w}\to \Varid{y}\to \Varid{z}\texttt{)}} \\
  \ensuremath{\Varid{w}} & \ensuremath{\mathbin{::}\Conid{FreeA}\;\Varid{f}\;\Varid{w}},
\end{align*}

we need to define a helper function

\begin{hscode}\SaveRestoreHook
\column{B}{@{}>{\hspre}l<{\hspost}@{}}%
\column{E}{@{}>{\hspre}l<{\hspost}@{}}%
\>[B]{}\Varid{t}\mathbin{::}\texttt{(}\texttt{(}\Varid{w}\mathbin{\texttt{,}}\Varid{x}\to \Varid{y}\texttt{)}\mathbin{\texttt{,}}\Varid{x}\texttt{)}\to \texttt{(}\Varid{w}\mathbin{\texttt{,}}\Varid{y}\texttt{)}{}\<[E]%
\\
\>[B]{}\Varid{t}\;\texttt{(}\texttt{(}\Varid{w}\mathbin{\texttt{,}}\Varid{v}\texttt{)}\mathbin{\texttt{,}}\Varid{x}\texttt{)}\mathrel{=}\texttt{(}\Varid{w}\mathbin{\texttt{,}}\Varid{v}\;\Varid{x}\texttt{)}{}\<[E]%
\ColumnHook
\end{hscode}\resethooks

and compute:

\begin{hscode}\SaveRestoreHook
\column{B}{@{}>{\hspre}l<{\hspost}@{}}%
\column{5}{@{}>{\hspre}l<{\hspost}@{}}%
\column{E}{@{}>{\hspre}l<{\hspost}@{}}%
\>[5]{}\Varid{pure}\;\texttt{(}\hsdot{\circ }{\texttt{.}}\texttt{)}\mathbin{\texttt{<*>}}\texttt{(}\Varid{g}\mathbin{\texttt{\color{red}{:\$:}}}\Varid{w}\texttt{)}\mathbin{\texttt{<*>}}\Varid{v}\mathbin{\texttt{<*>}}\Varid{x}{}\<[E]%
\\
\>[B]{}\equiv \mbox{\commentbegin  definition of \ensuremath{\Varid{pure}} and \ensuremath{\texttt{(}\mathbin{\texttt{<*>}}\texttt{)}}  \commentend}{}\<[E]%
\\
\>[B]{}\hsindent{5}{}\<[5]%
\>[5]{}\texttt{(}\Varid{fmap}\;\texttt{(}\texttt{(}\hsdot{\circ }{\texttt{.}}\texttt{)}\hsdot{\circ }{\texttt{.}}\texttt{)}\;\Varid{g}\mathbin{\texttt{\color{red}{:\$:}}}\Varid{w}\texttt{)}\mathbin{\texttt{<*>}}\Varid{v}\mathbin{\texttt{<*>}}\Varid{x}{}\<[E]%
\\
\>[B]{}\equiv \mbox{\commentbegin  definition of composition  \commentend}{}\<[E]%
\\
\>[B]{}\hsindent{5}{}\<[5]%
\>[5]{}\texttt{(}\Varid{fmap}\;\texttt{(}\lambda \Varid{g}\;\Varid{w}\;\Varid{v}\to \Varid{g}\;\Varid{w}\hsdot{\circ }{\texttt{.}}\Varid{v}\texttt{)}\;\Varid{g}\mathbin{\texttt{\color{red}{:\$:}}}\Varid{w}\texttt{)}\mathbin{\texttt{<*>}}\Varid{v}\mathbin{\texttt{<*>}}\Varid{x}{}\<[E]%
\\
\>[B]{}\equiv \mbox{\commentbegin  definition of \ensuremath{\texttt{(}\mathbin{\texttt{<*>}}\texttt{)}}  \commentend}{}\<[E]%
\\
\>[B]{}\hsindent{5}{}\<[5]%
\>[5]{}\texttt{(}\Varid{fmap}\;\Varid{uc}\;\texttt{(}\Varid{fmap}\;\texttt{(}\lambda \Varid{g}\;\Varid{w}\;\Varid{v}\to \Varid{g}\;\Varid{w}\hsdot{\circ }{\texttt{.}}\Varid{v}\texttt{)}\;\Varid{g}\texttt{)}\mathbin{\texttt{\color{red}{:\$:}}}\Varid{pair}\;\Varid{w}\;\Varid{v}\texttt{)}{}\<[E]%
\\
\>[B]{}\hsindent{5}{}\<[5]%
\>[5]{}\mathbin{\texttt{<*>}}\Varid{x}{}\<[E]%
\\
\>[B]{}\equiv \mbox{\commentbegin  \ensuremath{\Varid{f}} is a functor and definition of \ensuremath{\Varid{uc}}  \commentend}{}\<[E]%
\\
\>[B]{}\hsindent{5}{}\<[5]%
\>[5]{}\texttt{(}\Varid{fmap}\;\texttt{(}\lambda \Varid{g}\;\texttt{(}\Varid{w}\mathbin{\texttt{,}}\Varid{v}\texttt{)}\to \Varid{g}\;\Varid{w}\hsdot{\circ }{\texttt{.}}\Varid{v}\texttt{)}\;\Varid{g}\mathbin{\texttt{\color{red}{:\$:}}}\Varid{pair}\;\Varid{w}\;\Varid{v}\texttt{)}\mathbin{\texttt{<*>}}\Varid{x}{}\<[E]%
\\
\>[B]{}\equiv \mbox{\commentbegin  definition of \ensuremath{\texttt{(}\mathbin{\texttt{<*>}}\texttt{)}}  \commentend}{}\<[E]%
\\
\>[B]{}\hsindent{5}{}\<[5]%
\>[5]{}\Varid{fmap}\;\Varid{uc}\;\texttt{(}\Varid{fmap}\;\texttt{(}\lambda \Varid{g}\;\texttt{(}\Varid{w}\mathbin{\texttt{,}}\Varid{v}\texttt{)}\to \Varid{g}\;\Varid{w}\hsdot{\circ }{\texttt{.}}\Varid{v}\texttt{)}\;\Varid{g}\texttt{)}\mathbin{\texttt{\color{red}{:\$:}}}{}\<[E]%
\\
\>[B]{}\hsindent{5}{}\<[5]%
\>[5]{}\Varid{pair}\;\texttt{(}\Varid{pair}\;\Varid{w}\;\Varid{v}\texttt{)}\;\Varid{x}{}\<[E]%
\\
\>[B]{}\equiv \mbox{\commentbegin  \ensuremath{\Varid{f}} is a functor and definition of \ensuremath{\Varid{uc}}  \commentend}{}\<[E]%
\\
\>[B]{}\hsindent{5}{}\<[5]%
\>[5]{}\Varid{fmap}\;\texttt{(}\lambda \Varid{g}\;\texttt{(}\texttt{(}\Varid{w}\mathbin{\texttt{,}}\Varid{v}\texttt{)}\mathbin{\texttt{,}}\Varid{x}\texttt{)}\to \Varid{g}\;\Varid{w}\;\texttt{(}\Varid{v}\;\Varid{x}\texttt{)}\texttt{)}\;\Varid{g}\mathbin{\texttt{\color{red}{:\$:}}}{}\<[E]%
\\
\>[B]{}\hsindent{5}{}\<[5]%
\>[5]{}\Varid{pair}\;\texttt{(}\Varid{pair}\;\Varid{w}\;\Varid{v}\texttt{)}\;\Varid{x}{}\<[E]%
\\
\>[B]{}\equiv \mbox{\commentbegin  definition of \ensuremath{\Varid{uc}} and \ensuremath{\Varid{t}}  \commentend}{}\<[E]%
\\
\>[B]{}\hsindent{5}{}\<[5]%
\>[5]{}\Varid{fmap}\;\texttt{(}\lambda \Varid{g}\to \Varid{uc}\;\Varid{g}\hsdot{\circ }{\texttt{.}}\Varid{t}\texttt{)}\;\Varid{g}\mathbin{\texttt{\color{red}{:\$:}}}\Varid{pair}\;\texttt{(}\Varid{pair}\;\Varid{w}\;\Varid{v}\texttt{)}\;\Varid{x}{}\<[E]%
\\
\>[B]{}\equiv \mbox{\commentbegin  \ensuremath{\Varid{f}} is a functor  \commentend}{}\<[E]%
\\
\>[B]{}\hsindent{5}{}\<[5]%
\>[5]{}\Varid{fmap}\;\texttt{(}\hsdot{\circ }{\texttt{.}}\Varid{t}\texttt{)}\;\texttt{(}\Varid{fmap}\;\Varid{uc}\;\Varid{g}\texttt{)}\mathbin{\texttt{\color{red}{:\$:}}}\Varid{pair}\;\texttt{(}\Varid{pair}\;\Varid{w}\;\Varid{v}\texttt{)}\;\Varid{x}{}\<[E]%
\\
\>[B]{}\equiv \mbox{\commentbegin  equation \ref{naturality}  \commentend}{}\<[E]%
\\
\>[B]{}\hsindent{5}{}\<[5]%
\>[5]{}\Varid{fmap}\;\Varid{uc}\;\Varid{g}\mathbin{\texttt{\color{red}{:\$:}}}\Varid{fmap}\;\Varid{t}\;\texttt{(}\Varid{pair}\;\texttt{(}\Varid{pair}\;\Varid{w}\;\Varid{v}\texttt{)}\;\Varid{x}\texttt{)}{}\<[E]%
\\
\>[B]{}\equiv \mbox{\commentbegin  lemma \ref{lem:pure} (3 times) and \ensuremath{\Conid{FreeA}\;\Varid{f}} is a functor (3 times)  \commentend}{}\<[E]%
\\
\>[B]{}\hsindent{5}{}\<[5]%
\>[5]{}\Varid{fmap}\;\Varid{uc}\;\Varid{g}\mathbin{\texttt{\color{red}{:\$:}}}\texttt{(}\Varid{pure}\;\texttt{(}\hsdot{\circ }{\texttt{.}}\texttt{)}\mathbin{\texttt{<*>}}\Varid{fmap}\;\texttt{(}\mathbin{\texttt{,}}\texttt{)}\;\Varid{w}\mathbin{\texttt{<*>}}\Varid{v}\mathbin{\texttt{<*>}}\Varid{x}\texttt{)}{}\<[E]%
\\
\>[B]{}\equiv \mbox{\commentbegin  induction hypothesis for \ensuremath{\Varid{fmap}\;\texttt{(}\mathbin{\texttt{,}}\texttt{)}\;\Varid{w}}  \commentend}{}\<[E]%
\\
\>[B]{}\hsindent{5}{}\<[5]%
\>[5]{}\Varid{fmap}\;\Varid{uc}\;\Varid{g}\mathbin{\texttt{\color{red}{:\$:}}}\texttt{(}\Varid{fmap}\;\texttt{(}\mathbin{\texttt{,}}\texttt{)}\;\Varid{w}\mathbin{\texttt{<*>}}\texttt{(}\Varid{v}\mathbin{\texttt{<*>}}\Varid{x}\texttt{)}\texttt{)}{}\<[E]%
\\
\>[B]{}\equiv \mbox{\commentbegin  definition of \ensuremath{\texttt{(}\mathbin{\texttt{<*>}}\texttt{)}}  \commentend}{}\<[E]%
\\
\>[B]{}\hsindent{5}{}\<[5]%
\>[5]{}\texttt{(}\Varid{g}\mathbin{\texttt{\color{red}{:\$:}}}\Varid{w}\texttt{)}\mathbin{\texttt{<*>}}\texttt{(}\Varid{v}\mathbin{\texttt{<*>}}\Varid{x}\texttt{)}{}\<[E]%
\ColumnHook
\end{hscode}\resethooks
\end{proof}

\begin{lem}\label{lem:app:interchange}
Property \ref{app:interchange} holds for \ensuremath{\Conid{FreeA}\;\Varid{f}}, i.e. for all
\begin{align*}
  \ensuremath{\Varid{u}} & \ensuremath{\mathbin{::}\Conid{FreeA}\;\Varid{f}\;\texttt{(}\Varid{x}\to \Varid{y}\texttt{)}} \\
  \ensuremath{\Varid{x}} & \ensuremath{\mathbin{::}\Varid{x}},
\end{align*}
\begin{equation*}
\ensuremath{\Varid{u}\mathbin{\texttt{<*>}}\Varid{pure}\;\Varid{x}} \equiv \ensuremath{\Varid{pure}\;\texttt{(}\mathbin{\texttt{\$}}\Varid{x}\texttt{)}\mathbin{\texttt{<*>}}\Varid{u}}
\end{equation*}
\end{lem}
\begin{proof}
If \ensuremath{\Varid{u}} is of the form \ensuremath{\Conid{Pure}\;\mathtt{u}_\mathtt{0}}, then the conclusion follows immediately.

Let's assume, therefore, that \ensuremath{\Varid{u}\mathrel{=}\Varid{g}\mathbin{\texttt{\color{red}{:\$:}}}\Varid{w}}, for some \ensuremath{\Varid{w}\mathbin{::}\Varid{w}}, \ensuremath{\Varid{g}\mathbin{::}\Varid{f}\;\texttt{(}\Varid{w}\to \Varid{x}\to \Varid{y}\texttt{)}}, and that the lemma is true for structurally smaller values of \ensuremath{\Varid{u}}:

\begin{hscode}\SaveRestoreHook
\column{B}{@{}>{\hspre}l<{\hspost}@{}}%
\column{5}{@{}>{\hspre}l<{\hspost}@{}}%
\column{E}{@{}>{\hspre}l<{\hspost}@{}}%
\>[5]{}\texttt{(}\Varid{g}\mathbin{\texttt{\color{red}{:\$:}}}\Varid{w}\texttt{)}\mathbin{\texttt{<*>}}\Varid{pure}\;\Varid{x}{}\<[E]%
\\
\>[B]{}\equiv \mbox{\commentbegin  definition of \ensuremath{\texttt{(}\mathbin{\texttt{<*>}}\texttt{)}}  \commentend}{}\<[E]%
\\
\>[B]{}\hsindent{5}{}\<[5]%
\>[5]{}\Varid{fmap}\;\Varid{uc}\;\Varid{g}\mathbin{\texttt{\color{red}{:\$:}}}\Varid{pair}\;\Varid{w}\;\texttt{(}\Varid{pure}\;\Varid{x}\texttt{)}{}\<[E]%
\\
\>[B]{}\equiv \mbox{\commentbegin  definition of \ensuremath{\Varid{pair}}  \commentend}{}\<[E]%
\\
\>[B]{}\hsindent{5}{}\<[5]%
\>[5]{}\Varid{fmap}\;\Varid{uc}\;\Varid{g}\mathbin{\texttt{\color{red}{:\$:}}}\texttt{(}\Varid{fmap}\;\texttt{(}\mathbin{\texttt{,}}\texttt{)}\;\Varid{w}\mathbin{\texttt{<*>}}\Varid{pure}\;\Varid{x}\texttt{)}{}\<[E]%
\\
\>[B]{}\equiv \mbox{\commentbegin  induction hypothesis for \ensuremath{\Varid{fmap}\;\texttt{(}\mathbin{\texttt{,}}\texttt{)}\;\Varid{w}}  \commentend}{}\<[E]%
\\
\>[B]{}\hsindent{5}{}\<[5]%
\>[5]{}\Varid{fmap}\;\Varid{uc}\;\Varid{g}\mathbin{\texttt{\color{red}{:\$:}}}\texttt{(}\Varid{pure}\;\texttt{(}\mathbin{\texttt{\$}}\Varid{x}\texttt{)}\mathbin{\texttt{<*>}}\Varid{fmap}\;\texttt{(}\mathbin{\texttt{,}}\texttt{)}\;\Varid{w}\texttt{)}{}\<[E]%
\\
\>[B]{}\equiv \mbox{\commentbegin  \ensuremath{\Conid{FreeA}\;\Varid{f}} is a functor  \commentend}{}\<[E]%
\\
\>[B]{}\hsindent{5}{}\<[5]%
\>[5]{}\Varid{fmap}\;\Varid{uc}\;\Varid{g}\mathbin{\texttt{\color{red}{:\$:}}}\Varid{fmap}\;\texttt{(}\lambda \Varid{w}\to \texttt{(}\Varid{w}\mathbin{\texttt{,}}\Varid{x}\texttt{)}\texttt{)}\;\Varid{w}\texttt{)}{}\<[E]%
\\
\>[B]{}\equiv \mbox{\commentbegin  equation \ref{naturality}  \commentend}{}\<[E]%
\\
\>[B]{}\hsindent{5}{}\<[5]%
\>[5]{}\Varid{fmap}\;\texttt{(}\lambda \Varid{g}\;\Varid{w}\to \Varid{g}\;\texttt{(}\Varid{w}\mathbin{\texttt{,}}\Varid{x}\texttt{)}\texttt{)}\;\texttt{(}\Varid{fmap}\;\Varid{uc}\;\Varid{g}\texttt{)}\mathbin{\texttt{\color{red}{:\$:}}}\Varid{w}{}\<[E]%
\\
\>[B]{}\equiv \mbox{\commentbegin  \ensuremath{\Varid{f}} is a functor  \commentend}{}\<[E]%
\\
\>[B]{}\hsindent{5}{}\<[5]%
\>[5]{}\Varid{fmap}\;\texttt{(}\lambda \Varid{g}\;\Varid{w}\to \Varid{g}\;\Varid{w}\;\Varid{x}\texttt{)}\;\Varid{g}\mathbin{\texttt{\color{red}{:\$:}}}\Varid{w}{}\<[E]%
\\
\>[B]{}\equiv \mbox{\commentbegin  definition of \ensuremath{\Varid{fmap}} for \ensuremath{\Conid{FreeA}\;\Varid{f}}  \commentend}{}\<[E]%
\\
\>[B]{}\hsindent{5}{}\<[5]%
\>[5]{}\Varid{fmap}\;\texttt{(}\mathbin{\texttt{\$}}\Varid{x}\texttt{)}\;\texttt{(}\Varid{g}\mathbin{\texttt{\color{red}{:\$:}}}\Varid{w}\texttt{)}{}\<[E]%
\\
\>[B]{}\equiv \mbox{\commentbegin  definition of \ensuremath{\texttt{(}\mathbin{\texttt{<*>}}\texttt{)}}  \commentend}{}\<[E]%
\\
\>[B]{}\hsindent{5}{}\<[5]%
\>[5]{}\Varid{pure}\;\texttt{(}\mathbin{\texttt{\$}}\Varid{x}\texttt{)}\mathbin{\texttt{<*>}}\texttt{(}\Varid{g}\mathbin{\texttt{\color{red}{:\$:}}}\Varid{w}\texttt{)}{}\<[E]%
\ColumnHook
\end{hscode}\resethooks
\end{proof}

\begin{prop}
\ensuremath{\Conid{FreeA}\;\Varid{f}} is an applicative functor.
\end{prop}
\begin{proof}
Properties \ref{app:id} and \ref{app:hom} are straightforward to verify using
the fact that \ensuremath{\Conid{FreeA}\;\Varid{f}} is a functor, while properties \ref{app:comp} and
\ref{app:interchange} follow from lemmas \ref{lem:app:comp} and
\ref{lem:app:interchange} respectively.
\end{proof}

\section{\text{\tt FreeA} as a Left adjoint}
\label{sec:adjoint}
We are now going to make the statement that \ensuremath{\Conid{FreeA}\;\Varid{f}} is the free applicative
functor on \ensuremath{\Varid{f}} precise.

First of all, we will define a category $\mathcal{A}$ of applicative functors,
and show that \ensuremath{\Conid{FreeA}} is a functor
$$
\ensuremath{\Conid{FreeA}} : \mathcal{F} \to \mathcal{A},
$$
where $\mathcal{F}$ is the category of endofunctors of \ensuremath{\Conid{Hask}}.

Saying that \ensuremath{\Conid{FreeA}\;\Varid{f}} is the free applicative on \ensuremath{\Varid{f}}, then, amounts to saying
that \ensuremath{\Conid{FreeA}} is left adjoint to the forgetful functor $\mathcal{A} \to
\mathcal{F}$.

\begin{defn}
Let \ensuremath{\Varid{f}} and \ensuremath{\Varid{g}} be two applicative functors. An applicative natural
transformation between \ensuremath{\Varid{f}} and \ensuremath{\Varid{g}} is a polymorphic function
$$
\ensuremath{\Varid{t}\mathbin{::}\forall \Varid{a}\hsforall \hsdot{\circ }{\texttt{.}}\Varid{f}\;\Varid{a}\to \Varid{g}\;\Varid{a}}
$$
satisfying the following laws:
\begin{align}
  \ensuremath{\Varid{t}\;\texttt{(}\Varid{pure}\;\Varid{x}\texttt{)}} & \equiv \ensuremath{\Varid{pure}\;\Varid{x}} \label{eq:applicative1}\\
  \ensuremath{\Varid{t}\;\texttt{(}\Varid{h}\mathbin{\texttt{<*>}}\Varid{x}\texttt{)}} & \equiv \ensuremath{\Varid{t}\;\Varid{h}\mathbin{\texttt{<*>}}\Varid{t}\;\Varid{x}}. \label{eq:applicative2}
\end{align}
\end{defn}

We define the type of all applicative natural transformations between \ensuremath{\Varid{f}}
and \ensuremath{\Varid{g}}, we write, in Haskell,

\begin{hscode}\SaveRestoreHook
\column{B}{@{}>{\hspre}l<{\hspost}@{}}%
\column{E}{@{}>{\hspre}l<{\hspost}@{}}%
\>[B]{}\texttt{\color{blue}\textbf{type}}\;\Conid{AppNat}\;\Varid{f}\;\Varid{g}\mathrel{=}\forall \Varid{a}\hsforall \hsdot{\circ }{\texttt{.}}\Varid{f}\;\Varid{a}\to \Varid{g}\;\Varid{a}{}\<[E]%
\ColumnHook
\end{hscode}\resethooks

where the laws are implied.

Similarly, for any pair of functors \ensuremath{\Varid{f}} and \ensuremath{\Varid{g}}, we define

\begin{hscode}\SaveRestoreHook
\column{B}{@{}>{\hspre}l<{\hspost}@{}}%
\column{E}{@{}>{\hspre}l<{\hspost}@{}}%
\>[B]{}\texttt{\color{blue}\textbf{type}}\;\Conid{Nat}\;\Varid{f}\;\Varid{g}\mathrel{=}\forall \Varid{a}\hsforall \hsdot{\circ }{\texttt{.}}\Varid{f}\;\Varid{a}\to \Varid{g}\;\Varid{a}{}\<[E]%
\ColumnHook
\end{hscode}\resethooks

for the type of natural transformations between \ensuremath{\Varid{f}} and \ensuremath{\Varid{g}}.

Note that, by parametricity, polymorphic functions are automatically natural
transformations in the categorical sense, i.e, for all
\begin{align*}
  \ensuremath{\Varid{t}} & \ensuremath{\mathbin{::}\Conid{Nat}\;\Varid{f}\;\Varid{g}} \\
  \ensuremath{\Varid{h}} & \ensuremath{\mathbin{::}\Varid{a}\to \Varid{b}} \\
  \ensuremath{\Varid{x}} & \ensuremath{\mathbin{::}\Varid{f}\;\Varid{a}},
\end{align*}
\begin{equation*}
    \ensuremath{\Varid{t}\;\texttt{(}\Varid{fmap}\;\Varid{h}\;\Varid{x}\texttt{)}} \equiv \ensuremath{\Varid{fmap}\;\Varid{h}\;\texttt{(}\Varid{t}\;\Varid{x}\texttt{)}}.
\end{equation*}

It is clear that applicative functors, together with applicative natural
transformations, form a category, which we denote by $\mathcal{A}$, and
similarly, functors and natural transformations form a category $\mathcal{F}$.

\begin{prop}\label{prop:FreeAFunctor}
\ensuremath{\Conid{FreeA}} defines a functor $\mathcal{F} \to \mathcal{A}$.
\end{prop}
\begin{proof}
We already showed that \ensuremath{\Conid{FreeA}} sends objects (functors in our case) to
applicative functors.

We need to define the action of \ensuremath{\Conid{FreeA}} on morphisms (which are natural
transformations in our case):

\begin{hscode}\SaveRestoreHook
\column{B}{@{}>{\hspre}l<{\hspost}@{}}%
\column{8}{@{}>{\hspre}l<{\hspost}@{}}%
\column{10}{@{}>{\hspre}l<{\hspost}@{}}%
\column{21}{@{}>{\hspre}l<{\hspost}@{}}%
\column{E}{@{}>{\hspre}l<{\hspost}@{}}%
\>[B]{}\Varid{liftT}{}\<[8]%
\>[8]{}\mathbin{::}\texttt{(}\Conid{Functor}\;\Varid{f}\mathbin{\texttt{,}}\Conid{Functor}\;\Varid{g}\texttt{)}{}\<[E]%
\\
\>[8]{}\Rightarrow \Conid{Nat}\;\Varid{f}\;\Varid{g}{}\<[E]%
\\
\>[8]{}\to \Conid{AppNat}\;\texttt{(}\Conid{FreeA}\;\Varid{f}\texttt{)}\;\texttt{(}\Conid{FreeA}\;\Varid{g}\texttt{)}{}\<[E]%
\\
\>[B]{}\Varid{liftT}\;\anonymous \;{}\<[10]%
\>[10]{}\texttt{(}\Conid{Pure}\;\Varid{x}\texttt{)}{}\<[21]%
\>[21]{}\mathrel{=}\Conid{Pure}\;\Varid{x}{}\<[E]%
\\
\>[B]{}\Varid{liftT}\;\Varid{k}\;{}\<[10]%
\>[10]{}\texttt{(}\Varid{h}\mathbin{\texttt{\color{red}{:\$:}}}\Varid{x}\texttt{)}{}\<[21]%
\>[21]{}\mathrel{=}\Varid{k}\;\Varid{h}\mathbin{\texttt{\color{red}{:\$:}}}\Varid{liftT}\;\Varid{k}\;\Varid{x}{}\<[E]%
\ColumnHook
\end{hscode}\resethooks

First we verify that \ensuremath{\Varid{liftT}\;\Varid{k}} is an applicative natural transformation i.e.
it satisfies laws \ref{eq:applicative1} and \ref{eq:applicative2}. We use
equational reasoning for proving law \ref{eq:applicative1}:

\begin{hscode}\SaveRestoreHook
\column{B}{@{}>{\hspre}l<{\hspost}@{}}%
\column{7}{@{}>{\hspre}l<{\hspost}@{}}%
\column{E}{@{}>{\hspre}l<{\hspost}@{}}%
\>[7]{}\Varid{liftT}\;\Varid{k}\;\texttt{(}\Varid{pure}\;\Varid{x}\texttt{)}{}\<[E]%
\\
\>[B]{}\equiv \mbox{\commentbegin  definition of \ensuremath{\Varid{pure}}  \commentend}{}\<[E]%
\\
\>[B]{}\hsindent{7}{}\<[7]%
\>[7]{}\Varid{liftT}\;\Varid{k}\;\texttt{(}\Conid{Pure}\;\Varid{x}\texttt{)}{}\<[E]%
\\
\>[B]{}\equiv \mbox{\commentbegin  definition of \ensuremath{\Varid{liftT}}  \commentend}{}\<[E]%
\\
\>[B]{}\hsindent{7}{}\<[7]%
\>[7]{}\Conid{Pure}\;\Varid{x}{}\<[E]%
\\
\>[B]{}\equiv \mbox{\commentbegin  definition of \ensuremath{\Varid{pure}}  \commentend}{}\<[E]%
\\
\>[B]{}\hsindent{7}{}\<[7]%
\>[7]{}\Varid{pure}\;\Varid{x}{}\<[E]%
\ColumnHook
\end{hscode}\resethooks

For law \ref{eq:applicative2} we use induction on the size of the
first argument of \ensuremath{\texttt{(}\mathbin{\texttt{<*>}}\texttt{)}} as explained in section \ref{sec:totality}.
The base cases:

\begin{hscode}\SaveRestoreHook
\column{B}{@{}>{\hspre}l<{\hspost}@{}}%
\column{7}{@{}>{\hspre}l<{\hspost}@{}}%
\column{E}{@{}>{\hspre}l<{\hspost}@{}}%
\>[7]{}\Varid{liftT}\;\Varid{k}\;\texttt{(}\Conid{Pure}\;\Varid{h}\mathbin{\texttt{<*>}}\Conid{Pure}\;\Varid{x}\texttt{)}{}\<[E]%
\\
\>[B]{}\equiv \mbox{\commentbegin  definition of \ensuremath{\texttt{(}\mathbin{\texttt{<*>}}\texttt{)}}  \commentend}{}\<[E]%
\\
\>[B]{}\hsindent{7}{}\<[7]%
\>[7]{}\Varid{liftT}\;\Varid{k}\;\texttt{(}\Varid{fmap}\;\Varid{h}\;\texttt{(}\Conid{Pure}\;\Varid{x}\texttt{)}\texttt{)}{}\<[E]%
\\
\>[B]{}\equiv \mbox{\commentbegin  definition of \ensuremath{\Varid{fmap}}  \commentend}{}\<[E]%
\\
\>[B]{}\hsindent{7}{}\<[7]%
\>[7]{}\Varid{liftT}\;\Varid{k}\;\texttt{(}\Conid{Pure}\;\texttt{(}\Varid{h}\;\Varid{x}\texttt{)}\texttt{)}{}\<[E]%
\\
\>[B]{}\equiv \mbox{\commentbegin  definition of \ensuremath{\Varid{liftT}}  \commentend}{}\<[E]%
\\
\>[B]{}\hsindent{7}{}\<[7]%
\>[7]{}\Conid{Pure}\;\texttt{(}\Varid{h}\;\Varid{x}\texttt{)}{}\<[E]%
\\
\>[B]{}\equiv \mbox{\commentbegin  definition of \ensuremath{\Varid{fmap}}  \commentend}{}\<[E]%
\\
\>[B]{}\hsindent{7}{}\<[7]%
\>[7]{}\Varid{fmap}\;\Varid{h}\;\texttt{(}\Conid{Pure}\;\Varid{x}\texttt{)}{}\<[E]%
\\
\>[B]{}\equiv \mbox{\commentbegin  definition of \ensuremath{\texttt{(}\mathbin{\texttt{<*>}}\texttt{)}}  \commentend}{}\<[E]%
\\
\>[B]{}\hsindent{7}{}\<[7]%
\>[7]{}\Conid{Pure}\;\Varid{h}\mathbin{\texttt{<*>}}\Conid{Pure}\;\Varid{x}{}\<[E]%
\\
\>[B]{}\equiv \mbox{\commentbegin  definition of \ensuremath{\Varid{liftT}}  \commentend}{}\<[E]%
\\
\>[B]{}\hsindent{7}{}\<[7]%
\>[7]{}\Varid{liftT}\;\Varid{k}\;\texttt{(}\Conid{Pure}\;\Varid{h}\texttt{)}\mathbin{\texttt{<*>}}\Varid{liftT}\;\Varid{k}\;\texttt{(}\Conid{Pure}\;\Varid{x}\texttt{)}{}\<[E]%
\ColumnHook
\end{hscode}\resethooks

\begin{hscode}\SaveRestoreHook
\column{B}{@{}>{\hspre}l<{\hspost}@{}}%
\column{7}{@{}>{\hspre}l<{\hspost}@{}}%
\column{E}{@{}>{\hspre}l<{\hspost}@{}}%
\>[7]{}\Varid{liftT}\;\Varid{k}\;\texttt{(}\Conid{Pure}\;\Varid{h}\mathbin{\texttt{<*>}}\texttt{(}\Varid{i}\mathbin{\texttt{\color{red}{:\$:}}}\Varid{x}\texttt{)}\texttt{)}{}\<[E]%
\\
\>[B]{}\equiv \mbox{\commentbegin  definition of \ensuremath{\texttt{(}\mathbin{\texttt{<*>}}\texttt{)}}  \commentend}{}\<[E]%
\\
\>[B]{}\hsindent{7}{}\<[7]%
\>[7]{}\Varid{liftT}\;\Varid{k}\;\texttt{(}\Varid{fmap}\;\Varid{h}\;\texttt{(}\Varid{i}\mathbin{\texttt{\color{red}{:\$:}}}\Varid{x}\texttt{)}\texttt{)}{}\<[E]%
\\
\>[B]{}\equiv \mbox{\commentbegin  definition of \ensuremath{\Varid{fmap}}  \commentend}{}\<[E]%
\\
\>[B]{}\hsindent{7}{}\<[7]%
\>[7]{}\Varid{liftT}\;\Varid{k}\;\texttt{(}\Varid{fmap}\;\texttt{(}\Varid{h}\hsdot{\circ }{\texttt{.}}\texttt{)}\;\Varid{i}\mathbin{\texttt{\color{red}{:\$:}}}\Varid{x}\texttt{)}{}\<[E]%
\\
\>[B]{}\equiv \mbox{\commentbegin  definition of \ensuremath{\Varid{liftT}}  \commentend}{}\<[E]%
\\
\>[B]{}\hsindent{7}{}\<[7]%
\>[7]{}\Varid{k}\;\texttt{(}\Varid{fmap}\;\texttt{(}\Varid{h}\hsdot{\circ }{\texttt{.}}\texttt{)}\;\Varid{i}\texttt{)}\mathbin{\texttt{\color{red}{:\$:}}}\Varid{liftT}\;\Varid{k}\;\Varid{x}{}\<[E]%
\\
\>[B]{}\equiv \mbox{\commentbegin  \ensuremath{\Varid{k}} is natural  \commentend}{}\<[E]%
\\
\>[B]{}\hsindent{7}{}\<[7]%
\>[7]{}\Varid{fmap}\;\texttt{(}\Varid{h}\hsdot{\circ }{\texttt{.}}\texttt{)}\;\texttt{(}\Varid{k}\;\Varid{i}\texttt{)}\mathbin{\texttt{\color{red}{:\$:}}}\Varid{liftT}\;\Varid{k}\;\Varid{x}{}\<[E]%
\\
\>[B]{}\equiv \mbox{\commentbegin  definition of \ensuremath{\Varid{fmap}}  \commentend}{}\<[E]%
\\
\>[B]{}\hsindent{7}{}\<[7]%
\>[7]{}\Varid{fmap}\;\Varid{h}\;\texttt{(}\Varid{k}\;\Varid{i}\mathbin{\texttt{\color{red}{:\$:}}}\Varid{liftT}\;\Varid{k}\;\Varid{x}\texttt{)}{}\<[E]%
\\
\>[B]{}\equiv \mbox{\commentbegin  definition of \ensuremath{\texttt{(}\mathbin{\texttt{<*>}}\texttt{)}}  \commentend}{}\<[E]%
\\
\>[B]{}\hsindent{7}{}\<[7]%
\>[7]{}\Conid{Pure}\;\Varid{h}\mathbin{\texttt{<*>}}\texttt{(}\Varid{k}\;\Varid{i}\mathbin{\texttt{\color{red}{:\$:}}}\Varid{liftT}\;\Varid{k}\;\Varid{x}\texttt{)}{}\<[E]%
\\
\>[B]{}\equiv \mbox{\commentbegin  definition of \ensuremath{\Varid{liftT}}  \commentend}{}\<[E]%
\\
\>[B]{}\hsindent{7}{}\<[7]%
\>[7]{}\Varid{liftT}\;\Varid{k}\;\texttt{(}\Conid{Pure}\;\Varid{h}\texttt{)}\mathbin{\texttt{<*>}}\Varid{liftT}\;\Varid{k}\;\texttt{(}\Varid{i}\mathbin{\texttt{\color{red}{:\$:}}}\Varid{x}\texttt{)}{}\<[E]%
\ColumnHook
\end{hscode}\resethooks

The inductive case:

\begin{hscode}\SaveRestoreHook
\column{B}{@{}>{\hspre}l<{\hspost}@{}}%
\column{7}{@{}>{\hspre}l<{\hspost}@{}}%
\column{E}{@{}>{\hspre}l<{\hspost}@{}}%
\>[7]{}\Varid{liftT}\;\Varid{k}\;\texttt{(}\texttt{(}\Varid{h}\mathbin{\texttt{\color{red}{:\$:}}}\Varid{x}\texttt{)}\mathbin{\texttt{<*>}}\Varid{y}\texttt{)}{}\<[E]%
\\
\>[B]{}\equiv \mbox{\commentbegin  definition of \ensuremath{\texttt{(}\mathbin{\texttt{<*>}}\texttt{)}}  \commentend}{}\<[E]%
\\
\>[B]{}\hsindent{7}{}\<[7]%
\>[7]{}\Varid{liftT}\;\Varid{k}\;\texttt{(}\Varid{fmap}\;\Varid{uncurry}\;\Varid{h}\mathbin{\texttt{\color{red}{:\$:}}}\texttt{(}\Varid{fmap}\;\texttt{(}\mathbin{\texttt{,}}\texttt{)}\;\Varid{x}\mathbin{\texttt{<*>}}\Varid{y}\texttt{)}{}\<[E]%
\\
\>[B]{}\equiv \mbox{\commentbegin  definition of \ensuremath{\Varid{liftT}}  \commentend}{}\<[E]%
\\
\>[B]{}\hsindent{7}{}\<[7]%
\>[7]{}\Varid{k}\;\texttt{(}\Varid{fmap}\;\Varid{uncurry}\;\Varid{h}\texttt{)}\mathbin{\texttt{\color{red}{:\$:}}}\Varid{liftT}\;\Varid{k}\;\texttt{(}\Varid{fmap}\;\texttt{(}\mathbin{\texttt{,}}\texttt{)}\;\Varid{x}\mathbin{\texttt{<*>}}\Varid{y}\texttt{)}{}\<[E]%
\\
\>[B]{}\equiv \mbox{\commentbegin  inductive hypothesis  \commentend}{}\<[E]%
\\
\>[B]{}\hsindent{7}{}\<[7]%
\>[7]{}\Varid{k}\;\texttt{(}\Varid{fmap}\;\Varid{uncurry}\;\Varid{h}\texttt{)}\mathbin{\texttt{\color{red}{:\$:}}}{}\<[E]%
\\
\>[B]{}\hsindent{7}{}\<[7]%
\>[7]{}\texttt{(}\Varid{liftT}\;\Varid{k}\;\texttt{(}\Varid{fmap}\;\texttt{(}\mathbin{\texttt{,}}\texttt{)}\;\Varid{x}\texttt{)}\mathbin{\texttt{<*>}}\Varid{liftT}\;\Varid{k}\;\Varid{y}\texttt{)}{}\<[E]%
\\
\>[B]{}\equiv \mbox{\commentbegin  \ensuremath{\Varid{liftT}\;\Varid{k}} is natural  \commentend}{}\<[E]%
\\
\>[B]{}\hsindent{7}{}\<[7]%
\>[7]{}\Varid{k}\;\texttt{(}\Varid{fmap}\;\Varid{uncurry}\;\Varid{h}\texttt{)}\mathbin{\texttt{\color{red}{:\$:}}}{}\<[E]%
\\
\>[B]{}\hsindent{7}{}\<[7]%
\>[7]{}\texttt{(}\Varid{fmap}\;\texttt{(}\mathbin{\texttt{,}}\texttt{)}\;\texttt{(}\Varid{liftT}\;\Varid{k}\;\Varid{x}\texttt{)}\mathbin{\texttt{<*>}}\Varid{liftT}\;\Varid{k}\;\Varid{y}\texttt{)}{}\<[E]%
\\
\>[B]{}\equiv \mbox{\commentbegin  \ensuremath{\Varid{k}} is natural  \commentend}{}\<[E]%
\\
\>[B]{}\hsindent{7}{}\<[7]%
\>[7]{}\Varid{fmap}\;\Varid{uncurry}\;\texttt{(}\Varid{k}\;\Varid{h}\texttt{)}\mathbin{\texttt{\color{red}{:\$:}}}{}\<[E]%
\\
\>[B]{}\hsindent{7}{}\<[7]%
\>[7]{}\texttt{(}\Varid{fmap}\;\texttt{(}\mathbin{\texttt{,}}\texttt{)}\;\texttt{(}\Varid{liftT}\;\Varid{k}\;\Varid{x}\texttt{)}\mathbin{\texttt{<*>}}\Varid{liftT}\;\Varid{k}\;\Varid{y}\texttt{)}{}\<[E]%
\\
\>[B]{}\equiv \mbox{\commentbegin  definition of \ensuremath{\texttt{(}\mathbin{\texttt{<*>}}\texttt{)}}  \commentend}{}\<[E]%
\\
\>[B]{}\hsindent{7}{}\<[7]%
\>[7]{}\texttt{(}\Varid{k}\;\Varid{h}\mathbin{\texttt{\color{red}{:\$:}}}\Varid{liftT}\;\Varid{k}\;\Varid{x}\texttt{)}\mathbin{\texttt{<*>}}\Varid{liftT}\;\Varid{k}\;\Varid{y}{}\<[E]%
\\
\>[B]{}\equiv \mbox{\commentbegin  definition of \ensuremath{\Varid{liftT}}  \commentend}{}\<[E]%
\\
\>[B]{}\hsindent{7}{}\<[7]%
\>[7]{}\Varid{liftT}\;\Varid{k}\;\texttt{(}\Varid{h}\mathbin{\texttt{\color{red}{:\$:}}}\Varid{x}\texttt{)}\mathbin{\texttt{<*>}}\Varid{liftT}\;\Varid{k}\;\Varid{y}{}\<[E]%
\ColumnHook
\end{hscode}\resethooks

Now we need to verify that \ensuremath{\Varid{liftT}} satisfies the functor laws
\begin{align*}
  \ensuremath{\Varid{liftT}\;\Varid{id}} & \equiv \ensuremath{\Varid{id}} \\
  \ensuremath{\Varid{liftT}\;\texttt{(}\Varid{t}\hsdot{\circ }{\texttt{.}}\Varid{u}\texttt{)}} & \equiv \ensuremath{\Varid{liftT}\;\Varid{t}\hsdot{\circ }{\texttt{.}}\Varid{liftT}\;\Varid{u}}.
\end{align*}

The proof is a straightforward structural induction.
\end{proof}

We are going to need the following natural transformation
(which will be the unit of the adjunction \ref{adjunction}):

\begin{hscode}\SaveRestoreHook
\column{B}{@{}>{\hspre}l<{\hspost}@{}}%
\column{E}{@{}>{\hspre}l<{\hspost}@{}}%
\>[B]{}\Varid{one}\mathbin{::}\Conid{Functor}\;\Varid{f}\Rightarrow \Conid{Nat}\;\Varid{f}\;\texttt{(}\Conid{FreeA}\;\Varid{f}\texttt{)}{}\<[E]%
\\
\>[B]{}\Varid{one}\;\Varid{x}\mathrel{=}\Varid{fmap}\;\Varid{const}\;\Varid{x}\mathbin{\texttt{\color{red}{:\$:}}}\Conid{Pure}\;\texttt{(}\texttt{)}{}\<[E]%
\ColumnHook
\end{hscode}\resethooks

which embeds any functor \ensuremath{\Varid{f}} into \ensuremath{\Conid{FreeA}\;\Varid{f}} (we used a specialization of this
function for \ensuremath{\Conid{Option}} in section \ref{example:option_intro}).

\begin{lem}\label{lem:one_app}
$$
  \ensuremath{\Varid{g}\mathbin{\texttt{\color{red}{:\$:}}}\Varid{x}} \equiv \ensuremath{\Varid{one}\;\Varid{g}\mathbin{\texttt{<*>}}\Varid{x}}
$$
\end{lem}
\begin{proof}
Given

\begin{hscode}\SaveRestoreHook
\column{B}{@{}>{\hspre}l<{\hspost}@{}}%
\column{E}{@{}>{\hspre}l<{\hspost}@{}}%
\>[B]{}\Varid{h}\mathbin{::}\Varid{a}\to \texttt{(}\texttt{(}\texttt{)}\mathbin{\texttt{,}}\Varid{a}\texttt{)}{}\<[E]%
\\
\>[B]{}\Varid{h}\;\Varid{x}\mathrel{=}\texttt{(}\texttt{(}\texttt{)}\mathbin{\texttt{,}}\Varid{x}\texttt{)}{}\<[E]%
\ColumnHook
\end{hscode}\resethooks

it is easy to verify that:
\begin{equation}\label{lem:uncurry_h}
  \ensuremath{\texttt{(}\hsdot{\circ }{\texttt{.}}\Varid{h}\texttt{)}\hsdot{\circ }{\texttt{.}}\Varid{uncurry}\hsdot{\circ }{\texttt{.}}\Varid{const}} \equiv \ensuremath{\Varid{id}},
\end{equation}

so
\begin{hscode}\SaveRestoreHook
\column{B}{@{}>{\hspre}l<{\hspost}@{}}%
\column{5}{@{}>{\hspre}l<{\hspost}@{}}%
\column{E}{@{}>{\hspre}l<{\hspost}@{}}%
\>[5]{}\Varid{one}\;\Varid{g}\mathbin{\texttt{<*>}}\Varid{x}{}\<[E]%
\\
\>[B]{}\equiv \mbox{\commentbegin  definition of \ensuremath{\Varid{one}}  \commentend}{}\<[E]%
\\
\>[B]{}\hsindent{5}{}\<[5]%
\>[5]{}\texttt{(}\Varid{fmap}\;\Varid{const}\;\Varid{g}\mathbin{\texttt{\color{red}{:\$:}}}\Conid{Pure}\;\texttt{(}\texttt{)}\texttt{)}\mathbin{\texttt{<*>}}\Varid{x}{}\<[E]%
\\
\>[B]{}\equiv \mbox{\commentbegin  definition of \ensuremath{\texttt{(}\mathbin{\texttt{<*>}}\texttt{)}} and functor law for \ensuremath{\Varid{f}}  \commentend}{}\<[E]%
\\
\>[B]{}\hsindent{5}{}\<[5]%
\>[5]{}\Varid{fmap}\;\texttt{(}\Varid{uncurry}\hsdot{\circ }{\texttt{.}}\Varid{const}\texttt{)}\;\Varid{g}\mathbin{\texttt{\color{red}{:\$:}}}\Varid{fmap}\;\Varid{h}\;\Varid{x}{}\<[E]%
\\
\>[B]{}\equiv \mbox{\commentbegin  equation \ref{naturality} and functor law for \ensuremath{\Varid{f}}  \commentend}{}\<[E]%
\\
\>[B]{}\hsindent{5}{}\<[5]%
\>[5]{}\Varid{fmap}\;\texttt{(}\texttt{(}\hsdot{\circ }{\texttt{.}}\Varid{h}\texttt{)}\hsdot{\circ }{\texttt{.}}\Varid{uncurry}\hsdot{\circ }{\texttt{.}}\Varid{const}\texttt{)}\;\Varid{g}\mathbin{\texttt{\color{red}{:\$:}}}\Varid{x}{}\<[E]%
\\
\>[B]{}\equiv \mbox{\commentbegin  equation \ref{lem:uncurry_h}  \commentend}{}\<[E]%
\\
\>[B]{}\hsindent{5}{}\<[5]%
\>[5]{}\Varid{g}\mathbin{\texttt{\color{red}{:\$:}}}\Varid{x}{}\<[E]%
\ColumnHook
\end{hscode}\resethooks

\end{proof}

\begin{prop}
The \ensuremath{\Conid{FreeA}} functor is left adjoint to the forgetful functor $\mathcal{A}
\to \mathcal{F}$. Graphically:
\begin{align}\label{adjunction}
Hom_{\mathcal{F}}(\ensuremath{\Conid{FreeA}\;\Varid{f}}, \ensuremath{\Varid{g}}) \underset{\underset{\ensuremath{\Varid{raise}}}{\longleftarrow}}{\overset{\overset{\ensuremath{\Varid{lower}}}{\longrightarrow}}{\cong}} Hom_{\mathcal{A}}(\ensuremath{\Varid{f}}, \ensuremath{\Varid{g}})
\end{align}
\end{prop}
\begin{proof}
Given a functor \ensuremath{\Varid{f}} and an applicative functor \ensuremath{\Varid{g}}, we define a natural
bijection between \ensuremath{\Conid{Nat}\;\Varid{f}\;\Varid{g}} and \ensuremath{\Conid{AppNat}\;\texttt{(}\Conid{FreeA}\;\Varid{f}\texttt{)}\;\Varid{g}} as such:

\begin{hscode}\SaveRestoreHook
\column{B}{@{}>{\hspre}l<{\hspost}@{}}%
\column{8}{@{}>{\hspre}l<{\hspost}@{}}%
\column{E}{@{}>{\hspre}l<{\hspost}@{}}%
\>[B]{}\Varid{raise}{}\<[8]%
\>[8]{}\mathbin{::}\texttt{(}\Conid{Functor}\;\Varid{f}\mathbin{\texttt{,}}\Conid{Applicative}\;\Varid{g}\texttt{)}{}\<[E]%
\\
\>[8]{}\Rightarrow \Conid{Nat}\;\Varid{f}\;\Varid{g}{}\<[E]%
\\
\>[8]{}\to \Conid{AppNat}\;\texttt{(}\Conid{FreeA}\;\Varid{f}\texttt{)}\;\Varid{g}{}\<[E]%
\\
\>[B]{}\Varid{raise}\;\anonymous \;\texttt{(}\Conid{Pure}\;\Varid{x}\texttt{)}\mathrel{=}\Varid{pure}\;\Varid{x}{}\<[E]%
\\
\>[B]{}\Varid{raise}\;\Varid{k}\;\texttt{(}\Varid{g}\mathbin{\texttt{\color{red}{:\$:}}}\Varid{x}\texttt{)}\mathrel{=}\Varid{k}\;\Varid{g}\mathbin{\texttt{<*>}}\Varid{raise}\;\Varid{k}\;\Varid{x}{}\<[E]%
\\[\blanklineskip]%
\>[B]{}\Varid{lower}{}\<[8]%
\>[8]{}\mathbin{::}\texttt{(}\Conid{Functor}\;\Varid{f}\mathbin{\texttt{,}}\Conid{Applicative}\;\Varid{g}\texttt{)}{}\<[E]%
\\
\>[8]{}\Rightarrow \Conid{AppNat}\;\texttt{(}\Conid{FreeA}\;\Varid{f}\texttt{)}\;\Varid{g}{}\<[E]%
\\
\>[8]{}\to \Conid{Nat}\;\Varid{f}\;\Varid{g}{}\<[E]%
\\
\>[B]{}\Varid{lower}\;\Varid{k}\mathrel{=}\Varid{k}\hsdot{\circ }{\texttt{.}}\Varid{one}{}\<[E]%
\ColumnHook
\end{hscode}\resethooks

A routine verification shows that \ensuremath{\Varid{raise}} and \ensuremath{\Varid{lower}} are natural in
\ensuremath{\Varid{f}} and \ensuremath{\Varid{g}}. The proof that \ensuremath{\Varid{raise}\;\Varid{k}} satisfies the applicative natural 
transformation laws \ref{eq:applicative1} and \ref{eq:applicative2} is a
straightforward induction having the same structure as the proof that \ensuremath{\Varid{liftT}\;\Varid{k}}
satisfies these laws (proposition \ref{prop:FreeAFunctor}).
To show that \ensuremath{\Varid{f}} and \ensuremath{\Varid{g}} are inverses of each other,
we reason by induction and calculate in one direction:

\begin{hscode}\SaveRestoreHook
\column{B}{@{}>{\hspre}l<{\hspost}@{}}%
\column{5}{@{}>{\hspre}l<{\hspost}@{}}%
\column{E}{@{}>{\hspre}l<{\hspost}@{}}%
\>[5]{}\Varid{raise}\;\texttt{(}\Varid{lower}\;\Varid{t}\texttt{)}\;\texttt{(}\Conid{Pure}\;\Varid{x}\texttt{)}{}\<[E]%
\\
\>[B]{}\equiv \mbox{\commentbegin  definition of \ensuremath{\Varid{raise}}  \commentend}{}\<[E]%
\\
\>[B]{}\hsindent{5}{}\<[5]%
\>[5]{}\Varid{pure}\;\Varid{x}{}\<[E]%
\\
\>[B]{}\equiv \mbox{\commentbegin  \ensuremath{\Varid{t}} is an applicative natural transformation  \commentend}{}\<[E]%
\\
\>[B]{}\hsindent{5}{}\<[5]%
\>[5]{}\Varid{t}\;\texttt{(}\Varid{pure}\;\Varid{x}\texttt{)}{}\<[E]%
\\
\>[B]{}\equiv \mbox{\commentbegin  definition of \ensuremath{\Varid{pure}}  \commentend}{}\<[E]%
\\
\>[B]{}\hsindent{5}{}\<[5]%
\>[5]{}\Varid{t}\;\texttt{(}\Conid{Pure}\;\Varid{x}\texttt{)}{}\<[E]%
\ColumnHook
\end{hscode}\resethooks

\begin{hscode}\SaveRestoreHook
\column{B}{@{}>{\hspre}l<{\hspost}@{}}%
\column{5}{@{}>{\hspre}l<{\hspost}@{}}%
\column{E}{@{}>{\hspre}l<{\hspost}@{}}%
\>[5]{}\Varid{raise}\;\texttt{(}\Varid{lower}\;\Varid{t}\texttt{)}\;\texttt{(}\Varid{g}\mathbin{\texttt{\color{red}{:\$:}}}\Varid{x}\texttt{)}{}\<[E]%
\\
\>[B]{}\equiv \mbox{\commentbegin  definition of \ensuremath{\Varid{raise}}  \commentend}{}\<[E]%
\\
\>[B]{}\hsindent{5}{}\<[5]%
\>[5]{}\Varid{lower}\;\Varid{t}\;\Varid{g}\mathbin{\texttt{<*>}}\Varid{raise}\;\texttt{(}\Varid{lower}\;\Varid{t}\texttt{)}\;\Varid{x}{}\<[E]%
\\
\>[B]{}\equiv \mbox{\commentbegin  induction hypothesis  \commentend}{}\<[E]%
\\
\>[B]{}\hsindent{5}{}\<[5]%
\>[5]{}\Varid{lower}\;\Varid{t}\;\Varid{g}\mathbin{\texttt{<*>}}\Varid{t}\;\Varid{x}{}\<[E]%
\\
\>[B]{}\equiv \mbox{\commentbegin  definition of \ensuremath{\Varid{lower}}  \commentend}{}\<[E]%
\\
\>[B]{}\hsindent{5}{}\<[5]%
\>[5]{}\Varid{t}\;\texttt{(}\Varid{one}\;\Varid{g}\texttt{)}\mathbin{\texttt{<*>}}\Varid{t}\;\Varid{x}{}\<[E]%
\\
\>[B]{}\equiv \mbox{\commentbegin  \ensuremath{\Varid{t}} is an applicative natural transformation  \commentend}{}\<[E]%
\\
\>[B]{}\hsindent{5}{}\<[5]%
\>[5]{}\Varid{t}\;\texttt{(}\Varid{one}\;\Varid{g}\mathbin{\texttt{<*>}}\Varid{x}\texttt{)}{}\<[E]%
\\
\>[B]{}\equiv \mbox{\commentbegin  lemma \ref{lem:one_app}  \commentend}{}\<[E]%
\\
\>[B]{}\hsindent{5}{}\<[5]%
\>[5]{}\Varid{t}\;\texttt{(}\Varid{g}\mathbin{\texttt{\color{red}{:\$:}}}\Varid{x}\texttt{)}{}\<[E]%
\ColumnHook
\end{hscode}\resethooks

The other direction:

\begin{hscode}\SaveRestoreHook
\column{B}{@{}>{\hspre}l<{\hspost}@{}}%
\column{5}{@{}>{\hspre}l<{\hspost}@{}}%
\column{E}{@{}>{\hspre}l<{\hspost}@{}}%
\>[5]{}\Varid{lower}\;\texttt{(}\Varid{raise}\;\Varid{t}\texttt{)}\;\Varid{x}{}\<[E]%
\\
\>[B]{}\equiv \mbox{\commentbegin  definition of \ensuremath{\Varid{lower}}  \commentend}{}\<[E]%
\\
\>[B]{}\hsindent{5}{}\<[5]%
\>[5]{}\Varid{raise}\;\Varid{t}\;\texttt{(}\Varid{one}\;\Varid{x}\texttt{)}{}\<[E]%
\\
\>[B]{}\equiv \mbox{\commentbegin  definition of \ensuremath{\Varid{one}}  \commentend}{}\<[E]%
\\
\>[B]{}\hsindent{5}{}\<[5]%
\>[5]{}\Varid{raise}\;\Varid{t}\;\texttt{(}\Varid{fmap}\;\Varid{const}\;\Varid{x}\mathbin{\texttt{\color{red}{:\$:}}}\Conid{Pure}\;\texttt{(}\texttt{)}\texttt{)}{}\<[E]%
\\
\>[B]{}\equiv \mbox{\commentbegin  definition of \ensuremath{\Varid{raise}}  \commentend}{}\<[E]%
\\
\>[B]{}\hsindent{5}{}\<[5]%
\>[5]{}\Varid{t}\;\texttt{(}\Varid{fmap}\;\Varid{const}\;\Varid{x}\texttt{)}\mathbin{\texttt{<*>}}\Varid{pure}\;\texttt{(}\texttt{)}{}\<[E]%
\\
\>[B]{}\equiv \mbox{\commentbegin  \ensuremath{\Varid{t}} is natural  \commentend}{}\<[E]%
\\
\>[B]{}\hsindent{5}{}\<[5]%
\>[5]{}\Varid{fmap}\;\Varid{const}\;\texttt{(}\Varid{t}\;\Varid{x}\texttt{)}\mathbin{\texttt{<*>}}\Varid{pure}\;\texttt{(}\texttt{)}{}\<[E]%
\\
\>[B]{}\equiv \mbox{\commentbegin  \ensuremath{\Varid{fmap}\;\Varid{h}} $\equiv$ \ensuremath{\texttt{(}\texttt{(}\Varid{pure}\;\Varid{h}\texttt{)}\mathbin{\texttt{<*>}}\texttt{)}} in an applicative functor  \commentend}{}\<[E]%
\\
\>[B]{}\hsindent{5}{}\<[5]%
\>[5]{}\Varid{pure}\;\Varid{const}\mathbin{\texttt{<*>}}\Varid{t}\;\Varid{x}\mathbin{\texttt{<*>}}\Varid{pure}\;\texttt{(}\texttt{)}{}\<[E]%
\\
\>[B]{}\equiv \mbox{\commentbegin  \ensuremath{\Varid{t}} is natural  \commentend}{}\<[E]%
\\
\>[B]{}\hsindent{5}{}\<[5]%
\>[5]{}\Varid{pure}\;\texttt{(}\mathbin{\texttt{\$}}\texttt{(}\texttt{)}\texttt{)}\mathbin{\texttt{<*>}}\texttt{(}\Varid{pure}\;\Varid{const}\mathbin{\texttt{<*>}}\Varid{t}\;\Varid{x}\texttt{)}{}\<[E]%
\\
\>[B]{}\equiv \mbox{\commentbegin  applicative law \ref{app:comp}  \commentend}{}\<[E]%
\\
\>[B]{}\hsindent{5}{}\<[5]%
\>[5]{}\Varid{pure}\;\texttt{(}\hsdot{\circ }{\texttt{.}}\texttt{)}\mathbin{\texttt{<*>}}\Varid{pure}\;\texttt{(}\mathbin{\texttt{\$}}\texttt{(}\texttt{)}\texttt{)}\mathbin{\texttt{<*>}}\Varid{pure}\;\Varid{const}\mathbin{\texttt{<*>}}\Varid{t}\;\Varid{x}{}\<[E]%
\\
\>[B]{}\equiv \mbox{\commentbegin  applicative law \ref{app:hom} applied twice  \commentend}{}\<[E]%
\\
\>[B]{}\hsindent{5}{}\<[5]%
\>[5]{}\Varid{pure}\;\Varid{id}\mathbin{\texttt{<*>}}\Varid{t}\;\Varid{x}{}\<[E]%
\\
\>[B]{}\equiv \mbox{\commentbegin  applicative law \ref{app:id}  \commentend}{}\<[E]%
\\
\>[B]{}\hsindent{5}{}\<[5]%
\>[5]{}\Varid{t}\;\Varid{x}{}\<[E]%
\ColumnHook
\end{hscode}\resethooks
\end{proof}

\subsection{Example: option parsers (continued)}\label{example:option_raise}

With the help of the adjunction defined above by \ensuremath{\Varid{raise}} and \ensuremath{\Varid{lower}}
we are able to define some useful functions. In the case of
command-line option parsers, for example, it can be used for
computing the global default value of a parser:

\begin{hscode}\SaveRestoreHook
\column{B}{@{}>{\hspre}l<{\hspost}@{}}%
\column{E}{@{}>{\hspre}l<{\hspost}@{}}%
\>[B]{}\Varid{parserDefault}\mathbin{::}\Conid{FreeA}\;\Conid{Option}\;\Varid{a}\to \Conid{Maybe}\;\Varid{a}{}\<[E]%
\\
\>[B]{}\Varid{parserDefault}\mathrel{=}\Varid{raise}\;\Varid{optDefault}{}\<[E]%
\ColumnHook
\end{hscode}\resethooks

or for extracting the list of all the options in a parser:

\begin{hscode}\SaveRestoreHook
\column{B}{@{}>{\hspre}l<{\hspost}@{}}%
\column{3}{@{}>{\hspre}l<{\hspost}@{}}%
\column{5}{@{}>{\hspre}l<{\hspost}@{}}%
\column{E}{@{}>{\hspre}l<{\hspost}@{}}%
\>[B]{}\Varid{allOptions}\mathbin{::}\Conid{FreeA}\;\Conid{Option}\;\Varid{a}\to [\mskip1.5mu \Conid{String}\mskip1.5mu]{}\<[E]%
\\
\>[B]{}\Varid{allOptions}\mathrel{=}\Varid{getConst}\hsdot{\circ }{\texttt{.}}\Varid{raise}\;\Varid{f}{}\<[E]%
\\
\>[B]{}\hsindent{3}{}\<[3]%
\>[3]{}\texttt{\color{blue}\textbf{where}}{}\<[E]%
\\
\>[3]{}\hsindent{2}{}\<[5]%
\>[5]{}\Varid{f}\;\Varid{opt}\mathrel{=}\Conid{Const}\;[\mskip1.5mu \Varid{optName}\;\Varid{opt}\mskip1.5mu]{}\<[E]%
\ColumnHook
\end{hscode}\resethooks

\ensuremath{\Varid{allOptions}} works by first defining a function that takes an
option and returns a one-element list with the name of the option, and then
lifting it to the \ensuremath{\Conid{Const}} applicative functor.

The \ensuremath{\Varid{raise}} function can be thought of as a way to define a ``semantics'' for
the whole syntax of the DSL corresponding to \ensuremath{\Conid{FreeA}\;\Varid{f}}, given one for just the
individual atomic actions, expressed as a natural transformation from the
functor \ensuremath{\Varid{f}} to any applicative functor \ensuremath{\Varid{g}}.

When defining such a semantics using \ensuremath{\Varid{raise}}, the resulting function is
automatically an applicative natural transformation.  In some circumstances,
however, it is more convenient to define a function by pattern matching directly
on the constructors of \ensuremath{\Conid{FreeA}\;\Varid{f}}, like when the target does not have an obvious
applicative functor structure that makes the desired function an applicative
natural transformation.

For example, we can write a function that runs an applicative option parser
over a list of command-line arguments, accepting them in any order:

\begin{hscode}\SaveRestoreHook
\column{B}{@{}>{\hspre}l<{\hspost}@{}}%
\column{3}{@{}>{\hspre}l<{\hspost}@{}}%
\column{11}{@{}>{\hspre}l<{\hspost}@{}}%
\column{E}{@{}>{\hspre}l<{\hspost}@{}}%
\>[B]{}\Varid{matchOpt}{}\<[11]%
\>[11]{}\mathbin{::}\Conid{String}\to \Conid{String}{}\<[E]%
\\
\>[11]{}\to \Conid{FreeA}\;\Conid{Option}\;\Varid{a}{}\<[E]%
\\
\>[11]{}\to \Conid{Maybe}\;\texttt{(}\Conid{FreeA}\;\Conid{Option}\;\Varid{a}\texttt{)}{}\<[E]%
\\
\>[B]{}\Varid{matchOpt}\;\anonymous \;\anonymous \;\texttt{(}\Conid{Pure}\;\anonymous \texttt{)}\mathrel{=}\Conid{Nothing}{}\<[E]%
\\
\>[B]{}\Varid{matchOpt}\;\Varid{opt}\;\Varid{value}\;\texttt{(}\Varid{g}\mathbin{\texttt{\color{red}{:\$:}}}\Varid{x}\texttt{)}{}\<[E]%
\\
\>[B]{}\hsindent{3}{}\<[3]%
\>[3]{}\mid \Varid{opt}\equiv \text{\tt '-'}\mathbin{:}\text{\tt '-'}\mathbin{:}\Varid{optName}\;\Varid{g}{}\<[E]%
\\
\>[B]{}\hsindent{3}{}\<[3]%
\>[3]{}\mathrel{=}\Varid{fmap}\;\texttt{(}\mathbin{\texttt{<\$>}}\Varid{x}\texttt{)}\;\texttt{(}\Varid{optReader}\;\Varid{g}\;\Varid{value}\texttt{)}{}\<[E]%
\\
\>[B]{}\hsindent{3}{}\<[3]%
\>[3]{}\mid \Varid{otherwise}{}\<[E]%
\\
\>[B]{}\hsindent{3}{}\<[3]%
\>[3]{}\mathrel{=}\Varid{fmap}\;\texttt{(}\Varid{g}\mathbin{\texttt{\color{red}{:\$:}}}\texttt{)}\;\texttt{(}\Varid{matchOpt}\;\Varid{opt}\;\Varid{value}\;\Varid{x}\texttt{)}{}\<[E]%
\ColumnHook
\end{hscode}\resethooks

The \ensuremath{\Varid{matchOpt}} function looks for options in the parser which match the given
command-line argument, and, if successful, returns a modified parser where the
option has been replaced by a pure value.

Clearly, \ensuremath{\Varid{matchOpt}\;\Varid{opt}\;\Varid{value}} is not applicative, since, for instance, equation
\ref{eq:applicative1} is not satisfied.

\begin{hscode}\SaveRestoreHook
\column{B}{@{}>{\hspre}l<{\hspost}@{}}%
\column{3}{@{}>{\hspre}l<{\hspost}@{}}%
\column{5}{@{}>{\hspre}l<{\hspost}@{}}%
\column{12}{@{}>{\hspre}l<{\hspost}@{}}%
\column{E}{@{}>{\hspre}l<{\hspost}@{}}%
\>[B]{}\Varid{runParser}{}\<[12]%
\>[12]{}\mathbin{::}\Conid{FreeA}\;\Conid{Option}\;\Varid{a}{}\<[E]%
\\
\>[12]{}\to [\mskip1.5mu \Conid{String}\mskip1.5mu]{}\<[E]%
\\
\>[12]{}\to \Conid{Maybe}\;\Varid{a}{}\<[E]%
\\
\>[B]{}\Varid{runParser}\;\Varid{p}\;\texttt{(}\Varid{opt}\mathbin{:}\Varid{value}\mathbin{:}\Varid{args}\texttt{)}\mathrel{=}{}\<[E]%
\\
\>[B]{}\hsindent{3}{}\<[3]%
\>[3]{}\texttt{\color{blue}\textbf{case}}\;\Varid{matchOpt}\;\Varid{opt}\;\Varid{value}\;\Varid{p}\;\texttt{\color{blue}\textbf{of}}{}\<[E]%
\\
\>[3]{}\hsindent{2}{}\<[5]%
\>[5]{}\Conid{Nothing}\to \Conid{Nothing}{}\<[E]%
\\
\>[3]{}\hsindent{2}{}\<[5]%
\>[5]{}\Conid{Just}\;\Varid{p'}\to \Varid{runParser}\;\Varid{p'}\;\Varid{args}{}\<[E]%
\\
\>[B]{}\Varid{runParser}\;\Varid{p}\;[\mskip1.5mu \mskip1.5mu]\mathrel{=}\Varid{parserDefault}\;\Varid{p}{}\<[E]%
\\
\>[B]{}\Varid{runParser}\;\anonymous \;\anonymous \mathrel{=}\Conid{Nothing}{}\<[E]%
\ColumnHook
\end{hscode}\resethooks

Finally, \ensuremath{\Varid{runParser}} calls \ensuremath{\Varid{matchOpt}} with successive pairs of arguments, until
no arguments remain, at which point it uses the default values of the remaining
options to construct a result.

\section{Totality}\label{sec:totality}

All the proofs in this paper apply to a total fragment of Haskell, and
completely ignore the presence of bottom.  The Haskell subset we use can be
given a semantics in any locally presentable cartesian closed category.

In fact, if we assume that all the functors used throughout the paper are
accessible, all our inductive definitions can be regarded as initial algebras of
accessible functors.

For example, to realise \ensuremath{\Conid{FreeA}\;\Varid{f}}, assume \ensuremath{\Varid{f}} is $\kappa$-accessible for some
regular cardinal $\kappa$.  Then define a functor:
\[
A : \mathsf{Func}_\kappa(\mathcal{C}, \mathcal{C}) \to
    \mathsf{Func}_\kappa(\mathcal{C}, \mathcal{C}),
\]
where $\mathsf{Func}_\kappa$ is the category of $\kappa$-accessible endofunctors
of $\mathcal{C}$, which is itself locally presentable by proposition
\ref{prop:func_lp}.

The inductive definition of \ensuremath{\Conid{FreeA}\;\Varid{f}} above can then be regarded as the initial
algebra of $A$, given by:
\begin{equation} \label{eq:higher_functor}
(A G) a = a + \int^{b : \mathcal{C}} F [b, a] \times G b,
\end{equation}
where $[-,-]$ denotes the internal hom (exponential) in $\C$. Since $F$ and $G$
are locally presentable, and $\mathcal{C}$ is cocomplete, the coend exists by
lemma \ref{lem:day_alt}, and $A G$ is $\kappa$-accessible by lemma
\ref{lem:day_accessible}, provided $\kappa$ is large enough.

Furthermore, the functor $A$ itself is accessible by proposition
\ref{prop:day_op_accessible}, hence it has an initial algebra. Equation
\ref{naturality} is then a trivial consequence of this definition.

As for function definitions, most use primitive recursion, so they can be
realised by using the universal property of the initial algebra directly.

One exception is the definition of \ensuremath{\texttt{(}\mathbin{\texttt{<*>}}\texttt{)}}:

\begin{hscode}\SaveRestoreHook
\column{B}{@{}>{\hspre}l<{\hspost}@{}}%
\column{3}{@{}>{\hspre}l<{\hspost}@{}}%
\column{20}{@{}>{\hspre}l<{\hspost}@{}}%
\column{E}{@{}>{\hspre}l<{\hspost}@{}}%
\>[3]{}\texttt{(}\Varid{h}\mathbin{\texttt{\color{red}{:\$:}}}\Varid{x}\texttt{)}\mathbin{\texttt{<*>}}\Varid{y}{}\<[20]%
\>[20]{}\mathrel{=}\Varid{fmap}\;\Varid{uncurry}\;\Varid{h}\mathbin{\texttt{\color{red}{:\$:}}}\texttt{(}\texttt{(}\mathbin{\texttt{,}}\texttt{)}\mathbin{\texttt{<\$>}}\Varid{x}\mathbin{\texttt{<*>}}\Varid{y}\texttt{)}{}\<[E]%
\ColumnHook
\end{hscode}\resethooks

which contains a recursive call where the first argument, namely \ensuremath{\texttt{(}\mathbin{\texttt{,}}\texttt{)}\mathbin{\texttt{<\$>}}\Varid{x}}, is
not structurally smaller than the original one (\ensuremath{\Varid{h}\mathbin{\texttt{\color{red}{:\$:}}}\Varid{x}}).

To prove that this function is nevertheless well defined, we introduce a notion
of \emph{size} for values of type \ensuremath{\Conid{FreeA}\;\Varid{f}\;\Varid{a}}:

\begin{hscode}\SaveRestoreHook
\column{B}{@{}>{\hspre}l<{\hspost}@{}}%
\column{E}{@{}>{\hspre}l<{\hspost}@{}}%
\>[B]{}\Varid{size}\mathbin{::}\Conid{FreeA}\;\Varid{f}\;\Varid{a}\to \mathbb{N}{}\<[E]%
\\
\>[B]{}\Varid{size}\;\texttt{(}\Conid{Pure}\;\anonymous \texttt{)}\mathrel{=}\mathrm{0}{}\<[E]%
\\
\>[B]{}\Varid{size}\;\texttt{(}\anonymous \mathbin{\texttt{\color{red}{:\$:}}}\Varid{x}\texttt{)}\mathrel{=}\mathrm{1}\mathbin{+}\Varid{size}\;\Varid{x}{}\<[E]%
\ColumnHook
\end{hscode}\resethooks

To conclude that the definition of \ensuremath{\texttt{(}\mathbin{\texttt{<*>}}\texttt{)}} can be made sense of in our target
category, we just need to show that the size of the argument in the recursive
call is smaller than the size of the original argument, which is an immediate
consequence of the following lemma.

\begin{lem}
For any function \ensuremath{\Varid{f}\mathbin{::}\Varid{a}\to \Varid{b}} and \ensuremath{\Varid{u}\mathbin{::}\Conid{FreeA}\;\Varid{f}\;\Varid{a}},
\begin{equation*}
\ensuremath{\Varid{size}\;\texttt{(}\Varid{fmap}\;\Varid{f}\;\Varid{u}\texttt{)}} \equiv \ensuremath{\Varid{size}\;\Varid{u}}
\end{equation*}
\end{lem}
\begin{proof}
By induction:
\begin{hscode}\SaveRestoreHook
\column{B}{@{}>{\hspre}l<{\hspost}@{}}%
\column{5}{@{}>{\hspre}l<{\hspost}@{}}%
\column{E}{@{}>{\hspre}l<{\hspost}@{}}%
\>[5]{}\Varid{size}\;\texttt{(}\Varid{fmap}\;\Varid{f}\;\texttt{(}\Conid{Pure}\;\Varid{x}\texttt{)}\texttt{)}{}\<[E]%
\\
\>[B]{}\equiv \mbox{\commentbegin  definition of \ensuremath{\Varid{fmap}}  \commentend}{}\<[E]%
\\
\>[B]{}\hsindent{5}{}\<[5]%
\>[5]{}\Varid{size}\;\texttt{(}\Conid{Pure}\;\texttt{(}\Varid{f}\;\Varid{x}\texttt{)}\texttt{)}{}\<[E]%
\\
\>[B]{}\equiv \mbox{\commentbegin  definition of \ensuremath{\Varid{size}}  \commentend}{}\<[E]%
\\
\>[B]{}\hsindent{5}{}\<[5]%
\>[5]{}\mathrm{0}{}\<[E]%
\\
\>[B]{}\equiv \mbox{\commentbegin  definition of \ensuremath{\Varid{size}}  \commentend}{}\<[E]%
\\
\>[B]{}\hsindent{5}{}\<[5]%
\>[5]{}\Varid{size}\;\texttt{(}\Conid{Pure}\;\Varid{x}\texttt{)}{}\<[E]%
\ColumnHook
\end{hscode}\resethooks

\begin{hscode}\SaveRestoreHook
\column{B}{@{}>{\hspre}l<{\hspost}@{}}%
\column{5}{@{}>{\hspre}l<{\hspost}@{}}%
\column{E}{@{}>{\hspre}l<{\hspost}@{}}%
\>[5]{}\Varid{size}\;\texttt{(}\Varid{fmap}\;\Varid{f}\;\texttt{(}\Varid{g}\mathbin{\texttt{\color{red}{:\$:}}}\Varid{x}\texttt{)}\texttt{)}{}\<[E]%
\\
\>[B]{}\equiv \mbox{\commentbegin  definition of \ensuremath{\Varid{fmap}}  \commentend}{}\<[E]%
\\
\>[B]{}\hsindent{5}{}\<[5]%
\>[5]{}\Varid{size}\;\texttt{(}\Varid{fmap}\;\texttt{(}\Varid{f}\hsdot{\circ }{\texttt{.}}\texttt{)}\;\Varid{g}\mathbin{\texttt{\color{red}{:\$:}}}\Varid{x}\texttt{)}{}\<[E]%
\\
\>[B]{}\equiv \mbox{\commentbegin  definition of \ensuremath{\Varid{size}}  \commentend}{}\<[E]%
\\
\>[B]{}\hsindent{5}{}\<[5]%
\>[5]{}\mathrm{1}\mathbin{+}\Varid{size}\;\Varid{x}{}\<[E]%
\\
\>[B]{}\equiv \mbox{\commentbegin  definition of \ensuremath{\Varid{size}}  \commentend}{}\<[E]%
\\
\>[B]{}\hsindent{5}{}\<[5]%
\>[5]{}\Varid{size}\;\texttt{(}\Varid{g}\mathbin{\texttt{\color{red}{:\$:}}}\Varid{x}\texttt{)}{}\<[E]%
\ColumnHook
\end{hscode}\resethooks

\end{proof}

In most of our proofs using induction we carry out induction on the size of the
first argument of \ensuremath{\texttt{(}\mathbin{\texttt{<*>}}\texttt{)}} where size is defined by the above \ensuremath{\Varid{size}} function.

\section{Semantics}\label{sec:semantics}

In this section, we establish the results about accessible functors of locally
presentable categories that we used in section \ref{sec:totality} to justify the
inductive definition of \ensuremath{\Conid{FreeA}\;\Varid{f}}.

We begin with a technical lemma:

\begin{lem}\label{lem:kan_accessible}
Suppose we have the following diagram of categories and functors:
\[
\xymatrix{
\A \ar[d]_i \\
\B \ar[r]^F \ar[d]_K &
\C \\
\D \ar@{-->}[ru]_L
}
\]

where $\B$, and $\C$ are locally presentable, $F$ is accessible, and $i$ is the
inclusion of a dense small full subcategory of compact objects of $\B$.  Then
the pointwise left Kan extension $L$ of $F$ along $K$ exists and is equal to the
left Kan extension of $F i$ along $K i$.

If, furthermore, $\D$ is locally presentable and $K$ is accessible, then $L$ is
accessible.
\end{lem}
\begin{proof}
Let $\kappa$ be a regular cardinal such that $\B$ and $\C$ are $\kappa$-locally
presentable.

The pointwise left Kan extension $L$ can be obtained as a colimit:
\begin{equation} \label{eq:kan_colimit}
Ld = \colim_{\substack{b : \B\\g : Kb \to d}} F b,
\end{equation}
Where the indices range over the comma category $(K \downarrow d)$.  To show
that $L$ exists, it is therefore enough to prove that the colimit
\ref{eq:kan_colimit} can be realised as the small colimit:
\[
\colim_{\substack{a : \A\\f : K a \to d}} F a.
\]

For any $b : \B$, and $g : K b \to d$, we can express $b$ as a canonical
$\kappa$-filtered colimit of compact objects:
\[
b \cong \colim_{\substack{a : \A\\h : a \to b}} a.
\]
Since $F$ preserves $\kappa$-filtered colimits, we then get a morphism:
\[
F b \to \colim_{\substack{a : \A\\h : a \to b}} F a \to
\colim_{\substack{a : \A\\f : K a \to d}} F a.
\]
This gives a cocone for the colimit \ref{eq:kan_colimit}, and a straightforward
verification shows that it is universal.

As for the second statement, suppose $\D$ is also $\kappa$-locally presentable.
By possibly increasing $\kappa$, we can assume that $K a$ is $\kappa$-compact
for all $a : \A$ (such a $\kappa$ exists because $\A$ is small, and every object
of $\D$ is $\lambda$-compact for some $\lambda$).

Then, by the first part:
\[
L d = \int^{a : \A} F a \cdot \D(K a, d).
\]

Now, a filtered colimit in $d$ commutes with $\D(K a, -)$ because $K a$ is
compact, it commutes with $F a \cdot -$ because copowers are left adjoints, and
it commutes with coends because they are both colimits.

Therefore, $L$ is accessible.
\end{proof}

From now on, let $\B$ and $\C$ be categories with finite products, and $F, G :
\B \to \C$ be functors.

\begin{defn}\label{def:day}
The \emph{Day convolution} of $F$ and $G$, denoted $F \ast G$, is the pointwise
left Kan extension of the diagonal functor in the following diagram:
\[
\xymatrix{
\B \times \B \ar[r]^{F \times G} \ar[d]_{\times} \ar[rd] &
\C \times \C \ar[d]^{\times} \\
\B \ar@{-->}[r]_{F \ast G} &
\C
}
\]
\end{defn}

Note that the Day convolution of two functors might not exist, but it certainly
does if $\B$ is small and $\C$ is cocomplete.

\begin{lem}\label{lem:day_accessible}
Suppose that $\B$ and $\C$ are locally presentable and $F$ and $G$ are
accessible.  Then the Day convolution of $F$ and $G$ exists and is accessible.
\end{lem}
\begin{proof}
Immediate consequence of lemma \ref{lem:kan_accessible}.
\end{proof}

\begin{lem}\label{lem:day_alt}
Suppose that $\B$ is cartesian closed.  Then the Day convolution of $F$ and $G$
can be obtained as the coend:
\[
(F \ast G) b = \int^{y : \B} F [y, b] \times G y
\]
\end{lem}
\begin{proof}
By coend calculus:
\begin{equation*}
\begin{aligned}
  & (F \ast G) b \\
= & \int^{x y : \B} F x \times G y \cdot \C(x \times y, b) \\
= & \int^{y : \B} \left( \int^{x : \B} F x \cdot \C(x \times y, b) \right) \times G y \\
= & \int^{y : \B} \left( \int^{x : \B} F x \cdot \C(x, [y, b]) \right) \times G y \\
= & \int^{y : \B} F [y, b] \times G y
\end{aligned}
\end{equation*}
\end{proof}

\begin{prop}\label{prop:func_lp}
Let $\kappa$ be a regular cardinal, and $\B$ and $\C$ be locally
$\kappa$-presentable.  Then the category $\mathsf{Func}_\kappa(\B, \C)$ of
$\kappa$-accessible functors is locally $\kappa$-presentable.
\end{prop}
\begin{proof}
Let $\A$ be a dense small full subcategory of $\B$.  The obvious functor
$\mathsf{Func}_\kappa(\B, \C) \to \mathsf{Func}(\A, \C)$ is an equivalence of
categories (its inverse is given by left Kan extensions along the inclusion $\A
\to \B$), and $\mathsf{Func}(\A, \C)$ is locally $\kappa$-presentable (see for
example \cite{adamek}, corollary 1.54).
\end{proof}

\begin{prop}\label{prop:day_op_accessible}
Let $\kappa$ be a regular cardinal such that $\B$ and $\C$ are locally
$\kappa$-presentable, and the Day convolution of any two $\kappa$-accessible
functors is $\kappa$-accessible (which exists by lemma \ref{lem:day_accessible}).

Then the Day convolution operator
\[
\mathsf{Func}_\kappa(\B, \C) \times \mathsf{Func}_\kappa(\B, \C) \to
\mathsf{Func}_\kappa(\B, \C)
\]
is itself a $\kappa$-accessible functor.
\end{prop}
\begin{proof}
It is enough to show that $\ast$ preserves filtered colimits pointwise in its
two variables separately.  But this is clear, since filtered colimits commute
with finite products, copowers and coends.
\end{proof}

We can recast equation \ref{eq:higher_functor} in terms of Day convolution as
follows:
\begin{equation} \label{eq:higher_functor_day}
A G = \mathsf{Id} + F \ast G.
\end{equation}

Equation \ref{eq:higher_functor_day} makes precise the intuition that free
applicative functors are in some sense lists (i.e. free monoids).  In fact, the
functor $A$ is exactly the one appearing in the usual recursive definition of
lists, only in this case the construction is happening in the monoidal category
of accessible endofunctors equipped with Day convolution.

We also sketch the following purely categorical construction of free applicative
(i.e. lax monoidal) functors, which is not essential for the rest of the paper,
but is quite an easy consequence of the machinery developed in this section.

The  idea is  to  perform the  ``list''  construction in  one  step, instead  of
iterating individual Day convolutions using recursion.  Namely, for any category
$\C$, let  $\C^*$ be the \emph{free  monoidal category} generated by  $\C$.  The
objects (resp. morphisms) of $\C^*$  are lists of objects (resp. morphisms) of
$\C$.  Clearly, $\C^*$ is accessible if $\C$ is.

If $\C$ has finite products, there is a functor
\[
\epsilon : \C^* \to \C
\]
which maps a list to its corresponding product.  Note that $\epsilon$ is
accessible.  Furthermore, the assigment $\C \mapsto \C^*$ extends to a 2-functor on
$\mathsf{Cat}$ which preserves accessibility of functors.

Now, the free applicative $G$ on a functor $F : \C \to \C$ is simply defined to
be the Kan extension of $\epsilon \circ F^*$ along $\epsilon$:
\[
\xymatrix{
\C^* \ar[r]^{F^*} \ar[d]_\epsilon \ar[dr] &
\C^* \ar[d]^\epsilon \\
\C \ar@{-->}[r]_G & \C
}
\]

The functor $G$ is accessible by an appropriate generalisation of lemma
\ref{lem:kan_accessible}, and it is not hard to see that it is lax monoidal (see
for example \cite{constructing_applicatives}, proposition 4).  We omit the proof
that $G$ is a free object, which can be obtained by diagram chasing using the
universal property of Kan extensions.

\section{Related work}\label{sec:related}

The idea of free applicative functors is not entirely new.  There have been a
number of different definitions of free applicative functor over a given Haskell
functor, but none of them includes a proof of the applicative laws.

The first author of this paper published a specific instance of applicative
functors\footnote{\url{http://paolocapriotti.com/blog/2012/04/27/applicative-option-parser}}
similar to our example shown in section \ref{example:option_intro}.  The example
has later been expanded into a fully-featured Haskell library for command line
option parsing.\footnote{\url{http://hackage.haskell.org/package/optparse-applicative}}

Tom Ellis proposes a definition very similar to
ours,\footnote{\url{http://web.jaguarpaw.co.uk/~tom/blog/posts/2012-09-09-towards-free-applicatives.html}}
but uses a separate inductive type for the case corresponding to our \ensuremath{\texttt{(}\mathbin{\texttt{\color{red}{:\$:}}}\texttt{)}}
constructor.  He then observes that law \ref{app:hom} probably holds because of
the existential quantification, but does not provide a proof.  We solve this
problem by deriving the necessary equation \ref{naturality} as a ``free
theorem''.

Gerg\H{o} \'Erdi gives another similar
definition\footnote{\url{http://gergo.erdi.hu/blog/2012-12-01-static_analysis_with_applicatives/}},
but his version presents some redundancies, and thus fails to obey the
applicative laws.  For example, \ensuremath{\Conid{Pure}\;\Varid{id}\mathbin{\texttt{<*>}}\Varid{x}} can easily be distinguished from
\ensuremath{\Varid{x}} using a function like our \ensuremath{\Varid{count}} above, defined by pattern matching on the
constructors.

However, this is remedied by only exposing a limited interface which includes
the equivalent of our \ensuremath{\Varid{raise}} function, but \emph{not} the \ensuremath{\Conid{Pure}} and \ensuremath{\Conid{Free}}
constructors.  It is probably impossible to observe a violation of the laws
using the reduced interface, but that also means that definitions by pattern
matching, like the one for our \ensuremath{\Varid{matchOpt}} in section \ref{example:option_raise},
are prohibited.

The \text{\tt free} package on
hackage\footnote{\url{http://hackage.haskell.org/package/free}} contains a
definition essentially identical to our \ensuremath{\Conid{FreeAL}}, differing only in the order of
arguments.

Another approach, which differs significantly from the one presented in the
paper, underlies the definition contained in the \text{\tt free\char45{}functors} package on
hackage,\footnote{\url{http://hackage.haskell.org/package/free-functors}} and
uses a Church-like encoding (and the \text{\tt ConstraintKinds} GHC extension) to
generalise the construction of a free \ensuremath{\Conid{Applicative}} to any superclass of
\ensuremath{\Conid{Functor}}.

The idea is to use the fact that, if a functor $T$ has a left adjoint $F$, then
the monad $T \circ F$ is the codensity monad of $T$ (i.e. the right Kan
extension of $T$ along itself).  By taking $T$ to be the forgetful functor
$\mathcal{A} \to \mathcal{F}$, one can obtain a formula for $F$ using the
expression of a right Kan extension as an end.

One problem with this approach is that the applicative laws, which make up the
definition of the category $\mathcal{A}$, are left implicit in the universal
quantification used to represent the end.

In fact, specializing the code in \text{\tt Data\char46{}Functor\char46{}HFree} to the \text{\tt Applicative}
constraint, we get:

\begin{hscode}\SaveRestoreHook
\column{B}{@{}>{\hspre}l<{\hspost}@{}}%
\column{3}{@{}>{\hspre}l<{\hspost}@{}}%
\column{5}{@{}>{\hspre}l<{\hspost}@{}}%
\column{13}{@{}>{\hspre}c<{\hspost}@{}}%
\column{13E}{@{}l@{}}%
\column{17}{@{}>{\hspre}l<{\hspost}@{}}%
\column{E}{@{}>{\hspre}l<{\hspost}@{}}%
\>[B]{}\texttt{\color{blue}\textbf{data}}\;\Conid{FreeA'}\;\Varid{f}\;\Varid{a}\mathrel{=}\Conid{FreeA'}\;\{\mskip1.5mu {}\<[E]%
\\
\>[B]{}\hsindent{3}{}\<[3]%
\>[3]{}\Varid{runFreeA}{}\<[13]%
\>[13]{}\mathbin{::}{}\<[13E]%
\>[17]{}\forall \Varid{g}\hsforall \hsdot{\circ }{\texttt{.}}\Conid{Applicative}\;\Varid{g}{}\<[E]%
\\
\>[13]{}\Rightarrow {}\<[13E]%
\>[17]{}\texttt{(}\forall \Varid{x}\hsforall \hsdot{\circ }{\texttt{.}}\Varid{f}\;\Varid{x}\to \Varid{g}\;\Varid{x}\texttt{)}\to \Varid{g}\;\Varid{a}\mskip1.5mu\}{}\<[E]%
\\[\blanklineskip]%
\>[B]{}\texttt{\color{blue}\textbf{instance}}\;\Conid{Functor}\;\Varid{f}\Rightarrow \Conid{Functor}\;\texttt{(}\Conid{FreeA'}\;\Varid{f}\texttt{)}\;\texttt{\color{blue}\textbf{where}}{}\<[E]%
\\
\>[B]{}\hsindent{3}{}\<[3]%
\>[3]{}\Varid{fmap}\;\Varid{h}\;\texttt{(}\Conid{FreeA'}\;\Varid{t}\texttt{)}\mathrel{=}\Conid{FreeA'}\;\texttt{(}\Varid{fmap}\;\Varid{h}\hsdot{\circ }{\texttt{.}}\Varid{t}\texttt{)}{}\<[E]%
\\[\blanklineskip]%
\>[B]{}\texttt{\color{blue}\textbf{instance}}\;\Conid{Functor}\;\Varid{f}\Rightarrow \Conid{Applicative}\;\texttt{(}\Conid{FreeA'}\;\Varid{f}\texttt{)}\;\texttt{\color{blue}\textbf{where}}{}\<[E]%
\\
\>[B]{}\hsindent{3}{}\<[3]%
\>[3]{}\Varid{pure}\;\Varid{x}\mathrel{=}\Conid{FreeA'}\;\texttt{(}\lambda \anonymous \to \Varid{pure}\;\Varid{x}\texttt{)}{}\<[E]%
\\
\>[B]{}\hsindent{3}{}\<[3]%
\>[3]{}\Conid{FreeA'}\;\Varid{t1}\mathbin{\texttt{<*>}}\Conid{FreeA'}\;\Varid{t2}\mathrel{=}{}\<[E]%
\\
\>[3]{}\hsindent{2}{}\<[5]%
\>[5]{}\Conid{FreeA'}\;\texttt{(}\lambda \Varid{u}\to \Varid{t1}\;\Varid{u}\mathbin{\texttt{<*>}}\Varid{t2}\;\Varid{u}\texttt{)}{}\<[E]%
\ColumnHook
\end{hscode}\resethooks

Now, for law \ref{app:id} to hold, for example, we need to prove that the term
\ensuremath{\lambda \Varid{u}\to \Varid{pure}\;\Varid{id}\mathbin{\texttt{<*>}}\Varid{t}\;\Varid{u}} is equal to \ensuremath{\Varid{t}}.  This is strictly speaking false, as
those terms can be distinguished by taking any functor with an \ensuremath{\Conid{Applicative}}
instance that does not satisfy law \ref{app:id}, and as \ensuremath{\Varid{t}} a constant function
returning a counter-example for it.

Intuitively, however, the laws should hold provided we never make use of invalid
\ensuremath{\Conid{Applicative}} instances.  To make this intuition precise, one would probably
need to extend the language with quantification over equations, and prove a
parametricity result for this extension.

Another problem of the Church encoding is that, like \'Erdi's solution above, it
presents a more limited interface, and thus it is harder to use.  In fact, the
destructor \ensuremath{\Varid{runFreeA}} is essentially equivalent to our \ensuremath{\Varid{raise}} function, which
can only be used to define \emph{applicative} natural transformation.  Again, a
function like \ensuremath{\Varid{matchOpt}}, which is not applicative, could not be defined over
\ensuremath{\Conid{FreeA'}} in a direct way.

\section{Discussion and further work}\label{sec:discussion}

We have presented a practical definition of free applicative functor over any
Haskell functor, proved its properties, and showed some of its applications.  As
the examples in this paper show, free applicative functors solve certain
problems very effectively, but their applicability is somewhat limited.

For example, applicative parsers usually need an \ensuremath{\Conid{Alternative}} instance as well,
and the free applicative construction does not provide that.  One possible
direction for future work is trying to address this issue by modifying the
construction to yield a free \ensuremath{\Conid{Alternative}} functor, instead.

Unfortunately, there is no satisfactory set of laws for alternative functors: if
we simply define an alternative functor as a monoid object in $\mathcal{A}$,
then many commonly used instances become invalid, like the one for \ensuremath{\Conid{Maybe}}.
Using rig categories and their lax functors to formalise alternative functors
seems to be a workable strategy, and we are currently exploring it.

Another direction is formalizing the proofs in this paper in a proof assistant,
by embedding the total subset of Haskell under consideration into a type theory
with dependent types.

Our attempts to replicate the proofs in Agda have failed, so far, because of
subtle issues in the interplay between parametricity and the encoding of
existentials with dependent sums.

In particular, equation \ref{naturality} is inconsistent with a representation
of the existential as a $\Sigma$ type in the definition of \ensuremath{\Conid{FreeA}}.  For
example, terms like \ensuremath{\Varid{const}\;\texttt{(}\texttt{)}\mathbin{\texttt{\color{red}{:\$:}}}\Conid{Pure}\;\mathrm{3}} and \ensuremath{\Varid{id}\mathbin{\texttt{\color{red}{:\$:}}}\Conid{Pure}\;\texttt{(}\texttt{)}} are equal by
equation \ref{naturality}, but can obviously be distinguished using large
elimination.

This is not too surprising, as we repeatedly made use of size restrictions in
sections \ref{sec:totality} and \ref{sec:semantics}, and those will definitely
need to be somehow replicated in a predicative type theory like the one
implemented by Agda.

A reasonable compromise is to develop the construction only for
\emph{containers} \cite{containers}, for which one can prove that the
free applicative on the functor $S \rhd P$ is given, using the
notation at the end of section \ref{sec:semantics}, by $S^* \rhd
(\epsilon \circ P^*)$, where $S$ is regarded as a discrete category.

Another possible further development of the results in this paper is trying to
generalise the construction of a free applicative functor to functors of any
monoidal category.  In section \ref{sec:semantics} we focused on categories with
finite products, but it is clear that monoidal categories are the most natural
setting, as evidenced by the appearance of the corresponding 2-comonad on
$\mathsf{Cat}$.

Furthermore, an applicative functor is defined in \cite{applicative} as a lax
monoidal functor \emph{with a strength}, but we completely ignore strengths in
this paper.  This could be remedied by working in the more general setting of
$\mathcal{V}$-categories and $\mathcal{V}$-functors, for some monoidal category
$\mathcal{V}$.

\section{Acknowledgements}

We would like to thank Jennifer Hackett, Thorsten Altenkirch, Venanzio Capretta,
Graham Hutton, Edsko de Vries and Christian Sattler, for helpful suggestions
and insightful discussions on the topics presented in this paper.

\bibliographystyle{eptcs}
\bibliography{b}

\end{document}